\documentclass[aos,preprint]{imsart}

\RequirePackage[OT1]{fontenc}
\RequirePackage{amsthm,amsmath}
\RequirePackage[numbers]{natbib}

\usepackage{amssymb}
\usepackage{mathrsfs}
\usepackage{algorithm}
\usepackage[noend]{algpseudocode}
\usepackage{enumitem}

\usepackage{tikz,adjustbox}
\usepackage{nameref,zref-xr}

\RequirePackage[colorlinks,citecolor=blue,urlcolor=blue]{hyperref}

\usetikzlibrary{calc}
\usetikzlibrary{decorations.markings}
\usetikzlibrary{mindmap,backgrounds}

\startlocaldefs
\newcounter{partcount}
\theoremstyle{definition}

\newtheorem{remark}{Remark}
\theoremstyle{plain}
\newtheorem{lemma}{Lemma}
\newtheorem{proposition}{Proposition}
\newtheorem{theorem}{Theorem}

\newtheorem{assumption}{Assumption}

\newcommand{\ka}{d}

\renewcommand{\ss}{\mathbb{X}}
\newcommand{\sss}{\mathcal{X}}

\newcommand{\yy}{\mathbb{Y}}
\newcommand{\yyy}{\mathcal{Y}}

\newcommand{\iid}{i.i.d.}

\newcommand{\cvarphi}{\overline{\varphi}}

\newcommand{\primeParentSet}[4]{\overline{\mathcal{A}}^{(#1,#2)}_{#3}}
\newcommand{\parentSet}[4]{\mathcal{A}^{(#1,#2)}_{#3}}
\newcommand{\tailProdSet}[4]{\widetilde{\mathcal{A}}^{(#3,#2)}_#4}

\newcommand{\clsnSetn}[3]{\mathcal{D}^{#1}_{#3}(#2)}
\newcommand{\ciSet}[3]{\left\{(u,v)\in[#3]^2:\left.\xi^{u}_{#1}\independent \xi^{v}_{#1}\mids #2\right.\right\}}
\newcommand{\ciSetGen}[4]{\left\{(u,v)\in[#4]^2:\left.\xi^{u}_{#1}\independent \xi^{v}_{#2}\mids #3\right.\right\}}
\newcommand{\ipair}[1]{\mathcal{Q}_{#1}}

\newcommand{\mrbfX}[4]{X^{(#2,#3)}_{#1}}
\newcommand{\mrbfF}[3]{\F^{(#2,#3)}_{#1}}

\newcommand{\aMatrices}[2]{(A_{k})_{k\in [#2]}}
\newcommand{\aMatricesSh}[2]{\bba}
\newcommand{\radixMatrices}[2]{\mathbb{A}^{(#1,#2)}_{\mathrm{radix}}}

\newcommand{\radixSamplePartition}[3]{\mathcal{I}^{(#1,#2,#3)}_{\mathrm{radix}}}
\newcommand{\radixSamplePartitionElement}[4]{\mathcal{I}^{(#1,#2,#3)}_{#4}}
\newcommand{\mradixMatrices}[2]{\mathbb{A}^{(#1,#2)}_{\mathrm{mixed}}}
\newcommand{\mradixSamplePartition}[3]{\mathcal{I}^{(#1,#2,#3)}_{\mathrm{mixed}}}

\newcommand{\scale}[2]{\sqrt{\frac{#1-#2}{N} + \frac{1}{\np{#1}{#2}}}}
\newcommand{\scaleb}[2]{\left(\frac{#1-#2}{N} + \frac{1}{\np{#1}{#2}}\right)^{\frac{1}{2}}}

\newcommand{\mrbfVar}[3]{\sigma^2_{\mathrm{M},#1}(#2,#3)}
\newcommand{\mrbfVarHat}[3]{\hat{\sigma}^2_{\mathrm{M},#1}(#2,#3)}
\newcommand{\frbfVar}[3]{\sigma^2_{\mathrm{R},#1}(#2,#3)}
\newcommand{\frbfVarHat}[3]{\hat{\sigma}^2_{\mathrm{R},#1}(#2,#3)}

\newcommand{\np}[2]{N_{#1,#2}}

\newcommand{\partitionMap}[1]{\mathcal{J}\left(#1\right)}
\newcommand{\partitionMapNoPar}[1]{\mathcal{J}#1}

\newcommand{\bba}{\mathbb{A}}
\newcommand{\paths}[2]{\mathcal{P}_{#1}}
\newcommand{\pathsFrom}[2]{\mathcal{P}_{#2}^{(#1)}}
\newcommand{\upperPartFrom}[6]{\mathcal{U}_{#1}(#3,#4,#5,#6)}
\newcommand{\upperPartFromNoArg}[1]{\mathcal{U}_{#1}}
\newcommand{\lowerPartFrom}[3]{\mathcal{L}_{#3}(#1,#2)}
\newcommand{\lowerPartFromNoArg}[1]{\mathcal{L}_{#1}}

\newcommand{\specCollisionPairs}[4]{D_{#3}^{(#1,#2)}}
\newcommand{\specCollisionPairsFrom}[6]{D_{#1}(#3,#4,#5,#6)}
\newcommand{\specCollisionPairsFromNoArg}[1]{D_{#1}}
\newcommand{\specNonCollisionPairsFrom}[4]{P_{#1}^{(#3,#4)}}
\newcommand{\collisionStartSet}[3]{\primeParentSet{#1}{#2}{#3}{}\setminus\primeParentSet{#1-1}{#2}{#3}{4}}
\newcommand{\collisionStartSetSh}[3]{\mathcal{R}^{(#1,#2)}_{#3}}

\newcommand{\mids}{\,\middle|\,}
\newcommand{\iidsim}{\stackrel{\mathrm{i.i.d.}}{\thicksim}}

\newcommand{\bxi}[4]{\xi^{#1}_{#2}} 

\newcommand{\indist}[1]{\xrightarrow[#1]{{\mathrm{d}}}}
\newcommand{\indistsh}[1]{\xrightarrow[\phantom{iii}]{{\mathrm{d}}}}
\newcommand{\almostsurely}[2]{\xrightarrow[#1]{\mathrm{a.s.}}}
\newcommand{\almostsurelyOneArg}[1]{\xrightarrow[#1]{\mathrm{a.s.}}}
\newcommand{\almostsurelyOneArgSh}[1]{\xrightarrow[\phantom{iii}]{\mathrm{a.s.}}}
\newcommand{\inprob}[1]{\xrightarrow[#1]{\P}}
\newcommand{\surely}[1]{\xrightarrow[#1]{}}

\newcommand{\normal}[2]{\mathcal{N}(#1,#2)}	

\renewcommand{\P}{\mathbb{P}}

\newcommand{\graph}[2]{\mathcal{G}_{#1}}

\newcommand{\vertexset}[2]{\mathcal{V}_{#1}}

\newcommand{\edgeset}[2]{\mathcal{E}_{#1}}

\newcommand{\E}{\mathbb{E}}

\newcommand{\Y}{\mathsf{Y}}
\newcommand{\X}{\mathsf{X}}

\newcommand{\calF}{\mathcal{F}}
\newcommand{\calG}{\mathcal{G}}
\newcommand{\calI}{\mathcal{I}}

\newcommand{\F}{\mathcal{F}}
\newcommand{\G}{\mathcal{G}}

\newcommand{\real}{\mathbb{R}}

\newcommand{\integers}{\mathbb{Z}}
\newcommand{\N}{\mathbb{N}}

\newcommand{\boundMeas}[1]{\mathscr{B}_{\mathrm{b}}(#1)}

\newcommand{\pmeasure}[1]{\mathscr{P}(#1)}
\newcommand{\measure}[1]{\mathscr{M}(#1)}

\newcommand{\ind}[1]{\mathbb{I}_{#1}}
\newcommand{\defeq}{:=}

\newcommand{\id}{Id}

\newcommand{\abs}[1]{\left| #1 \right|}
\newcommand{\card}[1]{\left| #1 \right|}
\newcommand{\norm}[1]{\left\|#1\right\|}

\newcommand{\floor}[1]{\left\lfloor#1\right\rfloor}
\newcommand{\ceil}[1]{\left\lceil#1\right\rceil}
\newcommand{\ones}{\mathbf{1}}

\newcommand{\osc}[1]{\mathrm{osc}\left(#1\right)}
\newcommand{\infnorm}[1]{\left\|#1\right\|_{\infty}}

\renewcommand{\iff}{\Leftrightarrow}

\newcommand{\simiid}{\stackrel{\mathrm{i.i.d.}}{\sim}}
\newcommand{\pa}[1]{\mathrm{pa}(#1)}

\newcommand\independent{\protect\mathpalette{\protect\independenT}{\perp}}
\def\independenT#1#2{\mathrel{\rlap{$#1#2$}\mkern2mu{#1#2}}}

\newcommand{\lemmaref}[1]{{Lemma \ref{#1}}}
\newcommand{\propref}[1]{{Proposition \ref{#1}}}
\newcommand{\eqnref}[1]{(\ref{#1})}
\newcommand{\secref}[1]{{Section \ref{#1}}}
\newcommand{\assref}[1]{{Assumption \ref{#1}}}
\newcommand{\algrefmy}[1]{{Algorithm \ref{#1}}}
\newcommand{\remref}[1]{{Remark \ref{#1}}}
\newcommand{\theref}[1]{{Theorem \ref{#1}}}
\newcommand{\figref}[1]{{Figure \ref{#1}}}

\setattribute{journal}{name}{}

\endlocaldefs

\input{arrowsnew.tex}

\numberwithin{equation}{partcount}
\numberwithin{section}{partcount}
\numberwithin{figure}{partcount}
\numberwithin{page}{partcount}
\renewcommand{\theequation}{\arabic{equation}}
\renewcommand{\thesection}{\arabic{section}}

\renewcommand{\thepage}{\arabic{page}}

\begin{document}

\begin{frontmatter}

\title{Butterfly resampling: asymptotics for particle filters with constrained interactions}
\runtitle{Butterfly resampling}

\thankstext{u3}{Supported by the EPSRC through First Grant EP/KO23330/1 and SuSTaIn.}
\begin{aug}
\author{\fnms{Kari} \snm{Heine}\thanksref{u3}\ead[label=e3]{kari.heine@bristol.ac.uk}},
\author{\fnms{Nick} \snm{Whiteley}\thanksref{u3}\ead[label=e4]{nick.whiteley@bristol.ac.uk}},
\author{\fnms{A.~Taylan} \snm{Cemgil}\thanksref{u1}\ead[label=e1]{taylan.cemgil@boun.edu.tr}}
\and
\author{\fnms{Hakan} \snm{G\"ulda{\c s}}\thanksref{u1}\ead[label=e2]{hakan.guldas@boun.edu.tr}}


\affiliation{University of Bristol\thanksmark{u3} and Bo{\u g}azi{\c c}i University\thanksmark{u1}}

\address{Department of Mathematics\\
University of Bristol\\
University Walk \\
Bristol \\
BS8 1TW \\
\printead{e3}\\
\phantom{E-mail:\ }\printead*{e4}}

\address{Department of Computer Engineering\\
Bo{\u g}azi{\c c}i University\\
34342 Bebek\\
Istanbul\\
\printead{e1}\\
\phantom{E-mail:\ }\printead*{e2}}

\runauthor{K.~Heine et al.}
\end{aug}

\begin{abstract}
We generalize the elementary mechanism of sampling with replacement $N$ times from a weighted population of size $N$,
by introducing auxiliary variables and constraints on conditional independence characterised by modular congruence relations. Motivated by considerations of parallelism, a convergence study reveals how sparsity of the mechanism's conditional independence graph is related to fluctuation properties of particle filters which use it for resampling, in some cases exhibiting exotic scaling behaviour. The proofs involve detailed combinatorial analysis of conditional independence graphs.
\end{abstract}



\end{frontmatter}

\section{Introduction}
\label{sec:intro}

Let $\ss$ and $\yy$ be Polish state-spaces with Borel $\sigma$-algebras $\sss$ and $\yyy$. Let $\pi_0$ be a probability measure on $\sss$ and let $f:\ss\times\sss\rightarrow[0,1]$ and $g:\ss\times\yyy\rightarrow[0,1]$ be probability kernels. A hidden Markov model is a bi-variate process $(X,Y)$ where the signal process $X=(X_n)_{n\in\N}$ is a Markov chain with initial distribution $\pi_0$ and transition kernel $f$, and the observations $Y=(Y_n)_{n\in\N}$ are conditionally independent given $X$, with the conditional distribution of $Y_n$ given $X$ being $g(X_n,\cdot)$.

Suppose that for each $x\in\ss$, $g(x,\cdot)$ admits a strictly positive density $g(x,y)$ w.r.t. a $\sigma$-finite measure. Fix a $\yy$-valued sequence $(y_n)_{n\in\N}$ and define the operators $(\Phi_n)_{n\geq1}$  acting on probability measures,
\begin{equation}\Phi_n(\mu)(A) := \frac{\int_\ss g(x,y_{n-1})f(x,A)\mu(dx)}{\int_\ss g(x,y_{n-1})\mu(dx)},\quad A\in\sss.\label{intro:pred_filt1}
\end{equation}
Consider $\pi_n:=\Phi_n(\pi_{n-1})$, $n\geq1$. If one replaces $(y_n)_{n\in\N}$ in \eqref{intro:pred_filt1} with the random variables $(Y_n)_{n\in\N}$ then $\pi_n$ is a version of the regular conditional distribution of $X_n$ given $Y_0,\ldots,Y_{n-1}$. Particle filters \citep{gordon1993novel} approximate $(\pi_n)_{n\in\mathbb{N}}$ by sampling  $(\zeta_0^i)_{i=1}^N\iidsim\pi_0$, and for $n\geq1$,
\begin{equation}
(\hat{\zeta}^i_{n-1})_{i=1}^N \simiid  \dfrac{\sum_i g(\zeta_{n-1}^i,y_{n-1})\delta_{\zeta_{n-1}^i}}{\sum_i g(\zeta_{n-1}^i,y_{n-1})},\quad \zeta_n^i\sim f(\hat{\zeta}^i_{n-1},\cdot),\;i=1,\ldots,N,\label{eq:pf_intro}
\end{equation}
so in effect $(\zeta_n^i)_{i=1}^N \iidsim \Phi_n(\pi^N_{n-1})$, where $\pi_{n-1}^N:=N^{-1}\sum_i\delta_{\zeta_{n-1}^i}$.
This remarkably simple mechanism has found a huge number of applications. Under mild assumptions -- it suffices that for each $n$, $g(x,y_n)$ is bounded in $x$ -- a law of large numbers and central limit theorem hold \citep{del1999central,smc:the:C04,smc:the:K05,smc:the:DM08}; for $\real$-valued, bounded functions $\varphi$,
\begin{equation}
\pi_n^N (\varphi)\almostsurelyOneArg{N\rightarrow\infty}\pi_n(\varphi),\quad\sqrt{N} \left(\pi_n^N (\varphi)-\pi_n(\varphi)\right) \indist{N\rightarrow\infty} \normal{0}{\sigma_n^2(\varphi)},\label{intro:CLT}
\end{equation}
where for a measure $\mu$, $\mu(\varphi):=\int\varphi(x)\mu(dx)$. The asymptotic fluctuations of the particle approximation error are thus of order $1/\sqrt{N}$, as they would be if $(\zeta_n^i)_{i=1}^N\iidsim \pi_n$,  and it can be shown that $\sigma_n^2(\varphi)$ is never less than the asymptotic variance which would arise from such \iid~samples.

\subsection{Conditional independence and convergence}
The conditional independence and sampling with replacement, or \emph{resampling}, in \eqref{eq:pf_intro} leads to the $\sqrt{N}$ scaling in \eqref{intro:CLT}. This dependence structure also influences how particle filters are implemented and resampling hinders their parallelization \cite{lee2010utility}. Our contribution is to lay rigorous foundations for the design of algorithms better suited to modern computing architectures. We provide insight into consequences for convergence of imposing constraints on the conditional independence structure of a particle filter as a proxy for its communication pattern -- an important factor in efficiency of parallel and distributed algorithms \citep{Bertsekas:1997}. As a taster: for some new algorithms we establish results of the general form
\begin{equation*}
s(N,r) \left(\pi_n^N (\varphi)-\pi_n(\varphi)\right) \indist{N\rightarrow\infty} \normal{0}{\sigma_n^2(\varphi,r)},
\end{equation*}
where $s(N,r)$ is some increasing function of $N$ possibly other than $\sqrt{N}$, and $r$ is a parameter related to the sparsity of the algorithm's conditional independence graph. We shall investigate  the relationship between $r$, $s(N,r)$ and $\sigma_n^2(\varphi,r)$.

\subsection{Outline}
In \secref{sec:algorithms and main results} we introduce a new \emph{augmented resampling} algorithm, which generalizes the i.i.d.~sampling part of \eqref{eq:pf_intro}. We construct two instances of this algorithm, which we call \emph{butterfly resampling}, since their conditional independence graphs have the butterfly pattern well known from the Cooley-Tukey fast Fourier transform, but which is also a standard network topology in parallel computing \citep{savage:1998}. The butterfly structure stems from equivalence classes of conditionally \iid ~samples in our algorithms, characterized by \emph{modular congruence relations}, i.e. equivalence relations expressed in terms of modular arithmetic. In turn this demands that we develop some non-standard tools for studying convergence.
\begin{itemize}
    \item For the first butterfly algorithm, $s(N,r)=\sqrt{N/\log_r N}$. This exotic scaling is the price to pay for the number of incoming edges per vertex in its conditional independence graph being $r$ and the total number of edges being  $r N \log_r N$, versus respectively $N$ and $N^2$ for a standard particle filter.
    \item To achieve a more even balance between fluctuations and interaction constraints, we devise a second butterfly algorithm for which $s(N,r)=\sqrt{N}$, with an asymptotic variance upper bounded by $(2-r^{-1})\sigma_n^2(\varphi)$ where $\sigma_n^2(\varphi)$ is as in \eqref{intro:CLT}. For this algorithm some vertices have $r$ incoming edges, no vertex has greater than $N/r$ incoming edges and the total number of edges is $rN + N^2/r$.
\end{itemize}
Proofs and supporting results are in \secref{sec:convergence analysis part 1} onwards, prefaced by a guide for the reader to aid navigation of our analysis. Two key ingredients that are not usually encountered in theoretical accounts of particle filters are:
\begin{itemize}
\item we establish error bounds for certain sub-populations of the particle system, subsequently put to use in establishing limit theorems,
\item we conduct a detailed combinatorial analysis of conditional independence graphs, overcoming the biggest technical challenge in analysis of the second moment properties of butterfly sampling, which differ from those of standard particle filters.
\end{itemize}
The more technical results and most proofs are in the \ref{suppA}.

\subsection{Notation and conventions}


For all $x,y\in \real$, such that $y\neq 0$, we define $\floor{x} \defeq \max\{i \in \integers : i \leq x\}$, $\ceil{x} \defeq \min\{i \in \integers : i \geq x\}$ and $x \bmod y \defeq x - y\floor{x/y}$. For all $n\in\N$, we write $[n] \defeq \{1,\ldots,n\}$.Whenever a summation symbol $\Sigma$ appears without the summation set made explicit, the summation set is taken to be
$[N]$, for example we write $\Sigma_i$ for $\Sigma_{i=1}^N$. Also ~$\sum_{(i_0,\ldots,i_k)}$ is short for $\sum_{i_0}\cdots\sum_{i_k}$.

For a sequence $(M_k)_{k=1}^m$ of square matrices $\prod_{k=1}^m M_k \defeq M_1\cdots M_m$. Also the shorthand notations $M_{p:q} \defeq \prod_{k=p}^q M_k$, where $p\leq q$, and $M_{p:q} \defeq \prod_{k=0}^{p-q} M_{p-k}$, where $p\geq q$, will occasionally be used.
The symbol $\otimes$ denotes: Kronecker product for matrices, direct product for measures, and tensor product for functions. The interpretation will always be clear from the context.
For $n\in\N$, $I_n$ denotes the $n\times n$ identity matrix and $\ones_{1/n}$ denotes the $n\times n$ matrix which has $1/n$ as every entry. The notation $\id$ will be used for identity mappings in various contexts.


We denote by $\measure{\ss}$, $\pmeasure{\ss}$ and $\boundMeas{\ss}$ respectively the collections measures, probability measures and of $\real$-valued, measurable and bounded functions on $(\ss,\sss)$. For $\mu\in\measure{\ss}$, $\varphi\in\boundMeas{\ss}$, $A\in\sss$ and an integral kernel $K:\ss\times\sss\rightarrow\real_+$ we write $K(\varphi)(x):=\int K(x,dx^\prime)\varphi(x^\prime)$, $(\mu K)(A):=\int K(x,A)\mu(dx)$. For $\varphi\in\boundMeas{\ss}$, define $\infnorm{\varphi}:=\sup_{x\in\ss}\left|\varphi(x)\right|$ and $\osc{\varphi}:=\sup_{x,y\in\ss}\left|\varphi(x)-\varphi(y)\right|$. We assume an underlying probability space $(\Omega,\F,\P)$ on which all the random variables we encounter are defined. Convergence in probability under $\P$ is denoted by $\inprob{}$. For random variables $X,Y,Z$ we write $X\independent Y \;|\;Z$ to mean $X$ and $Y$ are conditionally independent given $Z$.

\section{Algorithms and main results}
\label{sec:algorithms and main results}

\subsection{Basics of particle filtering}

Since we consider a fixed observation sequence $(y_n)_{n\in\mathbb{T}}$, we shall write $g_n(x):=g(x,y_n)$. The following mild regularity condition is assumed to hold throughout this paper.
\begin{assumption}\label{ass:g_n_bounded_and_positive}
For each $n\in\N$, $\sup_x g_n(x)<\infty$ and $g_n(x)>0$, $\forall x$.
\end{assumption}
Algorithm \ref{alg:pf} is a basic particle filter. There are a number of ways to perform the \textsc{resample} operation. The multinomial method is:
\begin{equation}\label{eq:multinomial_resampling}
(\hat{\zeta}^i_n)_{i\in[N]} \simiid  \dfrac{\sum_i g_n(\zeta_n^i)\delta_{\zeta_n^i}}{\sum_i g_n(\zeta_n^i)},
\end{equation}
and in that case Algorithm \ref{alg:pf} is known as the Bootstrap Particle Filter (BPF).
\begin{algorithm}[!ht]
\caption{Particle filter}\label{alg:pf}
\begin{algorithmic}
\State
\For{$i=1,\ldots,N$}
	\State sample $\zeta_0^i \sim \pi_0$
\EndFor
\State set $(\hat{\zeta}^i_0)_{i\in[N]} \leftarrow \textsc{resample}\left((\zeta_0^i)_{i\in[N]},g_0\right)$
\State
\For{$n=1,2\ldots$}
	\For{$i=1,\ldots,N$}
		\State sample $\zeta_n^i \sim f(\hat{\zeta}_{n-1}^i,\cdot)$
	\EndFor
\State set $(\hat{\zeta}^i_n)_{i\in[N]} \leftarrow \textsc{resample}\left((\zeta_n^i)_{i\in[N]},g_n\right)$
\EndFor
\end{algorithmic}
\end{algorithm}

The following formulae are well defined and finite for $\varphi\in\boundMeas{\ss}$,
\begin{equation}\label{eq:asymp_var_bpf}
\arraycolsep=1.4pt
\begin{array}{rll}
\sigma_0^2(\varphi)&:=\pi_0((\varphi-\pi_0(\varphi))^2),&\\[.1cm]
\sigma_n^2(\varphi)&:=\hat{\sigma}_{n-1}^2(f(\varphi))+\hat{\pi}_{n-1}(f((\varphi-f(\varphi))^2)),&~n\geq 1,\\[.1cm]
\hat{\sigma}_n^2(\varphi)&:=\hat{\pi}_n((\varphi-\hat{\pi}_n(\varphi))^2)+\pi_n(g_n)^{-2}\sigma_n^2(g_n(\varphi-\hat{\pi}_n(\varphi))),&~n\geq0,
\end{array}
\end{equation}
where $\hat{\pi}_n(\varphi)=\pi_n(g_n\varphi)/\pi_n(g_n)$. Considering the empirical measures $\pi_n^N=N^{-1}\sum_i \delta_{\zeta_n^i}$ and $\hat{\pi}_n^N=N^{-1}\sum_i \delta_{\hat{\zeta}_n^i}$,
a direct application of e.g. the results of \cite{smc:the:C04} (assuming for the convergence in distribution that the quantities in \eqref{eq:asymp_var_bpf} are strictly positive) gives:
\begin{theorem}\label{thm:bootstrap}
For any $n\geq0$ and $\varphi\in\boundMeas{\ss}$, the BPF has the properties that
\begin{align*}
\pi_n^N(\varphi)-\pi_n(\varphi)\almostsurelyOneArg{N\to\infty}0,\quad&\sqrt{N}\left(\pi_n^N(\varphi)-\pi_n(\varphi)\right)\indist{N\rightarrow\infty}\normal{0}{\sigma_n^2(\varphi)},\\
\hat{\pi}_n^N(\varphi)-\hat{\pi}_n(\varphi)\almostsurelyOneArg{N\to\infty}0,\quad&\sqrt{N}\left(\hat{\pi}_n^N(\varphi)-\hat{\pi}_n(\varphi)\right)\indist{N\rightarrow\infty}\normal{0}{\hat{\sigma}_n^2(\varphi)}.
\end{align*}
\end{theorem}
This result will serve as a point of reference against which to compare convergence properties of our new algorithms. Various refinements and extensions of \theref{thm:bootstrap} exist \cite{del1999central,smc:the:K05,smc:the:DM08}, but to emphasize the novel aspects of our comparisons we eschew some technical generalities, many of our results can be generalized to larger function classes and settings beyond HMM's, and the structure of our algorithms can also be generalized without difficulty so as to incorporate other proposal and resampling schemes.

\subsection{Considerations of parallelism and the motivation for our approach}\label{sub:parallelism}
 It is standard practice in computer science to reason about parallelism by introducing a graphical computation/communication model which captures some essence of a practical architecture \cite[Ch. 7]{savage:1998}, \citep{leighton1992}. We adopt this philosophy. It is not the purpose of this paper to discuss implementation-specific details of programming etc.


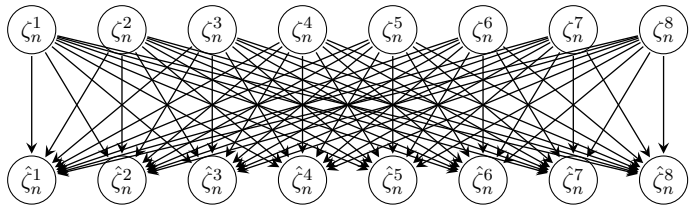
\begin{figure}[t!]
\begin{center}
\tikzset{
    vertex/.style = {
    	draw,
	    circle,
        fill      = white,
        outer sep = 2pt,
        inner sep = 2pt,
    }
}

\tikzset{
    mainvertex/.style = {
    	draw,
		circle,
        fill = black,
        outer sep = 2pt,
        inner sep = 2pt,
    }
}

\tikzset{
    demovertex/.style = {
    	draw,
		circle,
        fill = gray,
        outer sep = 2pt,
        inner sep = 2pt,
    }
}

\tikzset{
    basicvertex/.style = {
    	draw,
		circle,
		minimum size = .8cm
    }
}

\tikzset{
    demoedge/.style = {
		-stealthnew,
		shorten <=.15cm,
		shorten >=.15cm,
		arrowhead=2mm,
		line width=2pt
		}
}

\tikzset{
    basicedge/.style = {
		-stealthnew,
		color = black,
		shorten <=.35cm,
		shorten >=.35cm,
		arrowhead=1.5mm,
		line width=.5pt
		}
}

\begin{tikzpicture}[scale=0.8]
\tikzstyle{every node}=[scale=.8,font=\normalsize]

\def \vstep {-2.5cm}
\def \hstep {1.5cm}
\def \r {8}
\def \m {3} 
\def \N {8}
\def \X {4}
\def \Y {2}

\foreach \i in {0,...,1}
{
	\foreach \j in {1,...,8}
	{
  		\node[basicvertex] (\i\j) at (\j*\hstep,\i*\vstep) {};
	}
}	

\foreach \j in {1,...,8}
{
	\node[circle] at (\j*\hstep,0*\vstep) {$\zeta^{\j}_{n}$};
}

\foreach \j in {1,...,8}
{
	\node[circle] at (\j*\hstep,\vstep) {$\hat{\zeta}^{\j}_{n}$};
}

\foreach \i in {0}
{
	\foreach \j in {1,...,8}
	{	
		\foreach \k in {0,...,7}
		{
			\draw[basicedge]
			let
				\n1 = {int(mod((\j-1),\r^(\i))+\k*\r^(\i)+\r^(\i+1)*int(floor((\j-1)/\r^(\i+1))))+1},
				\n2 = {int(\i+1)}
			in
				(\j*\hstep,\i*\vstep) -- (\n1*\hstep,\n2*\vstep);
		}
	}
}	

\end{tikzpicture}
\end{center}
\caption{Conditional independence structure of multinomial resampling. $(\hat{\zeta}_n^i)_{i=1}^N$ are conditionally i.i.d.~draws from the distribution proportional to $\sum_i g_n(\zeta_n^i)\delta_{\zeta_n ^i}$.}
\label{fig:example of multinomial with zetas}
\end{figure}

Key to efficiency is an algorithm's \emph{communication pattern} -- the structure via which computational elements exchange information \cite{Bertsekas:1997}. The bottleneck in this regard for particle filters is the resampling operation, and its conditional independence graph, henceforth ``graph'' for brevity, provides a convenient and very simple model for its communication pattern if we associate each vertex in the graph with a separate processing unit and each edge with a communication link. Figure \ref{fig:example of multinomial with zetas} shows the graph for multinomial resampling \eqref{eq:multinomial_resampling}; one can think of each $\zeta_n^i$ and its weight $g_n(\zeta_n^i)$ as being stored locally at the $i$th vertex in the top row, and the $i$th vertex in the bottom row being tasked with sampling $\hat{\zeta}_n^i$. To achieve full parallelism, one would need $O(N^2)$ separate communication paths, ideally a separate physical connection corresponding to each edge in the graph. In practice, communication will be achieved through shared memory or a common data bus, inevitably leading to extensive memory traffic and delays as processors synchronize.

Our interest therefore turns to algorithms with more sparse graphs and -- again as is standard in parallel computing \citep[Ch. 3]{leighton1992} -- we can quantitatively summarize sparsity in terms of the total number of edges in the graph and the number of incoming edges per vertex, respectively $N^2$ and $N$ for multinomial resampling. Our aim is to explore the mathematical connections between these quantities and convergence properties as per \theref{thm:bootstrap}. Moreover, the graphs for the butterfly algorithms we devise match the structure of butterfly networks -- well known communication topologies in parallel computing \cite[Ch. 7]{savage:1998}.

\subsection{Literature}
There is a small but growing literature on theoretical analysis of particle algorithms with parallelism. The algorithms of \cite{Verge_island_particle} involve resampling at two hierarchical levels, and are presented with a study of asymptotic bias and variance. A recent preprint \cite{Verge_island_particle_conv} gives a central limit theorem. Some authors of the present paper \cite{whiteley_lee_heine_2014,lee_whiteley_2014} have studied the non-asymptotic stability properties of an ``$\alpha$SMC'' algorithm in which interaction between particles occurs adaptively, so as to keep the effective sample size above a given threshold. Despite some superficial similarities, the butterfly algorithms we devise are distinct from $\alpha$SMC in a number of ways, they do not involve any adaptation, their butterfly structure is entirely original and our study is focused on asymptotics. Some comments on stability are given in \secref{sec:discussion}, \remref{rem:stability}. Various issues of computational efficiency for standard algorithms are addressed by e.g. \citep{Murray_resampling_2014} and references therein.

\subsection{Augmented resampling}
\label{sec:augmented resampling}

We now introduce a new and general procedure called augmented resampling, which involves the following parameters:
\begin{itemize}
  \item $N$, the population size, as in Algorithm \ref{alg:pf}
  \item $m$, a positive integer
  \item $\aMatrices{N}{m}$, a sequence of non-negative matrices, each of size $N\times N$
\end{itemize}
The main idea is that we can use the matrices $\aMatrices{N}{m}$ to impose constraints on conditional independence of the random variables $\{\xi^{i}_k: i\in[N],~0\leq k \leq m\}$ in Algorithm \ref{alg:augmented_resampling}, the sampling steps of which are well-defined if $g$ is a member of $\boundMeas{\ss}$ and is strictly positive.
\begin{algorithm}[!ht]
\caption{Augmented resampling}
\label{alg:augmented_resampling}
\begin{algorithmic}
\label{alg:bf}
\Procedure{$(\xi_{\mathrm{out}}^i)_{i\in[N]}=\text{resample}\left((\xi_{\mathrm{in}}^i)_{i\in[N]},g\right)$}{}
\For{$i=1,\ldots,N$}
	\State $\xi^i_0 \leftarrow \xi_{\text{in}}^{i}$
	\State $V_{0}^{i} \leftarrow g(\xi_{0}^{i})$
\EndFor
\For{$k=1,\ldots,m$}
	\For{$i=1,\ldots,N$}
		\State set $V^i_k \leftarrow \sum_{j}A^{ij}_{k}V^j_{k-1}$
		\State sample $\xi^i_k\sim (V^i_k)^{-1}\sum_{j}A^{ij}_{k}V^j_{k-1}\delta_{\xi^j_{k-1}}$
	\EndFor
\EndFor
\For{$i=1,\ldots,N$}
	\State $\xi_{\text{out}}^{i} \leftarrow \xi_{m}^{i}$
\EndFor
\EndProcedure
\end{algorithmic}
\end{algorithm}


As a special case, consider $m=1$ and let $A_1=\mathbf{1}_{1/N}$.  Algorithm \ref{alg:augmented_resampling}  then delivers, by inspection,
\[
\xi_{\text{out}}^i = \xi_1^i \sim \dfrac{1/N\sum_j  V_0^j \delta_{\xi_0^j}}{1/N\sum_j V_0^j }
= \dfrac{\sum_j g(\xi_{\text{in}}^j) \delta_{\xi_{\text{in}}^j}}{ \sum_j g(\xi_{\text{in}}^j) },\quad i\in[N],
\]
thus augmented resampling generalizes the multinomial resampling scheme \eqref{eq:multinomial_resampling}. With $m\geq1$ it turns out that a fruitful approach is to consider certain $m$-fold factorizations of $\mathbf{1}_{1/N}$ embodied by the following assumption.
\begin{assumption}\label{ass:A_k}
For all $k\in[m]$, $A_k$ is a doubly-stochastic matrix and $\prod_{k=1}^m A_k=\mathbf{1}_{1/N}$.
\end{assumption}
Under this assumption, we can establish some simple but fundamental lack-of-bias and moment properties of augmented resampling. The proof of the following proposition is in \secref{subsec:martingale}.
\begin{proposition}\label{prop:intro_to_aug_resampling}
Fix $N\geq1$, and consider Algorithm \ref{alg:augmented_resampling} with  $g\in\boundMeas{\ss}$ such that $g(x)>0$ for all $x\in\ss$ . Fix $m\geq1$ and suppose that $\aMatrices{N}{m}$ satisfy \assref{ass:A_k}. Then for any $\varphi\in\boundMeas{\ss}$,
\begin{equation}
\E\left[\left.\frac{1}{N}\sum_i \varphi(\xi_{\mathrm{out}}^i)\right|\left(\xi_{\mathrm{in}}^i\right)_{i\in[N]}\right]=\dfrac{\sum_i g(\xi_{\mathrm{in}}^i)\varphi(\xi_{\mathrm{in}}^i)}{\sum_i g(\xi_{\mathrm{in}}^i)},\label{eq:aug_res_lack_of_bias}
\end{equation}
and for any $p\geq1$ there exists a finite constant $b_p$, depending only on $p$, such that no matter what the distribution of $\left(\xi_{\mathrm{in}}^i\right)_{i\in[N]}$ is,
\begin{eqnarray}
&&\E\Bigg[\Bigg|\Bigg(\frac{1}{N}\sum_i g(\xi_{\mathrm{in}}^i)\Bigg)\Bigg(\frac{1}{N}\sum_i \varphi(\xi_{\mathrm{out}}^i)\Bigg)-\frac{1}{N}\sum_i g(\xi_{\mathrm{in}}^i)\varphi(\xi_{\mathrm{in}}^i)\Bigg|^{p}\Bigg] \nonumber\\
&&\qquad \leq b_p \bigg(\frac{m}{N}\bigg)^{\frac{p}{2}} \infnorm{g}^p\osc{\varphi}^p. \label{eq:aug_res_L_p_bound}
\end{eqnarray}
\end{proposition}
It is of course implicit in the notation here that $m$ and the matrices $\aMatrices{N}{m}$ may depend on $N$. An immediate consequence of \eqref{eq:aug_res_L_p_bound} is that if, for example, $m$ is some non-decreasing function of $N$,  $\aMatrices{N}{m}$ satisfy \assref{ass:A_k} for every $N$, and $\sum_{N=1}^\infty(m/N)^{p/2}<\infty$ for some $p\geq1$, then
$$\Bigg(\frac{1}{N}\sum_i g(\xi_{\mathrm{in}}^i)\Bigg)\Bigg(\frac{1}{N}\sum_i \varphi(\xi_{\mathrm{out}}^i)-\frac{\sum_i g(\xi_{\mathrm{in}}^i)\varphi(\xi_{\mathrm{in}}^i)}{\sum_i g(\xi_{\mathrm{in}}^i)}\Bigg) \;\almostsurelyOneArg{N\rightarrow\infty}\; 0,$$
without requiring any convergence of $N^{-1} \sum_i g(\xi_{\mathrm{in}}^i)$ or $N^{-1} \sum_i g(\xi_{\mathrm{in}}^i)\varphi(\xi_{\mathrm{in}}^i)$. However even if these quantities do converge, without further assumption there is no guarantee of a corresponding central limit theorem and more structure is needed to establish non-trivial limits for the moments in \eqref{eq:aug_res_L_p_bound} when suitably rescaled. We next introduce parameterised families of the matrices $\aMatrices{N}{m}$ which give rise to this structure and which are pursuant to the aims described in \secref{sub:parallelism}.


%
%

\subsection{Radix-\texorpdfstring{$r$}{r} resampling algorithm}
\label{sec:fixed radix main results}

For each $r\geq 2$ and $m\geq 1$, consider the family of matrices
\begin{equation}\label{eq:def fixed radix matrices}
\radixMatrices{r}{m}:=(A_k)_{k\in[m]},\quad A_k = I_{r^{m-k}} \otimes  \ones_{1/r} \otimes I_{r^{k-1}}, \quad k\in[m].
\end{equation}
We shall refer to  \algrefmy{alg:augmented_resampling}  applied with the matrices in \eqref{eq:def fixed radix matrices} and $N=r^m$ as the \emph{radix-\texorpdfstring{$r$}{r} butterfly resampling algorithm}. Examples of the matrices in  \eqref{eq:def fixed radix matrices} are shown in \figref{fig:radix matrices}.

\begin{figure}
\begin{tabular}{ccc}
$A_1$&$A_2$&$A_3$ \\[-2mm]
\mbox{
\adjustbox{scale=.8}{
\begin{minipage}{.32\textwidth}
\begin{equation*}
\!\!\!\!\!\!\left(
\arraycolsep=4pt
\begin{array}{cccc|cccc}
\frac{1}{2} & \frac{1}{2} & \cdot & \cdot & \cdot & \cdot & \cdot & \cdot \rule[-1.55mm]{0pt}{4.7mm} \\
\frac{1}{2} & \frac{1}{2} & \cdot & \cdot & \cdot & \cdot & \cdot & \cdot \rule[-1.55mm]{0pt}{4.7mm} \\
\cdot & \cdot & \frac{1}{2} & \frac{1}{2} & \cdot & \cdot & \cdot & \cdot \rule[-1.55mm]{0pt}{4.7mm} \\
\cdot & \cdot & \frac{1}{2} & \frac{1}{2} & \cdot & \cdot & \cdot & \cdot \rule[-1.55mm]{0pt}{4.7mm} \\ \hline
\cdot & \cdot & \cdot & \cdot & \frac{1}{2} & \frac{1}{2} & \cdot & \cdot \rule[-1.55mm]{0pt}{4.7mm} \\
\cdot & \cdot & \cdot & \cdot & \frac{1}{2} & \frac{1}{2} & \cdot & \cdot \rule[-1.55mm]{0pt}{4.7mm} \\
\cdot & \cdot & \cdot & \cdot & \cdot & \cdot & \frac{1}{2} & \frac{1}{2} \rule[-1.55mm]{0pt}{4.7mm} \\
\cdot & \cdot & \cdot & \cdot & \cdot & \cdot & \frac{1}{2} & \frac{1}{2} \rule[-1.55mm]{0pt}{4.7mm} \\
\end{array}
\right)
\end{equation*}
\end{minipage}}} &
\mbox{
\adjustbox{scale=.8}{
\begin{minipage}{.32\textwidth}
\begin{equation*}
\!\!\!\!\!\!\left(
\arraycolsep=4pt
\begin{array}{cccc|cccc}
 \frac{1}{2}& \cdot & \frac{1}{2}& \cdot & \cdot & \cdot & \cdot & \cdot  \rule[-1.55mm]{0pt}{4.7mm} \\
 \cdot & \frac{1}{2}& \cdot & \frac{1}{2}& \cdot & \cdot & \cdot & \cdot  \rule[-1.55mm]{0pt}{4.7mm} \\
 \frac{1}{2}& \cdot & \frac{1}{2}& \cdot & \cdot & \cdot & \cdot & \cdot  \rule[-1.55mm]{0pt}{4.7mm} \\
 \cdot & \frac{1}{2}& \cdot & \frac{1}{2}& \cdot & \cdot & \cdot & \cdot  \rule[-1.55mm]{0pt}{4.7mm} \\ \hline
 \cdot & \cdot & \cdot & \cdot & \frac{1}{2}& \cdot & \frac{1}{2}& \cdot  \rule[-1.55mm]{0pt}{4.7mm} \\
 \cdot & \cdot & \cdot & \cdot & \cdot & \frac{1}{2}& \cdot & \frac{1}{2} \rule[-1.55mm]{0pt}{4.7mm} \\
 \cdot & \cdot & \cdot & \cdot & \frac{1}{2}& \cdot & \frac{1}{2}& \cdot  \rule[-1.55mm]{0pt}{4.7mm} \\
 \cdot & \cdot & \cdot & \cdot & \cdot & \frac{1}{2}& \cdot & \frac{1}{2} \rule[-1.55mm]{0pt}{4.7mm} \\
\end{array}
\right)
\end{equation*}
\end{minipage}}} &
\mbox{
\adjustbox{scale=.8}{
\begin{minipage}{.32\textwidth}
\begin{equation*}
\!\!\!\!\!\!\left(
\arraycolsep=4pt
\begin{array}{cccc|cccc}
 \frac{1}{2}& \cdot & \cdot & \cdot & \frac{1}{2}& \cdot & \cdot & \cdot  \rule[-1.55mm]{0pt}{4.7mm} \\
 \cdot & \frac{1}{2}& \cdot & \cdot & \cdot & \frac{1}{2}& \cdot & \cdot  \rule[-1.55mm]{0pt}{4.7mm} \\
 \cdot & \cdot & \frac{1}{2}& \cdot & \cdot & \cdot & \frac{1}{2}& \cdot  \rule[-1.55mm]{0pt}{4.7mm} \\
 \cdot & \cdot & \cdot & \frac{1}{2}& \cdot & \cdot & \cdot & \frac{1}{2} \rule[-1.55mm]{0pt}{4.7mm} \\ \hline
 \frac{1}{2}& \cdot & \cdot & \cdot & \frac{1}{2}& \cdot & \cdot & \cdot  \rule[-1.55mm]{0pt}{4.7mm} \\
 \cdot & \frac{1}{2}& \cdot & \cdot & \cdot & \frac{1}{2}& \cdot & \cdot  \rule[-1.55mm]{0pt}{4.7mm} \\
 \cdot & \cdot & \frac{1}{2}& \cdot & \cdot & \cdot & \frac{1}{2}& \cdot  \rule[-1.55mm]{0pt}{4.7mm} \\
 \cdot & \cdot & \cdot & \frac{1}{2}& \cdot & \cdot & \cdot & \frac{1}{2} \rule[-1.55mm]{0pt}{4.7mm} \\
\end{array}
\right)
\end{equation*}
\end{minipage}}}
\end{tabular}
\caption{The matrices $\radixMatrices{2}{3}=(A_1,A_2,A_3)$ for the radix-$2$ algorithm.}
\label{fig:radix matrices}
\end{figure}
The algebraic structure of \eqref{eq:def fixed radix matrices} dictates the conditional independence structure of butterfly resampling. As a step towards illustrating this connection we now derive a modular congruence characterization of the non-zero matrix entries. For each $k\in[m]$ and $r\geq2$ introduce the following congruence relation on $[N]$,
\begin{equation*}
i \stackrel{(k,r)}{\sim} j \quad \Longleftrightarrow \quad\begin{cases} \floor{\dfrac{i-1}{r^k}} = \floor{\dfrac{j-1}{r^k}}, & \\
\quad \quad \quad\text{and} \phantom{\Big|}& \\
(i-1)\bmod r^{k-1}=(j-1)\bmod r^{k-1}. &
\end{cases}
\end{equation*}

\begin{lemma} \label{lem:modular congruence realtion fixed}
The matrices in \eqref{eq:def fixed radix matrices} satisfy Assumption \ref{ass:A_k}. Moreover they are symmetric, have entries which are either $1/r$ or zero, and the non-zero entries are characterized by:
$$
A_{k}^{ij}>0 \quad\Longleftrightarrow \quad i \stackrel{(k,r)}{\sim} j.
$$
\end{lemma}
Since the matrices in \eqref{eq:def fixed radix matrices} are a key and novel ingredient in our algorithms, we present the proof of the lemma before discussing its interpretation.
\begin{proof}
First we recall the \emph{mixed product property} of Kronecker product, that is, for any matrices $A$,$B$,$C$ and $D$, such that the products $AC$ and $BD$ are defined, one has (see, e.g.~\cite{horn_et_johnson})
\begin{equation}\label{eq:mixed product property}
(A \otimes B)(C\otimes D)=(AC)\otimes (BD).
\end{equation}
Also we note that for any two square matrices $A$ of size $M$ and $B$ of size $N$, the Kronecker product has the element-wise formula:
\begin{equation}\label{eq:elementwise Kronecker formula}
(A\otimes B)^{ij}=A^{\floor{\frac{i-1}{N}}+1,\floor{\frac{j-1}{N}}+1} B^{((i-1)\bmod N)+1,((j-1)\bmod N)+1},
\end{equation}
where $i,j\in[MN]$. From the element-wise formula we see immediately that $A\otimes B$ is symmetric if $A$ and $B$ are symmetric. Hence, by \eqnref{eq:def fixed radix matrices}, $A_k$ is symmetric for all $k\in[m]$. By applying \eqnref{eq:mixed product property} twice to the definition of $A_k$, one has $A_kA_k=A_k$, i.e.~$A_k$ is idempotent. By the associativity of the Kronecker product and two applications of the element-wise formula, we also have for the matrices in \eqref{eq:def fixed radix matrices} the expression
\begin{align}
A_k^{ij} &= I_{r^{m-k}}^{\floor{\frac{i-1}{r^k}}+1,\floor{\frac{j-1}{r^k}}+1}\ones_{1/r}^{\left(\floor{\frac{i-1}{r^{k-1}}}\bmod r\right)+1, \left(\floor{\frac{j-1}{r^{k-1}}}\bmod r\right)+1} \nonumber\\
&\times I_{r^{k-1}}^{\left((i-1)\bmod r^{k-1}\right)+1, \left((j-1)\bmod r^{k-1}\right)+1},	\label{eq:elementwise formula for radix}
\end{align}
where we have also used the fact that $\lfloor\lfloor(i-1)/r^{k-1}\rfloor / r\rfloor = \lfloor(i-1)/r^{k}\rfloor$, and $\lfloor\lfloor(j-1)/r^{k-1}\rfloor / r\rfloor = \lfloor(j-1)/r^{k}\rfloor$.
From this we see immediately that $A_{k}^{ij}\in\{0,1/r\}$.

By the idempotence, symmetry and the facts that by \eqnref{eq:elementwise formula for radix}, $A_{k}^{ii}=1/r$ and $A_{k}^{ij}\in\{0,1/r\}$ one has
\begin{equation*}
\frac{1}{r} = A_k^{ii} = (A_kA_k)^{ii} = (A_k^TA_k)^{ii} = \sum_{j\in[r^m]} (A^{ij}_k)^2 = \frac{p}{r^2} \iff p = r,
\end{equation*}
where $p$ is the number of non-zero elements on the $i$th column of $A_k$. Hence the double stochasticity of \assref{ass:A_k} follows by symmetry.

To prove the remaining part of \assref{ass:A_k}, we assume that for some $k>1$, $\prod_{q=1}^{k-1}A_{q} = I_{r^{m-k+1}}\otimes \ones_{1/r^{k-1}}$. By \eqnref{eq:def fixed radix matrices}, this clearly holds for $k=2$. Then by the associativity and the mixed product property \eqnref{eq:mixed product property}
\begin{eqnarray*}
\textstyle\prod_{q=1}^{k}A_{q}
&=& \textstyle\big(\prod_{q=1}^{k-1}A_{q}\big)A_{k}\\
&=& \big(I_{r^{m-k+1}}\otimes \ones_{1/r^{k-1}}\big)\big(I_{r^{m-k}} \otimes  \ones_{1/r} \otimes I_{r^{k-1}}\big)\\
&=& \big(I_{r^{m-k+1}}(I_{r^{m-k}} \otimes  \ones_{1/r})\big) \otimes \big(\ones_{1/r^{k-1}}I_{r^{k-1}} \big)\\
&=& (I_{r^{m-k}} \otimes  \ones_{1/r}) \otimes \ones_{1/r^{k-1}}\\
&=& I_{r^{m-k}} \otimes \ones_{1/r^{k}},
\end{eqnarray*}
i.e.~$\prod_{q=1}^{k}A_{q} = I_{r^{m-k}} \otimes \ones_{1/r^{k}}$ for all $k\in[m]$, from which the remaining part of \assref{ass:A_k} follows by substituting $k=m$.

Finally, the required equivalence then holds by \eqnref{eq:elementwise formula for radix}, because $\ones_{1/r}$ has all entries strictly positive and $I_{r^{m-k}}$ and $I_{r^{k-1}}$ are identity matrices.
\end{proof}

Using \lemmaref{lem:modular congruence realtion fixed}, we have by inspection of \algrefmy{alg:augmented_resampling} that for radix-$r$ resampling with any $i\in [N]$ and $k\in[m]$,
\begin{equation}\label{eq:but_cond_id}
\xi_k^i \sim \frac{\sum_j A_k^{ij}V_{k-1}^j \delta_{\xi_{k-1}^j} }{ \sum_j A_k^{ij}V_{k-1}^j } = \frac{\sum_{\{j:  i\stackrel{(k,r)}{\sim}j \}} V_{k-1}^j \delta_{\xi_{k-1}^j} }{ \sum_{\{j:  i\stackrel{(k,r)}{\sim}j \}} V_{k-1}^j},
\end{equation}
and the following conditional independence holds:
\begin{equation}\label{eq:but_cond_ind}
i\stackrel{(k,r)}{\sim} j \quad\Longrightarrow \quad\xi_k^i \independent \xi_k^j \,\big|\, (\xi_{k-1}^u,V_{k-1}^u; u \stackrel{(k,r)}{\sim} i).
\end{equation}
These kind of considerations underly much of our convergence study. As illustrated in Figure \ref{fig:example of fixed and mixed radix} (a), for radix-$r$ resampling, the parameter $r$, which is equal to $|\{j:  i\stackrel{(k,r)}{\sim}j \}|$ for all $i\in[N],k\in[m]$, is the number of incoming edges for the vertices corresponding to the random variables $\{ \xi^i_k ; i\in[N] , k\in[m]\}$. Recalling that here $N=r^m$, the total number of edges in the graph is then $rN \log_r N$.


{

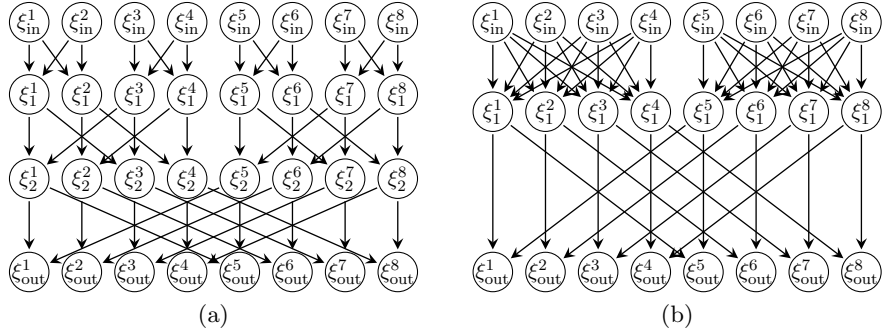
\begin{figure}[t!]
\begin{center}
\tikzset{
    vertex/.style = {
    	draw,
	    circle,
        fill      = white,
        outer sep = 2pt,
        inner sep = 2pt,
    }
}

\tikzset{
    mainvertex/.style = {
    	draw,
		circle,
        fill = black,
        outer sep = 2pt,
        inner sep = 2pt,
    }
}

\tikzset{
    demovertex/.style = {
    	draw,
		circle,
        fill = gray,
        outer sep = 2pt,
        inner sep = 2pt,
    }
}

\tikzset{
    basicvertex/.style = {
    	draw,
		circle,
		minimum size = .7cm
    }
}

\tikzset{
    basicinnervertex/.style = {
    	draw,
		circle,
		minimum size = .6cm
    }
}

\tikzset{
    demoedge/.style = {
		-stealthnew,
		shorten <=.15cm, 
		shorten >=.15cm,
		arrowhead=2mm,
		line width=2pt    
		}
}

\tikzset{
    basicedge/.style = {
		-stealthnew,
		color = black,
		shorten <=.3cm, 
		shorten >=.3cm,
		arrowhead=1.5mm,
		line width=.5pt    
		}
}

\begin{tabular}{c@{}c}
\mbox{
\begin{minipage}{0.47\textwidth}
\begin{center}
\mbox{
\begin{adjustbox}{trim={1mm} {0mm} {1mm} {0mm},clip}
\begin{tikzpicture}[scale=.7]
\tikzstyle{every node}=[scale=.75,font=\normalsize]
\def \vstep {-1.5cm}
\def \hstep {1cm}
\def \r {2}
\def \m {3} 
\def \N {8}
\def \X {4}
\def \Y {2}
\def \vspaces{{1,2,3,4}}

\foreach \i in {0,...,3}
{
	\foreach \j in {1,...,8}
	{
		\pgfmathparse{pow(\vspaces[\i],1.3)/1.6};
  	\node[basicvertex] (\i\j) at (\j*\hstep,\pgfmathresult*\vstep) {};
	}
}	


\foreach \j in {1,...,8}
{
	\pgfmathparse{pow(\vspaces[0],1.3)/1.6};
	\node[circle] at (\j*\hstep,\pgfmathresult*\vstep) {$\xi^{\j}_{\mathrm{in}}$};
}

\foreach \j in {1,...,8}
{
	\pgfmathparse{pow(\vspaces[\m],1.3)/1.6};
	\node[circle] at (\j*\hstep,\pgfmathresult*\vstep) {$\xi^{\j}_{\mathrm{out}}$};
}

\foreach \i in {1,...,2}
{
	\foreach \j in {1,...,8}
	{
		\pgfmathparse{pow(\vspaces[\i],1.3)/1.6};
  	\node[circle] at (\j*\hstep,\pgfmathresult*\vstep) {$\xi^{\j}_{\i}$};
	}
}	

\foreach \i in {0,...,2}
{
	\foreach \j in {1,...,8}
	{	
		\foreach \k in {0,...,1}
		{
			\pgfmathsetmacro{\result}{pow(\vspaces[\i],1.3)/1.6}
			\pgfmathsetmacro{\resulttwo}{pow(\vspaces[\i+1],1.3)/1.6}
			\draw[basicedge]
			let 
				\n1 = {int(mod((\j-1),\r^(\i))+\k*\r^(\i)+\r^(\i+1)*int(floor((\j-1)/\r^(\i+1))))+1}, 
				\n2 = {\resulttwo}
			in 							
				(\j*\hstep,\result*\vstep) -- (\n1*\hstep,\n2*\vstep);
		}
	}
}	

\end{tikzpicture}
\end{adjustbox}
}
\end{center}
\end{minipage}
}
& 
\mbox{
\begin{minipage}{.47\textwidth}
\begin{center}
\mbox{
\begin{adjustbox}{trim={1mm} {0mm} {1mm} {0mm},clip}
\begin{tikzpicture}[scale=0.7]
\tikzstyle{every node}=[scale=.75,font=\normalsize]

\def \vstep {-1.7cm}
\def \bigvstep {-2.3732cm}
\def \hstep {1cm}
\def \r {4}
\def \m {3} 
\def \N {8}
\def \X {4}
\def \Y {2}

\foreach \i in {0,...,1}
{
	\foreach \j in {1,...,8}
	{
  		\node[basicvertex] (\i\j) at (\j*\hstep,\i*\vstep) {};
	}
}
\foreach \i in {2}
{
	\foreach \j in {1,...,8}
	{
		\node[basicvertex] (\i\j) at (\j*\hstep,\i*\bigvstep) {};
	}
}

\foreach \j in {1,...,8}
{
	\node[circle] at (\j*\hstep,0*\vstep) {$\xi^{\j}_{\mathrm{in}}$};
}

\foreach \j in {1,...,8}
{
	\node[circle] at (\j*\hstep,2*\bigvstep) {$\xi^{\j}_{\mathrm{out}}$};
}

\foreach \i in {1}
{
	\foreach \j in {1,...,8}
	{
  		\node[circle] at (\j*\hstep,\i*\vstep) {$\xi^{\j}_{\i}$};
	}
}	

\foreach \i in {0}
{
	\foreach \j in {1,...,8}
	{	
		\foreach \k in {0,...,3}
		{
			\draw[basicedge]
			let 
				\n1 = {mod((\j-1),\r^(\i))+\k*\r^(\i)+\r^(\i+1)*int(floor((\j-1)/\r^(\i+1)))+1}, 
				\n2 = {int(\i+1)} 
			in 
				(\j*\hstep,\i*\vstep) -- (\n1*\hstep,\n2*\vstep);
		}
	}
}	

\def \r {2}
\foreach \i in {2}
{
	\foreach \j in {1,...,8}
	{	
		\foreach \k in {0,...,1}
		{
			\draw[basicedge]
			let 
				\n1 = {mod((\j-1),\r^(\i))+\k*\r^(\i)+\r^(\i+1)*int(floor((\j-1)/\r^(\i+1)))+1}, 
				\n2 = {2} 
			in 
				(\j*\hstep,\vstep) -- (\n1*\hstep,\n2*\bigvstep);
		}
	}
}	

\end{tikzpicture}
\end{adjustbox}
}
\end{center}
\end{minipage}
}\\
(a) & (b)
\end{tabular}
\end{center}
\caption{The conditional independence structure of (a) the radix-$r$ algorithm with $r=2$, $m = 3$ and $N=8$ and (b) the mixed radix-$r$ algorithm with $r=2$, $c=4$ and $N=8$.}
\label{fig:example of fixed and mixed radix}
\end{figure}

}

%
\begin{figure}[t!]
\begin{center}
\tikzset{
    vertex/.style = {
    	draw,
	    circle,
        fill      = white,
        outer sep = 1.5pt,
        inner sep = 1.5pt,
    }
}

\tikzset{
    mainvertex/.style = {
    	draw,
		circle,
        fill = black,
        outer sep = 2pt,
        inner sep = 2pt,
    }
}

\tikzset{
    demovertex/.style = {
    	draw,
		circle,
        fill = gray,
        outer sep = 2pt,
        inner sep = 2pt,
    }
}

\mbox{
\begin{adjustbox}{trim={.0\width} {.08\height} {0.0\width} {.25\height},clip}
\begin{tikzpicture}[scale=0.8]

\def \vstep {-.55cm}
\def \vstepf {-2.8cm}
\def \vstepm {-2.8cm}
\def \hstep {.45cm}
\def \graphHoffset {0cm}
\def \graphVoffset {3cm}
\def \fmHoffset {9cm}
\def \r {2}
\def \N {8}
\def \X {4}
\def \Y {2}
\def\kms{{2,4,8,16}}
\def\kmsM{{2,4,6,8}}
\edef\voffsetA {0}
\edef\voffsetB {0}
\edef\hoffsetB {0}
\edef\hoffsetBB {0}
\edef\mlab{0}
\def\maxH{16}
\def\levels{{0,1,1.585,2}}
\edef \rm {2}
\edef\lev{0}

\edef\tmp{0}
	
\node[circle,align=left] at (2.25*\hstep + 4.75*\hstep+1.5*\hstep,1cm) {\makebox[4cm][l]{Radix-$2$ butterfly. $N=2^m$, $m=1,2,3,4$.}};
	
\foreach \m in {0,...,3}
{
	\X = \kms[\m]
	\pgfmathsetmacro\stages{int(\m+1)}
	\pgfmathsetmacro\X{\kms[\m]}	
	\pgfmathparse{\voffsetA+\m}
	\xdef\voffsetA{\pgfmathresult}
	\pgfmathparse{0*\vstepf}
	\xdef\voffsetB{\pgfmathresult}	

	\pgfmathparse{\hoffsetB+pow(2,\m)*\hstep + 1.25*\hstep}
	\xdef\hoffsetB{\pgfmathresult}	
	

	
	\foreach \i in {0,...,\stages}
	{
		\foreach \j in {1,...,\X}
		{
				
	  		\node[vertex] (\i\j) at (\j*\hstep+\hoffsetB,\i*\vstep+\voffsetB) {};
			
		}
	}	

	
	\pgfmathparse{pow(2,\m+1)}
	\xdef\tmp{\pgfmathresult}
	\node[circle,align=left] at (1pt*\hoffsetB+4.75*\hstep+1.5*\hstep,.5cm) {\makebox[4cm][l]{$N=\pgfmathprintnumber{\tmp}$:}};
	
	
	\pgfmathsetmacro\stagestwo{\m}
	\pgfmathsetmacro\Xtwo{\X-1}
	
	\foreach \i in {0,...,\stagestwo}
	{
		\foreach \j in {1,...,\X}
		{	
			\foreach \k in {0,...,1}
			{
				\draw[-stealthnew,shorten <=.1cm, shorten >=.1cm,arrowhead=1mm]
				let
					\n1 = {int(mod((\j-1),\r^(\i))+\k*\r^(\i)+\r^(\i+1)*int(floor((\j-1)/\r^(\i+1))))+1},
					\n2 = {int(\i+1)}
				in
					(\j*\hstep+\hoffsetB,\i*\vstep+\voffsetB) --
					(\n1*\hstep+\hoffsetB,\n2*\vstep+\voffsetB);
			}
		}
	}	
	
}

\xdef\hoffsetB{0}
\xdef\voffsetB{-3.8cm}

\node[circle,align=left] at (2.25*\hstep + 4.75pt*\hstep + 1.5*\hstep,1cm+\voffsetB) {\makebox[4cm][l]{Mixed radix-$2$ butterfly. $N=2c$, $c=1,2,3,4$.}};
	
\foreach \m in {0,...,3}
{
	\X = \kmsM[\m]
	\pgfmathsetmacro\stages{int(\m+1)}
	\pgfmathsetmacro\X{\kmsM[\m]}	
	\pgfmathparse{\voffsetA+\m}
	\xdef\voffsetA{\pgfmathresult}
	
	
	\pgfmathparse{\hoffsetB+pow(2,\m)*\hstep + 1.25*\hstep}
	\xdef\hoffsetB{\pgfmathresult}	
	

	\foreach \i in {0,...,1}
	{
		\foreach \j in {1,...,\X}
		{
  			\node[vertex] at (\j*\hstep+\hoffsetB,\i*\vstep+\voffsetB) {};
		}
	}
	\foreach \i in {2}
	{
		\foreach \j in {1,...,\X}
		{
			\node[vertex] at (\j*\hstep+\hoffsetB,\i*\vstep+\voffsetB) {};
		}
	}
	
	\pgfmathparse{0.005*\vstep+\voffsetB}
	\xdef\mlab{\pgfmathresult}
	
	
	\node[circle] at (1pt*\hoffsetB+4.75*\hstep+1.5*\hstep,.5cm+\voffsetB) {\makebox[4cm][l]{$N=\pgfmathprintnumber{\X}$:}};
	

	\pgfmathparse{.5*\X}
	\xdef\r{\pgfmathresult}
	\pgfmathparse{\r-1}
	\xdef\rm{\pgfmathresult}

	\foreach \i in {0}
	{
		\foreach \j in {1,...,\X}
		{	
			\foreach \k in {0,...,\rm}
			{
				\draw[-stealthnew,shorten <=.1cm, shorten >=.1cm,arrowhead=1mm]
				let
					\n1 = {int(mod((\j-1),\r^(\i))+\k*\r^(\i)+\r^(\i+1)*int(floor((\j-1)/\r^(\i+1))))+1},
					\n2 = {int(\i+1)}
				in
					(\j*\hstep+\hoffsetB,\i*\vstep+\voffsetB) -- (\n1*\hstep+\hoffsetB,\n2*\vstep+\voffsetB);
			}
		}
	}
		
	\xdef\r{2}
	\pgfmathparse{\r-1}
	\xdef\rm{\pgfmathresult}
	
	\xdef\lev{\levels[\m]}
	\foreach \i in {\lev}
	{
		\foreach \j in {1,...,\X}
		{	
			\foreach \k in {0,...,\rm}
			{
				\draw[-stealthnew,shorten <=.1cm, shorten >=.1cm,arrowhead=1mm]
				let
					\n1 = {mod((\j-1),\r^(\i))+\k*\r^(\i)+\r^(\i+1)*int(floor((\j-1)/\r^(\i+1)))+1},
					\n2 = {2}
				in
					(\j*\hstep+\hoffsetB,\vstep+\voffsetB) -- (\n1*\hstep+\hoffsetB,\n2*\vstep+\voffsetB);
			}
		}
	}	
}

\end{tikzpicture}
\end{adjustbox}
}
\end{center}
\caption{Growth of the conditional independence graphs for radix-$2$ and mixed radix-$2$ algorithms.}
\label{fig:graph growth illustration}
\end{figure} 
As a visual preface to our convergence results, Figure \ref{fig:graph growth illustration} shows the sequence of graphs corresponding to $\radixMatrices{2}{m}$ for $m=1,2,3,4$.
The bound of Proposition \ref{prop:intro_to_aug_resampling} with $m=\log_r N$ for radix-$r$ resampling and $p=1$ is
\begin{equation*}
b_1 \sqrt{\frac{\log_r N}{N}} \infnorm{g}\osc{\varphi}.
\end{equation*}
It turns out that $\sqrt{\log_r N / N}$ is, asymptotically, the exact scale of the stochastic error for the particle filter when radix-$r$ resampling is used. However, this is far from trivial to prove due to the intricacies of the butterfly dependence structure and, in particular, the fact that there are several equivalence classes of conditionally-i.i.d. samples as per \eqref{eq:but_cond_id}-\eqref{eq:but_cond_ind}, rather than a single such equivalence class for multinomial resampling  \eqref{eq:multinomial_resampling}.  For $r\geq2$, $\varphi \in \boundMeas{\ss}$ and $n\geq 1$ define
\begin{equation}
\label{eq:asymptotic variance fixed radix filter 1}
\arraycolsep=1.4pt
\begin{array}{rcl}
\frbfVar{0}{\varphi}{r} &\defeq& \pi_0((\varphi - \pi_0(\varphi))^2),\\[.15cm]
\frbfVar{n}{\varphi}{r} &\defeq& \frbfVarHat{n-1}{f(\varphi)}{r},\\[.15cm]
\frbfVarHat{0}{\varphi}{r} &\defeq& (1-r^{-1})\hat{\pi}_0((\varphi - \hat{\pi}_0(\varphi))^2),\\[.15cm]
\frbfVarHat{n}{\varphi}{r} &\defeq& (1 - {r^{-1}})\hat{\pi}_n((\varphi - \hat{\pi}_n(\varphi))^2) \\[.15cm]
&&+~\pi_n(g_n)^{-2}\frbfVar{n}{g_n(\varphi - \hat{\pi}_n(\varphi))}{r}.
\end{array}
\end{equation}
Assuming that the above quantities are all strictly positive,  we have:
\begin{theorem}\label{thm:radix-r}
For any $\varphi\in\boundMeas{\ss}$ and $r\geq2$, the particle filter with radix-$r$ butterfly resampling has the properties that
\begin{equation}\label{eq:radix-r_front_clt_1}
\arraycolsep=1.4pt
\begin{array}{rclrcl}
\pi_0^N(\varphi)-\pi_0(\varphi)&\almostsurelyOneArgSh{}&0,&~\displaystyle\sqrt{N}\big(\pi_0^N(\varphi)-\pi_0(\varphi)\big)&\indistsh{}&\normal{0}{\frbfVar{0}{\varphi}{r}},\\ 
\hat{\pi}_0^N(\varphi)-\hat{\pi}_0(\varphi)&\almostsurelyOneArgSh{}&0,&~\displaystyle\sqrt{\frac{N}{\log_r N}}\big(\hat{\pi}_0^N(\varphi)-\hat{\pi}_0(\varphi)\big)&\indistsh{}&\normal{0}{\frbfVarHat{0}{\varphi}{r}}, 
\end{array}
\end{equation}
and for any $n\geq1$,
\begin{equation}\label{eq:radix-r_front_clt_2}
\arraycolsep=1.4pt
\begin{array}{rclrcl}
\pi_n^N(\varphi)-\pi_n(\varphi)&\almostsurelyOneArgSh{}&0,&~\displaystyle\sqrt{\frac{N}{\log_r N}}\big(\pi_n^N(\varphi)-\pi_n(\varphi)\big)&\indistsh{}&\normal{0}{\frbfVar{n}{\varphi}{r}},\\
\hat{\pi}_n^N(\varphi)-\hat{\pi}_n(\varphi)&\almostsurelyOneArgSh{}&0,&~\displaystyle\sqrt{\frac{N}{\log_r N}}\big(\hat{\pi}_n^N(\varphi)-\hat{\pi}_n(\varphi)\big)&\indistsh{}&\normal{0}{\frbfVarHat{n}{\varphi}{r}},
\end{array}
\end{equation}
where in \eqref{eq:radix-r_front_clt_1}--\eqref{eq:radix-r_front_clt_2} the convergence is as $N\to\infty$ along the sequence of integer population sizes $(r^m;m=1,2,\ldots)$ for which the radix-$r$ butterfly resampling algorithm is defined.
\end{theorem}
\begin{remark}\label{rem:stability}
Under various conditions on the HMM and observation sequence, \cite{smc:the:C04,whiteley2013stability} have proved uniform bounds of the form $\sup_n \sigma _n ^2 (\varphi)<\infty$ and \cite{del2001interacting,favetto2009asymptotic, douc2014stability} have shown that the sequence $(\sigma_n^2(\varphi))_{n\geq0}$, regarded as a function of random observations, is tight. In the present setting, it is easily checked that $\frbfVar{n}{\varphi}{r} \leq \sigma_n^2(\varphi)$ and $\frbfVarHat{n}{\varphi}{r}  \leq \hat{\sigma}_n^2(\varphi)$, allowing immediate transfer of the aforementioned results to the particle filter with radix-$r$ resampling.
\end{remark}

One interpretation of \theref{thm:radix-r} is that constraining interaction so that the degree of any vertex in graph does not grow with $N$ leads to slower convergence than the BPF. This leads us to consider our second butterfly resampling scheme.

\subsection{Mixed radix-\texorpdfstring{$r$}{r} resampling algorithm}
\label{sec:mixed radix main results}

For each $r\geq2$ and $c\geq 1$ consider the pair of matrices,
\begin{equation}
\mradixMatrices{r}{c}= (A_1,A_2),  \quad A_k = I_{r^{2-k}}\otimes \ones_{1/(r^{k-1}c^{2-k})} \otimes I_{c^{k-1}}, \quad k \in \{1,2\}.
\label{eq:mradix matrix definition}
\end{equation}
We shall refer to \algrefmy{alg:augmented_resampling} applied with the matrices in \eqref{eq:mradix matrix definition}, $m=2$ and $N=rc$ as the \emph{mixed radix-\texorpdfstring{$r$}{r} butterfly resampling algorithm}. For each $k\in\{1,2\}$ and $r\geq2$ introduce the following congruence relation on $[N]$:
\begin{equation}\label{eq:mixed_congruence_rel}
i \stackrel{(k,r)}{\sim} j \quad \Longleftrightarrow \quad\begin{cases} \floor{\dfrac{i-1}{r^{k-1}c}} = \floor{\dfrac{j-1}{r^{k-1}c}}, & \\
\quad \quad \quad\text{and}\phantom{\Big|} & \\
(i-1)\bmod c^{k-1}=(j-1)\bmod c^{k-1}. &
\end{cases}
\end{equation}

\begin{lemma} The matrices in \eqref{eq:mradix matrix definition} satisfy Assumption \ref{ass:A_k}. Moreover they are symmetric, for $k\in\{1,2\}$, $A_k$ has entries which are either $1/(r^{k-1}c^{2-k})$ or zero, and the non-zero entries are characterized by:
$$
A_{k}^{ij}>0 \quad\Longleftrightarrow \quad i \stackrel{(k,r)}{\sim} j.
$$
\end{lemma}
\begin{proof}
The symmetry follows from \eqnref{eq:elementwise Kronecker formula} and the fact that $A_1$ and $A_2$ are defined as Kronecker products of symmetric matrices. Also the idempotence of $A_1$ and $A_2$ as well as
\begin{equation*}
(I_{r} \otimes \ones_{1/c}) (\ones_{1/r} \otimes I_{c}) = I_{r}\ones_{1/r} \otimes \ones_{1/c}I_{c} = \ones_{1/r} \otimes \ones_{1/c} = \ones_{1/(rc)},
\end{equation*}
follow from the mixed product property \eqnref{eq:mixed product property}, proving the product part of \assref{ass:A_k}. Similarly as in the proof of \lemmaref{lem:modular congruence realtion fixed}, we have by two applications of the element-wise formula \eqnref{eq:elementwise Kronecker formula}
\begin{align*}
A_{k}^{ij} &= I_{r^{2-k}}^{\floor{\frac{i-1}{r^{k-1}c}}+1,\floor{\frac{j-1}{r^{k-1}c}}+1}\ones_{1/(r^{k-1}c^{2-k})}^{\left(\floor{\frac{i-1}{c^{k-1}}} \bmod r^{k-1}c^{2-k}\right)+1,\left(\floor{\frac{j-1}{c^{k-1}}} \bmod r^{k-1}c^{2-k}\right)+1} \\
&\times~ I_{c^{k-1}}^{((i-1)\bmod c^{k-1})+1,((j-1)\bmod c^{k-1})+1},
\end{align*}
where we have also used the fact that $\lfloor\lfloor(i-1)/c^{k-1}\rfloor / r^{k-1}c^{2-k}\rfloor = \lfloor(i-1)/r^{k-1}c\rfloor$, and $\lfloor\lfloor(j-1)/c^{k-1}\rfloor / r^{k-1}c^{2-k}\rfloor = \lfloor(j-1)/r^{k-1}c\rfloor$.

From this it is clear that $A_{k}^{ij}\in \{0,1/(r^{k-1}c^{2-k})\}$, and $A_{k}^{ii} = 1/(r^{k-1}c^{2-k})$. The double stochasticity then follows from these facts similarly as in the proof of \lemmaref{lem:modular congruence realtion fixed} by the symmetry and idempotence. Finally, by the positivity of all elements of $\ones_{1/(r^{k-1}c^{2-k})}$, the required equivalence follows.
\end{proof}
The formulae \eqref{eq:but_cond_id}-\eqref{eq:but_cond_ind} hold for the the mixed radix-$r$ algorithm, with the congruence relation \eqref{eq:mixed_congruence_rel}.
Figures \ref{fig:example of fixed and mixed radix} (b) and \ref{fig:graph growth illustration}  show the graphs, the latter for the case $r=2$ and $c=1,2,3,4$. For the mixed radix $r$-algorithm, note that the number of rows $m+1=3$ is fixed, $r$ is equal to the degree of the vertices in the bottom row, and $c$ is equal to the number of incoming edges for the vertices in the middle row.

It turns out that the mixed radix $r$-algorithm has the same rate of convergence as the BPF.
For all $r\geq 2$ and $\varphi \in \boundMeas{\ss}$ define
\begin{equation}\label{eq:asymptotic variance mixed radix filter}
\arraycolsep=1.4pt
\begin{array}{rcll}
\mrbfVar{0}{\varphi}{r}&\defeq& \pi_0((\varphi - \pi_0(\varphi))^2),&\\[.15cm]
\mrbfVar{n}{\varphi}{r}&\defeq& \mrbfVarHat{n-1}{f(\varphi)}{r} + \hat{\pi}_{n-1}(f((\varphi - f(\varphi))^2)),&~n\geq 1,\\[.15cm]
\mrbfVarHat{n}{\varphi}{r}&\defeq& \left(2-r^{-1}\right)\hat{\pi}_n\left((\varphi - \hat{\pi}_n(\varphi))^2\right) \\[.15cm]
&& + {\pi_n(g_n)^{-2}}\mrbfVar{n}{g_n(\varphi-\hat{\pi}_n(\varphi))}{r}, &~n\geq 0.
\end{array}
\end{equation}
Assuming the quantities in \eqref{eq:asymptotic variance mixed radix filter} are strictly positive, we have:
\begin{theorem}\label{thm:mixed radix-r}
For any $\varphi\in\boundMeas{\ss}$ and $r\geq 2$, the particle filter with mixed radix-$r$ butterfly resampling has the properties that for any $n\geq0$,
\begin{equation}\label{eq:mradix-r_front_clt_1}
\arraycolsep=1.4pt
\begin{array}{rclrcl}
\pi_n^N(\varphi)-\pi_n(\varphi)&\almostsurelyOneArgSh{N\to\infty}&0,&~~\sqrt{N}\big(\pi_n^N(\varphi)-\pi_n(\varphi)\big)&\indistsh{N\rightarrow\infty}&\normal{0}{\mrbfVar{n}{\varphi}{r}},\\
\hat{\pi}_n^N(\varphi)-\hat{\pi}_n(\varphi)&\almostsurelyOneArgSh{N\to\infty}&0,&~~\sqrt{N}\big(\hat{\pi}_n^N(\varphi)-\hat{\pi}_n(\varphi)\big)&\indistsh{N\rightarrow\infty}&\normal{0}{\mrbfVarHat{n}{\varphi}{r}},
\end{array}
\end{equation}
where the convergence is as $N\to\infty$ along the sequence of integer population sizes $(rc\,;c=1,2,\ldots)$ for which the mixed radix-$r$ butterfly scheme is defined.
\end{theorem}
A simple induction shows that for any $n\geq0$, $\sigma_n^2(\varphi)\leq \mrbfVar{n}{\varphi}{r} \leq \left(2-r^{-1}\right) \sigma_n^2(\varphi)$, and the same inequalities hold with $\mrbfVar{n}{\varphi}{r},\sigma_n^2(\varphi)$ replaced by $\mrbfVarHat{n}{\varphi}{r},\hat{\sigma}_n^2(\varphi)$. Thus the stability properties of \remref{rem:stability} also apply to the particle filter with mixed radix-$r$ resampling.

\subsection{Discussion}\label{sec:discussion}
A summary of the edge characteristics for the graphs of the algorithms we have considered is as follows (excluding vertices $(\xi_{\mathrm{in}}^i)_{i\in[N]}$).

\vspace{0.5cm}
\begin{center}
  \begin{tabular}{ l | c | c }

      & Incoming edges per vertex & Total edges \\ \hline
    Multinomial & $N$ & $N^2$ \\ \hline
    Radix-$r$ butterfly & $r$ & $r N \log_r N$ \\ \hline
    Mixed radix-$r$ butterfly & $r$ or $N/r$ & $rN + N^2/r$ \\
    \hline
  \end{tabular}
\end{center}
\vspace{0.5cm}
With this as a backdrop, let us compare and contrast Theorems \ref{thm:bootstrap}-\ref{thm:mixed radix-r}. The behaviour of $\pi_0^N(\varphi)$ is of course common to all three results.  Theorem \ref{thm:radix-r} shows the unusual scaling of the radix-$r$ algorithm; the higher the value of $r$ the faster the convergence, but for any finite $r$, the convergence is slower than that of the BPF. This phenomenon and the factor of $(1-r^{-1})$ present in $\frbfVarHat{0}{\varphi}{r}$ and $\frbfVarHat{n}{\varphi}{r}$ have underlying connections to the facts displayed in the table above, namely that the number of incoming edges per node for the radix-$r$ butterfly is fixed to $r$ and in particular is non-increasing in $N$, a characteristic not shared with the BPF, for which the number of incoming edges is $N$.

Note the term $\hat{\pi}_{n-1}(f((\varphi-f(\varphi))^2))$ is present in the functional $\sigma_n^2(\varphi)$ in \eqref{eq:asymp_var_bpf} but absent from $\frbfVar{n}{\varphi}{r}$ in \eqref{eq:asymptotic variance fixed radix filter 1}; the explanation is that for radix-$r$ resampling, the error associated with resampling is of order $\sqrt{\log_r N / N}$, where as the error associated with sampling $\zeta_n^i \sim f(\hat{\zeta}_{n-1}^{i},\cdot)$ for each $i\in[N]$ is of order $\sqrt{1/N}$, and therefore makes no contribution to the asymptotic variance (although it will contribute to the non-asymptotic variance in general). On the other hand Theorem \ref{thm:mixed radix-r} shows that the mixed radix-$r$ algorithm has the same scaling as the BPF, and the term $\hat{\pi}_{n-1}(f((\varphi-f(\varphi))^2))$ does appear in $\mrbfVar{n}{\varphi}{r}$. The difference is the factor of $(2-1/r)$ in $\mrbfVarHat{n}{\varphi}{r}$, which has underlying connections to the facts that for the mixed radix-$r$ algorithm, $m=2$ is a constant, and some vertices have $r$ incoming edges.

Let us close with some remarks about generality. One can derive as many instances of augmented resampling as one can factorizations of $\ones_{1/N}$ into non-negative matrices, there are many alternatives to the two butterfly algorithms we have studied. Also, in practice, one could easily combine butterfly sampling with other techniques such as stratified and adaptive resampling leading to variance reductions. Lastly, we note that the butterfly resampling schemes could be applied as part of many other algorithms and statistical procedures, not just particle filters.

\section{Analysis part I - augmented resampling and preparatory results}
\label{sec:convergence analysis part 1}

\subsection{A guide for the reader}\label{subsec:guide}

The remainder of the paper is structured so that the main results and ideas are given in Sections \ref{sec:convergence analysis part 1}-\ref{sec:convergence of particle filters}, which we recommend the reader browse first to get a sense for our strategy, before getting into the details of the proofs and more technical results in the \ref{suppA}. After some preliminaries in \secref{subsec:distributions}, the cornerstone of our analysis is a novel block-wise martingale difference decomposition result, \propref{prop:generalized martingale decomposition} of \secref{subsec:martingale}, which allows us to quantify the errors associated with certain sub-populations of the particle system, and we later put it to use in establishing the CLT's.


\theref{thm:Douc and Moulines} in \secref{sec:conditional clt for amrtingale array} is a conditional CLT for triangular martingale arrays proved by \citep{smc:the:DM08}, which we shall apply, while \secref{sec:martingale array mapping} describes how we map the martingales of \propref{prop:generalized martingale decomposition} in the cases of the two butterfly resampling schemes onto the triangular array format. Propositions \ref{prop:tens_prod_formula} and \ref{prop:tensor product formulation} provide novel tools to quantify second moment properties of augmented resampling, with a view to verifying the conditions of \theref{thm:Douc and Moulines}.


Statements and main proof steps of LLN's and CLT's for single applications of butterfly resampling, Theorems \ref{thm:butterfly_lln_fixed_radix}-\ref{thm:butterfly_clt_mixed_radix}, are then given in \secref{sec:convergence analysis part 2}. These rely on a number of novel but highly technical results given in the \ref{suppA}, in turn utilizing Propositions \ref{prop:generalized martingale decomposition}-\ref{prop:tensor product formulation}. An outline of proofs for Theorems \ref{thm:radix-r} and \ref{thm:mixed radix-r}, the LLN's and CLT's for particle filters, is given in Section \ref{sec:convergence of particle filters}, with the details in the \ref{suppA}.

\subsection{Probability law of the augmented resampling algorithm}
\label{subsec:distributions}

We begin building the theory with a more explicit probabilistic description of a single instance of Algorithm \ref{alg:augmented_resampling}. Consider $\xi_{\mathrm{in}}\defeq(\xi_{\mathrm{in}}^i)_{i\in[N]}$ and $(\xi_k)_{k\in[m]}$, where $\xi_k\defeq(\xi_k^1,\ldots,\xi_k^N)$ and each $\xi_{\mathrm{in}}^i $ and each $\xi_k^i$ are $\ss$-valued random elements. By convention, set $\xi_0\defeq\xi_{\mathrm{in}}$,  $\xi_{0}^i\defeq\xi_{\mathrm{in}}^i$ and $\xi_{\mathrm{out}}\defeq\xi_m$, $\xi_{\mathrm{out}}^i\defeq\xi_{m}^i$. Unless otherwise explicitly stated, the parameters $N,m\geq 1$ are assumed fixed and we write $\bba^{(N,m)}\defeq\aMatrices{N}{m}$ for the sequence of matrices parameterizing the augmented resampling algorithm. Moreover, the following regularity condition, prototypical of \assref{ass:g_n_bounded_and_positive}, is imposed from henceforth on the function $g$ passed to \algrefmy{alg:augmented_resampling}.
\begin{assumption}\label{ass:g_bounded_and_positive}
The function $g$ belongs to $\boundMeas{\ss}$ and is strictly positive.
\end{assumption}

Define for $i\in[N]$ and $k\in[m]$,
\begin{equation}
V_0^i:=g(\xi_0^i),\quad\quad V_k^i:=\sum_j A_k ^{ij}V_{k-1}^j.\label{eq:V_defn_proofs}
\end{equation}
The following facts about the $V_k^i$'s shall be used repeatedly.
\begin{lemma}\label{lem:facts_about_Vs}
Fix $N,m\geq 1$. For any $i\in[N]$ and $0\leq k \leq m$,
\begin{enumerate}[itemsep=4pt, topsep=5pt, partopsep=0pt,label={{(\roman*)}}]
\item \label{it:measurability of V} $V_k^i$ is measurable w.r.t.~$\sigma(\xi_{\mathrm{in}})$,
\item \label{it:boundedness of V} $V_k^i\leq\infnorm{g}$.
\end{enumerate}
If, in addition, $\bba^{(N,m)}$
satisfies \assref{ass:A_k}, then $V_m^i=N^{-1}\sum_j g(\xi_{\mathrm{in}}^j)$ for all $i\in[N]$.

\end{lemma}
\begin{proof}
From \eqref{eq:V_defn_proofs} we have $V_0^i=g(\xi_{\mathrm{in}}^i)$ and a simple induction shows that for $k\in[m]$,
\begin{equation}
V_k^{i_k}=\sum_{(i_0,\ldots,i_{k-1})}g(\xi_{\mathrm{in}}^{i_0})\prod_{q=1}^k A_q^{i_qi_{q-1}}.\label{eq:V_k_unwind}
\end{equation}
It is then clear that $V_k^{i_k}$ is measurable w.r.t.~$\sigma(\xi_{\mathrm{in}})$. Since each $A_k$ is a row-stochastic matrix, the bound $V_k^i\leq\infnorm{g}$ holds.
Applying \eqref{eq:V_k_unwind} in the case $k=m$ and using the assumption $\prod_{k=1}^m A_k =\mathbf{1}_{1/N}$ we find
\begin{equation*}
V_m^{i_m}=\sum_{(i_0,...,i_{m-1})}g(\xi_{\mathrm{in}}^{i_0})\prod_{q=1}^m A_q^{i_qi_{q-1}}=\sum_{i_0}g(\xi_{\mathrm{in}}^{i_0})\Bigg(\prod_{q=1}^m A_q\Bigg)^{i_mi_0}=\frac{1}{N}\sum_{i_0}g(\xi_{\mathrm{in}}^{i_0}).
\end{equation*}
\end{proof}

Algorithm \ref{alg:augmented_resampling} corresponds to the following distributional prescription. For each $k\in[m]$ the random elements $(\xi_k^i)_{i\in[N]}$ are conditionally independent given $(\xi_0,\ldots,\xi_{k-1})$, a property which will be frequently referred to as \emph{one step conditional independence}. Moreover, for each $i\in[N]$ and $S\in\sss$,
\begin{equation}
\P\big(\xi_k^i\in S \,\big|\, \xi_0,\ldots,\xi_{k-1}\big)=\frac{1}{V_k^i}\sum_j A_k^{ij}V_{k-1}^j \mathbb{I}_S(\xi_{k-1}^j).\label{eq:P_in_terms_of_V_proofs}
\end{equation}
Since $V_{k-1}^i$ is measurable w.r.t.~$\sigma(\xi_0)$, we notice from \eqref{eq:P_in_terms_of_V_proofs} that in fact
\begin{equation*}
\P\big(\xi_k^i\in S \,\big|\, \xi_0,\ldots,\xi_{k-1}\big)=\P\big(\xi_k^i\in S \,\big|\,\xi_0, \big(\xi_{k-1}^j;j \in[N], A_k^{ij}>0\big) \big).
\end{equation*}
We have also an explicit expression for the conditional marginal distribution of $\xi^i_k$, given $(\xi_0,\ldots,\xi_{q})$ where $0\leq q < k-1$, according to the following result for which the proof is given in \secref{sec:preparatory results} of the \ref{suppA}.
\begin{lemma}
\label{lem:conditional marginal}
Fix $N,m\geq 1$. If $\bba^{(N,m)}$ satisfies \assref{ass:A_k}, then for all $i\in[N]$, $k\in [m]$ and $S \in \sss$
\begin{equation*}
\P\left(\xi^{i}_m \in S \mids \xi_0,\ldots,\xi_{m-k}\right) = \frac{1}{V^{i}_m}\sum_{j}\Bigg(\prod_{q=0}^{k-1}A_{m-q}\Bigg)^{ij}V^{j}_{m-k}\ind{}(\bxi{j}{m-k}{}{}\in S). 
\end{equation*}
\end{lemma}

\subsection{Block-wise martingale decomposition}
\label{subsec:martingale}

Given $N\geq 1$ and a partition $\calI$ of $[N]$, $\calI = \{\calI_{u}\subset [N]:u\in [\card{\calI}]\}$, let $\partitionMap{\calI}$ be the set of all functions $J:[\card{\calI}]\rightarrow [N]$ such that for each $u\in [\card{\calI}]$, $J(u)$ is some member of $\mathcal{I}_u$.

This section addresses martingale decomposition of error terms of the form
\begin{equation}\label{eq:block-wise error term}
\Bigg(\frac{1}{N}\sum_{i}g(\xi_{\mathrm{in}}^i)\Bigg)\Bigg(\frac{1}{\card{\calI}}\sum_{i=1}^{\card{\calI}}\varphi(\xi_{\mathrm{out}}^{J(i)})\Bigg)-\frac{1}{N}\sum_{i}g(\xi_{\mathrm{in}}^i)\varphi(\xi_{\mathrm{in}}^i).
\end{equation}
Note that in the special case $\calI=\{\{u\};u\in[N]\}$, we have $\big|\partitionMap{\calI}\big|=1$, the unique member of $\partitionMap{\calI}$ is $J=Id$ and \eqref{eq:block-wise error term} reduces to the quantity in \eqnref{eq:aug_res_L_p_bound}. We shall use the generality of \eqref{eq:block-wise error term} beyond this special case to help prove our CLT's. Loosely speaking, we shall be concerned with partitions $\calI$ such that for any $(i,j)\in \calI_u \times \calI_v$ and some $d\in [m]$,
\begin{equation}\label{eq:Assmp_4_informal}
\arraycolsep=0.25cm
\begin{array}{lcl}
u = v & \Rightarrow & \P(\xi^{i}_{\mathrm{out}}\in\cdot |\xi_{0},\ldots,\xi_{m-\ka} )=\P(\xi^{j}_{\mathrm{out}}\in\cdot |\xi_{0},\ldots,\xi_{m-\ka} ),\\
u \neq v & \Rightarrow & \xi^{i}_{\mathrm{out}} \independent \xi^{j}_{\mathrm{out}} \,\big|\, \xi_{0},\ldots,\xi_{m-\ka}.
\end{array}
\end{equation}
 Whether or not \eqref{eq:Assmp_4_informal} holds obviously depends on the choice of matrices $\bba^{(N,m)}$, a matter which we shall formalize in \assref{ass:partition and cond indep} below.

 Let us now proceed with the precise details. We shall make multiple uses of the objects which we define next and this flexibility is accommodated by our notation, which is a little intricate, but provides just what we need.

For $m\geq 1$, define the index mappings $p_N:[Nm] \to [N]$ and $s_N:[Nm]\to[m]$, for each $\varrho\in[Nm]$ as
\begin{equation*}
p_N(\varrho) \defeq ((\varrho-1)\bmod N) + 1,\qquad
s_N(\varrho) \defeq \ceil{\frac{\varrho}{N}}.
\end{equation*}

Now for given $\ka\in[m]$, a partition $\calI$ of $[N]$ and $J \in \partitionMap{\calI}$,
we define
the $\sigma$-algebras 
$\big(\calF^{(N,m)}_{\varrho}\big)_{0\leq \varrho\leq(m-\ka)N+\card{\calI}}$ as
\begin{equation}
\calF^{(N,m)}_\varrho =
\begin{cases}
\displaystyle\sigma(\xi_{\mathrm{in}}),&\varrho = 0, \\
\displaystyle\calF^{(N,m)}_{\varrho-1} \vee \sigma\big(\bxi{p_{N}(\varrho)}{s_{N}(\varrho)}{}{}\big),&0<\varrho \leq N^{\ast},\\
\displaystyle\calF^{(N,m)}_{\varrho-1} \vee \sigma\big(\bxi{J(p_{N}(\varrho))}{m}{}{}\big),&\varrho > N^{\ast},\\
\end{cases}\label{eq:piecewise sigma fields}
\end{equation}
where $N^{\ast} \defeq (m-\ka)N$.

For $\varphi\in\boundMeas{\ss}$, let
\begin{equation}
\cvarphi_N(x):=\varphi(x)-\frac{\sum_i g(\xi_0^i)\varphi(\xi_0^i)}{\sum_i g(\xi_0^i)},\label{eq:varphi_bar_defn}
\end{equation}
and by writing $\cvarphi_{N,q}^{i} = \cvarphi_N(\xi^{i}_{q})$ for brevity, for all $i\in[N]$ and $0\leq q \leq m$, define the sequence $\big(X_\varrho^{(N,m)}\big)_{\varrho \in [(m-\ka)N+\card{\calI}]}$,

\begin{align}
&X_\varrho^{(N,m)} \defeq \nonumber\\[.2cm]
&\begin{cases}
\dfrac{S_{N,m,\ka}V_{q}^{i}}{N}\Bigg(\cvarphi_{N,q}^{i} - \dfrac{1}{V^{i}_{q}}\displaystyle\sum_{j}A_{q}^{ij}V_{q-1}^{j}\cvarphi_{N,q-1}^{j}\Bigg),&\!\!\!\!\varrho \leq N^{\ast},
\\
\dfrac{S_{N,m,\ka}V^{i}_m}{\card{\calI}}\Bigg(\cvarphi_{N,m}^{i}
 - \dfrac{1}{V^{i}_{m}} \displaystyle\sum_{j} \Bigg(\prod_{p=0}^{\ka-1}A_{m-p}\Bigg)^{ij}V_{m-\ka}^{j}\cvarphi_{N,m-\ka}^{j}\Bigg),&\!\!\!\!\varrho > N^{\ast},
\end{cases}\label{eq:explicit X}
\end{align}
where $q = s_{N}(\varrho)$, $i = p_{N}(\varrho)$ for all $0<\varrho\leq N^{\ast}$ and $i = J(p_{N}(\varrho))$ for all $N^{\ast}<\varrho\leq N^{\ast}+\card{\calI}$. The scaling factor $S_{N,m,\ka}$ is
\begin{equation}\label{eq:def scale factor}
S_{N,m,\ka}\defeq\left(\dfrac{m-\ka}{N} + \dfrac{1}{\card{\calI}}\right)^{-1/2}.
\end{equation}
We stress that $\calF^{(N,m)}_{\varrho}$ depends on $d,J$; $X_\varrho^{(N,m)}$ depends on $d,|\calI|,J,\varphi$; and $S_{N,m,\ka}$ depends on $|\calI|$; but these dependencies are suppressed from the notation.
%

The following assumption, which we shall invoke in \propref{prop:generalized martingale decomposition}, demands some specific relationships between the matrices $\bba^{(N,m)}$, the partition $\calI$ and the parameter $d$.
\begin{assumption}\label{ass:partition and cond indep}
For given $N,m\geq 1$, $\ka\in [m]$, $\bba^{(N,m)}$, and $\calI =\{\calI_{u}\subset [N]:u\in [\card{\calI}]\}$, the sequence of matrices $\bba^{(N,m)}$ satisfies \assref{ass:A_k} and the triple $(\bba^{(N,m)},\calI,d)$ has the following properties:
\begin{enumerate}[itemsep=4pt, topsep=5pt, partopsep=0pt,label={{(\roman*)}}]
	\item \label{it:is partition} $\calI$ is a partition of $[N]$ such that for all $u \in [\card{\calI}]$, $\big|\calI_{u}\big| = N/\card{\calI}\geq \ka$.
	\item \label{it:equivalence classes} For all $u \in [\card{\calI}]$, $j_1,j_2 \in \calI_{u}$ and $i\in [N]$,
	\begin{equation*}
\Bigg(\prod_{q=0}^{\ka-1}A_{m-q}\Bigg)^{j_1i} = \Bigg(\prod_{q=0}^{\ka-1}A_{m-q}\Bigg)^{j_2i}.
\end{equation*}
	\item \label{it:cond indep} For all $u,v \in [\card{\calI}]$ such that $u\neq v$, and $(i,j) \in \calI_{u}\times \calI_{v}$,  $\xi^{i}_{\mathrm{out}} \independent \xi^j_{\mathrm{out}} \big| \xi_0,\ldots,\xi_{m-\ka}$.
\end{enumerate}
\end{assumption}
\begin{remark}
The condition \ref{it:is partition} means that $\calI$ partitions $[N]$ into sets of equal sizes. By \lemmaref{lem:conditional marginal},  \ref{it:equivalence classes} ensures that the random variables $\xi^{i}_{\mathrm{out}}$ and $\xi^{j}_{\mathrm{out}}$, where $i$ and $j$ belong to the same element of the partition $\calI$, have conditionally identical distributions given $\xi_{0},\ldots,\xi_{m-\ka}$. Together with \ref{it:cond indep} this formalizes \eqref{eq:Assmp_4_informal}.
\end{remark}
\begin{remark}\label{rem:generality of additional assumptions}
\assref{ass:partition and cond indep} reduces to exactly \assref{ass:A_k} in the case that $\ka=1$ and $\calI=\{\{u\};u\in[N]\}$. To see this, note that then: $\card{\calI} = N$, so \ref{it:is partition} is satisfied;  $\calI_{u}=\{u\}$, so \ref{it:equivalence classes} is satisfied; and \ref{it:cond indep} is satisfied due to the one step conditional independence property of augmented resampling, stated above \eqref{eq:P_in_terms_of_V_proofs}.
\end{remark}

We can now present the martingale decomposition. The proof is given \secref{sec:preparatory results} of the \ref{suppA}.
\begin{proposition}\label{prop:generalized martingale decomposition}
If for some $N,m\geq 1$ and $\ka\in[m]$, $(\bba^{(N,m)},\calI,d)$ satisfies \assref{ass:partition and cond indep}, then for all $\varphi \in \boundMeas{\ss}$, $J \in \partitionMap{\calI}$ and $\varrho \in [(m-\ka)N+\card{\calI}]$,  the following hold:
\begin{enumerate}[itemsep=4pt, topsep=5pt, partopsep=0pt,label={{(\roman*)}}]
\item \label{eq:measurability} $X_{\varrho}^{(N,m)}$ is measurable w.r.t.~$\calF^{(N,m)}_\varrho$,
\item \label{eq:zero expectations} $\E\Big[X_{\varrho}^{(N,m)}\Big| \calF^{(N,m)}_{\varrho-1}\Big] = 0$,
\item $X_{\varrho}^{(N,m)}$ is bounded by
\begin{equation}\label{eq:boundedness of X}
\abs{X_\varrho^{(N,m)}} \leq
\begin{cases}
{S_{N,m,\ka}}{N^{-1}}\norm{g}_\infty\osc{\varphi},&\varrho\leq(m-\ka)N,\\
{S_{N,m,\ka}}{{\card{\calI}}^{-1}}\norm{g}_\infty\osc{\varphi},&\varrho>(m-\ka)N,
\end{cases}
\end{equation}
\item and we have the decomposition
\begin{align}
&\frac{1}{S_{N,m,\ka}}\sum_{\varrho=1}^{(m-\ka)N+\card{\calI}} X_\varrho^{(N,m)} \nonumber\\
&= \frac{1}{\card{\calI}}\sum_{i_m=1}^{\card{\calI}} V^{J(i_m)}_m\cvarphi_{N}(\bxi{J(i_m)}{m}{}{}) \label{it:martingale representation} \\
&= \Bigg(\frac{1}{N}\sum_{i}g(\xi_{\mathrm{in}}^i)\Bigg)\Bigg(\frac{1}{\card{\calI}}\sum_{i=1}^{\card{\calI}}\varphi(\xi_{\mathrm{out}}^{J(i)})\Bigg)-\frac{1}{N}\sum_{i}g(\xi_{\mathrm{in}}^i)\varphi(\xi_{\mathrm{in}}^i).\label{it:another martingale representation}
\end{align}
\end{enumerate}
\end{proposition}

We can now prove \propref{prop:intro_to_aug_resampling}.
\begin{proof}[Proof of \propref{prop:intro_to_aug_resampling}]
Let us choose $\ka=1$, $\calI=\{\{u\};u\in[N]\}$
and $J = \id$. In this case, $\card{\calI} = N$, $(m-\ka)N+\card{\calI} = Nm$, $S_{N,m,\ka} = \sqrt{N/m}$. \assref{ass:partition and cond indep} is satisfied for any $\aMatrices{N}{m}$ satisfying \assref{ass:A_k}  -- see Remark \ref{rem:generality of additional assumptions}. Therefore we can apply \propref{prop:generalized martingale decomposition}. The lack-of-bias property \eqref{eq:aug_res_lack_of_bias} follows immediately from \propref{prop:generalized martingale decomposition}\ref{eq:zero expectations}, \eqnref{it:another martingale representation} and the tower property of conditional expectation. For the moment bound \eqref{eq:aug_res_L_p_bound}, we apply the Burkholder-Davis-Gundy inequality and \eqnref{eq:boundedness of X} to obtain
\begin{align*}
&\E\left[\left|\sum_{\varrho\in[Nm]} X_\varrho^{(N,m)}\right|^p\right]\leq  b_p \E\left[\left|\sqrt{\sum_{\varrho\in[Nm]}\left(X_\varrho^{(N,m)}\right)^2}\right|^p\right]
\leq  b_p \infnorm{g}^p\osc{\varphi}^p.
\end{align*}
\end{proof}


\noindent{\textbf{Warning:}} Throughout the remainder of Sections \ref{sec:convergence analysis part 1}-\ref{sec:convergence of particle filters}, whenever the sequences $\big(\calF^{(N,m)}_{\varrho}\big)_{0\leq \varrho\leq(m-\ka)N+\card{\calI}}$ and $\big(X_\varrho^{(N,m)}\big)_{\varrho\in [(m-\ka)N+\card{\calI}]}$ appear, they are taken to be as in \eqref{eq:piecewise sigma fields} and \eqref{eq:explicit X} with specifically $\ka=1$, $\calI=\{\{u\}:u\in[N]\}$ and $J = \id$.


\subsection{Conditional CLT for martingale array}
\label{sec:conditional clt for amrtingale array}

In light of \propref{prop:generalized martingale decomposition}, for each $N$ and $m$, $\big(X_\varrho^{(N,m)}\big)_{\varrho\in[Nm]}$ is clearly a martingale difference sequence w.r.t.~$\big(\F_\varrho^{(N,m)}\big)_{0\leq \varrho \leq Nm}$. Our strategy is to study its behaviour using the following result, which is a special case of \citep[Theorem A.3]{smc:the:DM08}.

Let $(\ell_n)_{n\geq1}$ be a sequence of positive integer constants. Let $( U_{n,\varrho})_{\varrho \in [\ell_n]}$ be a triangular array of random variables and let $( \G_{n,\varrho} )_{0\leq \varrho \leq \ell_n}$ be a triangular array of sub-$\sigma$-algebras of the $\sigma$-algebra $\F$ of the underlying probability space, such that for each $n$ and $\varrho\in[\ell_n]$, $U_{n,\varrho}$ is $\G_{n,\varrho}$-measurable and $\G_{n,\varrho-1}\subseteq \G_{n,\varrho}$.
\begin{theorem}\label{thm:Douc and Moulines}
Assume that  $\E \left[ \left.  U_{n,\varrho}^2 \right| \G_{n,\varrho-1} \right]<\infty$ for any $n$ and $\varrho\in[\ell_n]$, and
\begin{align}
\E \left[ \left.  U_{n,\varrho} \right | \G_{n,\varrho-1}  \right]&=0,\label{eq:douc_cond_mart}\\[.35cm]
\sum_{\varrho\in[\ell_n]} \E \left[ \left.  U_{n,\varrho}^2 \mathbb{I}\{|U_{n,\varrho}|\geq \epsilon\} \right| \G_{n,\varrho-1} \right]&\inprob{n\rightarrow\infty} 0,\quad~\,\text{for any} \quad\epsilon>0, \label{eq:douc_cond_neg}\\
\sum_{\varrho\in[\ell_n]} \E \left[ \left.  U_{n,\varrho}^2 \right| \G_{n,\varrho-1} \right] &\inprob{n\rightarrow\infty} \sigma^2,\quad\text{for some}~\sigma^2>0. \label{eq:douc_cond_var}
\end{align}
Then, for any real $u$,
\begin{equation*}
\E\Bigg[\exp \Bigg( i u\sum_{\varrho\in[\ell_n]} U_{n,\varrho} \Bigg) \Bigg| \G_{n,0} \Bigg]\inprob{n\rightarrow\infty} \exp\left(-(u^2/2)\sigma^2\right).
\end{equation*}
\end{theorem}

\subsection{Triangular martingale array representation of butterfly resampling algorithms}
\label{sec:martingale array mapping}

In order to apply \theref{thm:Douc and Moulines} we need to map the martingales of \secref{subsec:martingale} onto the format of \theref{thm:Douc and Moulines}. This is done in a different way for each of the two butterfly resampling algorithms.

For the radix-$r$ algorithm, we have a fixed positive integer $r\geq 2$ and $N=r^m$ with $m\geq 1$. For the variables in \theref{thm:Douc and Moulines} we take $n=m$, $\ell_n=N m=r^m m$, and $U_{n,\varrho}=X_\varrho^{(r^m,m)}$ for all $\varrho \in [mr^m]$ and $\G_{n,\varrho}=\F_\varrho^{(r^m,m)}$ for $0\leq \varrho \leq m r^m$. In simple terms, the $m$th row of the array involves the random variables in an instance of the butterfly resampling scheme with population size $N=r^m$.


For the mixed radix-$r$ algorithm, we have a fixed positive integer $r\geq 2$ and the population size $N$ is taken to be an integer multiple of $r$, i.e.~$N=rc$ where $c\geq 1$. $m=2$ is a constant. For the variables in \theref{thm:Douc and Moulines} we take $n=c$, $\ell_n=2N=2rc$, $U_{n,\varrho}=\mrbfX{\varrho}{rc}{2}{\varphi}$ for all $\varrho\in [2rc]$, and $\G_{n,\varrho}=\mrbfF{\varrho}{rc}{2}$ for $0\leq \varrho \leq 2rc$.

For each of the butterfly algorithms, it is then easily checked that:  $\G_{n,\varrho-1}\subseteq \G_{n,\varrho}$, using \eqref{eq:piecewise sigma fields};   $U_{n,\varrho}$ is $\G_{n,\varrho}$-measurable, using \propref{prop:generalized martingale decomposition}; and finally $\E \left[ \left.  U_{n,\varrho}^2 \right| \G_{n,\varrho-1} \right]<\infty$ using \eqref{eq:boundedness of X}.

Our aim is to verify the remaining conditions of \theref{thm:Douc and Moulines}, the most challenging is \eqnref{eq:douc_cond_var}, and our next step is to develop some tools which help.

\subsection{Conditional variance and collision analysis}
\label{sec:conditional variance and collision analysis main}	

We shall use the following proposition to establish the connection between the conditional second moment of the martingale of \propref{prop:generalized martingale decomposition} and  the conditional independence structure of the augmented resampling algorithm through the matrices $\aMatrices{N}{m}$. The proof of the proposition is given in \secref{sec:proofs for collision analysis} of \ref{suppA}, and is partly inspired by \citep{smc:the:CdMG11}.
\newcommand{\Gam}[2]{\Gamma_{#1,#2}}
\newcommand{\cPhi}{\overline{\Phi}}
\newcommand{\aprod}[4]{A_{#1,#2}^{(#3,#4)}}
\newcommand{\ccomp}[4]{\mathcal{C}_{{#3}_{#1:#2},{#4}_{#1:#2}}}
\begin{proposition}\label{prop:tens_prod_formula}
For any $N\geq2$, $m\geq1$, $\varphi \in \boundMeas{\ss}$ and for any sequence of row stochastic matrices $\aMatrices{N}{m}$
\begin{align}
\label{eq:tensor product formula}
&\quad\frac{m}{N}\E\Bigg[\Bigg(\sum_{\varrho\in[Nm]} X_\varrho^{(N,m)}\Bigg)^2\,\Bigg|\,\F_{0}^{(N,m)}\Bigg]\\
&
=\displaystyle\sum_{\left(i_{0},j_{0},\ldots,i_{m},j_{m}\right)}\left(\dfrac{1}{N^{2}}\prod_{k=0}^{m-1}A_{k+1}^{i_{k+1}i_{k}}A_{k+1}^{j_{k+1}j_{k}}\right)g(\xi_{0}^{i_{0}})g(\xi_{0}^{j_{0}})\ccomp{1}{m}{i}{j}(\Phi)\big(\xi_{0}^{i_{0}},\xi_{0}^{j_{0}}\big)\nonumber
\end{align}
where $\Phi = \cvarphi_{N}^{\otimes 2}$, 
$\ccomp{1}{m}{i}{j} \defeq \mathcal{C}_{\mathbb{I}[i_{1}=j_{1}]}\cdots \mathcal{C}_{\mathbb{I}[i_{m}=j_{m}]}$,
and $\mathcal{C}_{0}$ and $\mathcal{C}_{1}$ act on functions $\boundMeas{\ss^2}\to\boundMeas{\ss^2}$ to the right as $\mathcal{C}_{0}:=Id$ and $\left(\mathcal{C}_{1}\Phi\right)(x,x^{\prime}):=\Phi(x,x)$.
\end{proposition}
When operating on the function $\varphi^{\otimes 2}$, the composite operator $\ccomp{1}{m}{i}{j}$ satisfies
\begin{equation}\label{eq:collision operator}
\ccomp{1}{m}{i}{j}(\varphi^{\otimes 2})(x,x') =
\begin{cases}
\varphi(x)\varphi(x), &\text{if  $i_k = j_k$ for some $k\in [m]$},\\
\varphi(x)\varphi(x'), &\text{otherwise.}
\end{cases}
\end{equation}
To determine which of the cases in \eqnref{eq:collision operator} is true, is equivalent to asking whether the sequences $(i_0,\ldots,i_m)$ and $(j_0,\ldots,j_m)$ have a common element $i_k = j_k$ for some $k \in [m]$, i.e.~if these sequences \emph{collide}. Consequently, formulating more tractable expressions for the r.h.s.~of \eqnref{eq:tensor product formula} boils down to finding the sets of pairs $(i_0,\ldots,i_m),(j_0,\ldots,j_m)$ for which the term $\prod_{k=0}^{m-1}A_{k+1}^{i_{k+1}i_{k}}A_{k+1}^{j_{k+1}j_{k}}$ is non-zero, and identifying their collisions. We term this \emph{collision analysis}. In order to state a resulting expression for the r.h.s.~of \eqnref{eq:tensor product formula}, we need to introduce the following notations.

For all $i\in[N]$ and $k\in [m]$,
\begin{eqnarray}
\paths{\bba}{} &\defeq& \Big\{(j_0,\ldots,j_m) \in [N]^{m+1} : \textstyle\prod_{k=0}^{m-1}A_{k+1}^{j_{k+1}j_{k}} \neq 0\Big\},\label{eq:def paths}\\
\pathsFrom{i}{\bba} &\defeq& \Big\{(j_0,\ldots,j_m)\in\paths{\bba}{2}:j_0 = i\Big\},\label{eq:paths from}\\
\parentSet{k}{i}{\bba}{4} &\defeq& \Big\{j\in[N]: A_k^{ij} \neq 0\Big\}, \label{eq:parent columns}\\
\primeParentSet{k}{i}{\bba}{m} &\defeq& \Big\{j\in [N] : \left(\textstyle{\prod_{q=0}^{k-1}} A_{k-q}\right)^{ij} > 0\Big\},  \label{eq:prime parents}\\
\collisionStartSetSh{k}{i}{\bba} &\defeq& \collisionStartSet{k}{i}{\bba},\quad\text{where}\quad  \primeParentSet{0}{i}{\bba}{m} \defeq \{i\}.\label{eq:collision start sets}
\end{eqnarray}
To interpret these sets, consider a directed graph $\graph{\bba}{} \defeq (\vertexset{\bba}{},\edgeset{\bba}{})$ with vertices and edges defined by
\begin{eqnarray}
\vertexset{\bba}{} &\defeq& \{\xi_{k}^i: 0\leq k \leq m, i \in [N]\},\label{eq:def vertex set}\\
\edgeset{\bba}{} &\defeq& \{(\xi_{k-1}^j,\xi_k^{i}): A_k^{ij}\neq 0, k\in[m], i,j\in[N]\},\label{eq:def edge set}
\end{eqnarray}
respectively. Suppose that the graph is arranged in the form of an array where $\xi^i_k$ is the vertex on the $k$th row and $i$th column, as shown in \figref{fig:set illustration}. In this case, $\paths{\bba}{}$ denotes the set of all paths in the graph starting from the top row and ending at the bottom row, and $\pathsFrom{i}{\bba}$ is this set restricted to those paths starting from $\xi^{i}_0$. Sets $\parentSet{k}{i}{\bba}{}$ determine the column indices of the parents of $\xi^i_k$
and sets $\primeParentSet{k}{i}{\bba}{m}$ determine the column indices of those vertices on the first row from which there exists a path to the vertex $\xi^i_k$. An illustration of these definitions is given in \figref{fig:set illustration}.

\begin{figure}
\mbox{
\begin{tabular}{@{}c@{}ccc@{}}
\mbox{
\begin{minipage}{.03\textwidth}
\begin{flushright}
\input{Simple0.tex} 
\end{flushright}
\end{minipage} 
}
& 
\mbox{
\begin{minipage}{.26\textwidth}
\begin{center}

%

\tikzset{
    vertex/.style = {
    	draw,
	    circle,
        fill      = white,
        outer sep = 2pt,
        inner sep = 2pt,
    }
}

\tikzset{
    mainvertex/.style = {
    	draw,
		circle,
        fill = black,
        outer sep = 2pt,
        inner sep = 2pt,
    }
}

\tikzset{
    demovertex/.style = {
    	draw,
		circle,
        fill = gray,
        outer sep = 2pt,
        inner sep = 2pt,
    }
}

\mbox{
\trimbox{5mm 1mm 4mm 1mm}{
\begin{tikzpicture}

\def \vstep {-1cm}
\def \hstep {1cm}
\def \r {2}
\def \m {3} 
\def \N {8}
\def \X {4}
\def \Y {2}

\foreach \i in {0,...,2}
{
	\foreach \j in {1,...,4}
	{
  		\node[vertex] (\i\j) at (\j*\hstep,\i*\vstep) {};
	}
}	

\foreach \j in {1,...,4}
{
	\node[circle] at (\j*\hstep,0*\vstep+.5cm) {{\small$\j$}};
}


\node[mainvertex] at (4*\hstep,2*\vstep) {$$};
\node[demovertex] at (2*\hstep,1*\vstep) {$$};
\node[demovertex] at (4*\hstep,1*\vstep) {$$};

\foreach \i in {0,...,1}
{
	\foreach \j in {1,...,4}
	{	
		\foreach \k in {0,...,1}
		{
			\draw[-stealthnew,shorten <=.15cm, shorten >=.15cm,arrowhead=1.5mm]
			let 
				\n1 = {int(mod((\j-1),\r^(\i))+\k*\r^(\i)+\r^(\i+1)*int(floor((\j-1)/\r^(\i+1))))+1}, 
				\n2 = {int(\i+1)} 
			in 
				(\j*\hstep,\i*\vstep) -- (\n1*\hstep,\n2*\vstep);
		}
	}
}	

\draw[-stealthnew,shorten <=.15cm, shorten >=.15cm,arrowhead=3mm,line width=2pt] (2*\hstep,1*\vstep) -- (4*\hstep,2*\vstep);
\draw[-stealthnew,shorten <=.15cm, shorten >=.15cm,arrowhead=3mm,line width=2pt] (4*\hstep,1*\vstep) -- (4*\hstep,2*\vstep);

\end{tikzpicture}
}} 
\end{center}
\end{minipage} 
}
& 
\mbox{
\begin{minipage}{.26\textwidth}
\begin{center}

%

\tikzset{
    vertex/.style = {
    	draw,
	    circle,
        fill      = white,
        outer sep = 2pt,
        inner sep = 2pt,
    }
}

\tikzset{
    mainvertex/.style = {
    	draw,
		circle,
        fill = black,
        outer sep = 2pt,
        inner sep = 2pt,
    }
}

\tikzset{
    demovertex/.style = {
    	draw,
		circle,
        fill = gray,
        outer sep = 2pt,
        inner sep = 2pt,
    }
}

\mbox{
\trimbox{5mm 1mm 4mm 1mm}{
\begin{tikzpicture}

\def \vstep {-1cm}
\def \hstep {1cm}
\def \r {2}
\def \m {3} 
\def \N {8}
\def \X {4}
\def \Y {2}

\foreach \i in {0,...,2}
{
	\foreach \j in {1,...,4}
	{
  		\node[vertex] (\i\j) at (\j*\hstep,\i*\vstep) {};
	}
}	

\foreach \j in {1,...,4}
{
	\node[circle] at (\j*\hstep,0*\vstep+.5cm) {{\small$\j$}};
}


\node[demovertex] at (1*\hstep,0*\vstep) {$$};
\node[demovertex] at (2*\hstep,0*\vstep) {$$};
\node[demovertex] at (3*\hstep,0*\vstep) {$$};
\node[demovertex] at (4*\hstep,0*\vstep) {$$};
\node[mainvertex] at (4*\hstep,2*\vstep) {$$};

\foreach \i in {0,...,1}
{
	\foreach \j in {1,...,4}
	{	
		\foreach \k in {0,...,1}
		{
			\draw[-stealthnew,shorten <=.15cm, shorten >=.15cm,arrowhead=1.5mm]
			let 
				\n1 = {int(mod((\j-1),\r^(\i))+\k*\r^(\i)+\r^(\i+1)*int(floor((\j-1)/\r^(\i+1))))+1}, 
				\n2 = {int(\i+1)} 
			in 
				(\j*\hstep,\i*\vstep) -- (\n1*\hstep,\n2*\vstep);
		}
	}
}	

\draw[-stealthnew,shorten <=.15cm, shorten >=.15cm,arrowhead=3mm,line width=2pt] (\hstep,0*\vstep) -- (2*\hstep,1*\vstep);
\draw[-stealthnew,shorten <=.15cm, shorten >=.15cm,arrowhead=3mm,line width=2pt] (2*\hstep,0*\vstep) -- (2*\hstep,1*\vstep);
\draw[-stealthnew,shorten <=.15cm, shorten >=.15cm,arrowhead=3mm,line width=2pt] (3*\hstep,0*\vstep) -- (4*\hstep,1*\vstep);
\draw[-stealthnew,shorten <=.15cm, shorten >=.15cm,arrowhead=3mm,line width=2pt] (4*\hstep,0*\vstep) -- (4*\hstep,1*\vstep);
\draw[-stealthnew,shorten <=.15cm, shorten >=.15cm,arrowhead=3mm,line width=2pt] (2*\hstep,\vstep) -- (4*\hstep,2*\vstep);
\draw[-stealthnew,shorten <=.15cm, shorten >=.15cm,arrowhead=3mm,line width=2pt] (4*\hstep,\vstep) -- (4*\hstep,2*\vstep);

\end{tikzpicture}
}} 
\end{center}
\end{minipage} 
}
& 
\mbox{
\begin{minipage}{.28\textwidth}
\begin{center}

%

\tikzset{
    vertex/.style = {
    	draw,
	    circle,
        fill      = white,
        outer sep = 2pt,
        inner sep = 2pt,
    }
}

\tikzset{
    mainvertex/.style = {
    	draw,
		circle,
        fill = black,
        outer sep = 2pt,
        inner sep = 2pt,
    }
}

\tikzset{
    demovertex/.style = {
    	draw,
		circle,
        fill = gray,
        outer sep = 2pt,
        inner sep = 2pt,
    }
}

\tikzset{
    demoedge/.style = {
    -stealthnew,
    shorten <=.15cm, 
    shorten >=.15cm,
    arrowhead=1.5mm,
    line width=1.2pt,
    }
}

\mbox{
\trimbox{5mm 1mm 4mm 1mm}{
\begin{tikzpicture}

\def \vstep {-.667cm}
\def \hstep {.5cm}
\def \r {2}
\def \m {3} 
\def \N {8}
\def \X {4}
\def \Y {2}

\foreach \i in {0,...,3}
{
	\foreach \j in {1,...,8}
	{
  		\node[vertex] (\i\j) at (\j*\hstep,\i*\vstep) {};
	}
}	

\foreach \j in {1,...,8}
{
	\node[circle] at (\j*\hstep,0*\vstep+.5cm) {{\small$\j$}};
}


\node[demovertex] at (5*\hstep,3*\vstep) {$$};

\node[demovertex] at (5*\hstep,2*\vstep) {$$};
\node[demovertex] at (1*\hstep,2*\vstep) {$$};


\draw[demoedge] (5*\hstep,2*\vstep) -- (5*\hstep,3*\vstep) ;
\draw[demoedge] (1*\hstep,2*\vstep) -- (5*\hstep,3*\vstep) ;

\draw[demoedge] (5*\hstep,1*\vstep) -- (5*\hstep,2*\vstep) ;
\draw[demoedge] (7*\hstep,1*\vstep) -- (5*\hstep,2*\vstep) ;

\node[demovertex] at (1*\hstep,1*\vstep) {$$};
\node[demovertex] at (3*\hstep,1*\vstep) {$$};

\draw[demoedge] (1*\hstep,1*\vstep) -- (1*\hstep,2*\vstep) ;
\draw[demoedge] (3*\hstep,1*\vstep) -- (1*\hstep,2*\vstep) ;

\node[demovertex] at (5*\hstep,1*\vstep) {$$};
\node[demovertex] at (7*\hstep,1*\vstep) {$$};

\node[demovertex] at (1*\hstep,0*\vstep) {$$};
\node[demovertex] at (2*\hstep,0*\vstep) {$$};
\node[demovertex] at (3*\hstep,0*\vstep) {$$};
\node[demovertex] at (4*\hstep,0*\vstep) {$$};

\draw[demoedge] (1*\hstep,0*\vstep) -- (1*\hstep,1*\vstep) ;
\draw[demoedge] (2*\hstep,0*\vstep) -- (1*\hstep,1*\vstep) ;
\draw[demoedge] (3*\hstep,0*\vstep) -- (3*\hstep,1*\vstep) ;
\draw[demoedge] (4*\hstep,0*\vstep) -- (3*\hstep,1*\vstep) ;

\node[demovertex] at (6*\hstep,0*\vstep) {$$};
\node[demovertex] at (7*\hstep,0*\vstep) {$$};
\node[demovertex] at (8*\hstep,0*\vstep) {$$};

\draw[demoedge] (6*\hstep,0*\vstep) -- (5*\hstep,1*\vstep) ;
\draw[demoedge] (7*\hstep,0*\vstep) -- (7*\hstep,1*\vstep) ;
\draw[demoedge] (8*\hstep,0*\vstep) -- (7*\hstep,1*\vstep) ;


\foreach \i in {0,...,2}
{
	\foreach \j in {1,...,8}
	{	
		\foreach \k in {0,...,1}
		{
			\draw[-stealthnew,shorten <=.15cm, shorten >=.15cm,arrowhead=1mm]
			let 
				\n1 = {int(mod((\j-1),\r^(\i))+\k*\r^(\i)+\r^(\i+1)*int(floor((\j-1)/\r^(\i+1))))+1}, 
				\n2 = {int(\i+1)} 
			in 
				(\j*\hstep,\i*\vstep) -- (\n1*\hstep,\n2*\vstep);
		}
	}
}	

\node[draw,rectangle,minimum width=.45cm,minimum height= .4cm,rounded corners=3pt] at (6*\hstep,0*\vstep) {$$};
\node[draw,rectangle,minimum width=.95cm,minimum height= .4cm,rounded corners=3pt] at (7.5*\hstep,0*\vstep) {$$};
\node[draw,rectangle,minimum width=1.95cm,minimum height=.4cm,rounded corners=3pt] at (2.5*\hstep,0*\vstep) {$$};


\end{tikzpicture}
}} 
\end{center}
\end{minipage} 
} \\
\\
&$\parentSet{2}{4}{\bba}{2} = \{2,4\}$ & $\primeParentSet{2}{4}{\bba}{2} = \{1,2,3,4\}$ & 
\begin{minipage}[t]{.23\textwidth} 
$\collisionStartSetSh{1}{5}{\bba} = \{6\}$

$\collisionStartSetSh{2}{5}{\bba} = \{7,8\}$

$\collisionStartSetSh{3}{5}{\bba} = \{1,2,3,4\}$
\vspace{.2cm}
\end{minipage}\\
&(a) & (b) & (c) 
\end{tabular}
}
\caption{(a) The column indices of the parents (gray) of $\xi^{4}_{2}$ (black) constitute the set $\parentSet{2}{4}{\bba}{}$.
(b) The column indices of the ancestors (gray) of $\xi^{4}_{2}$ (black) on the first row constitute the set $\primeParentSet{2}{4}{\bba}{}$.
(c) The vertices whose column indices constitute the sets $\collisionStartSetSh{1}{5}{\bba}$, $\collisionStartSetSh{2}{5}{\bba}$ and $\collisionStartSetSh{3}{5}{\bba}$ are highlighted by rectangles.}
\label{fig:set illustration}
\end{figure}

The following assumption shall be invoked in \propref{prop:tensor product formulation}. It serves to impose some structure which is common to the matrices which define radix-$r$ and the mixed radix-$r$ butterfly resampling algorithms. For fixed $N,m\geq 1$ and for any sequence $\bba = \aMatrices{}{m}$, 
we write $\bba_{p:q} \defeq (A_{k})_{k=p}^q$ where $0<p\leq q\leq m$.
\begin{assumption}\label{ass:A_k extra}\label{ass:cardinality invariance}
For given $N,m\geq 1$, the matrices $\bba = \bba^{(N,m)}$ satisfy \assref{ass:A_k}, and, in addition, one has for all $p,q \in [m]$ and $i,j\in [N]$
\begin{enumerate}[itemsep=0pt, topsep=5pt, partopsep=0pt,label={{(\roman*)}}]
\item \label{it:symmetry} Symmetry: $A_p^{ij} = A_p^{ji}$,\phantom{$\big|\parentSet{p}{i}{\bba}{}\big|=\big|\parentSet{p}{j}{\bba}{}\big|$}
\item \label{it:commutativity} Commutativity: $A_qA_p = A_pA_q$,\phantom{$\big|\parentSet{p}{i}{\bba}{}\big|=\big|\parentSet{p}{j}{\bba}{}\big|$}
\item \label{it:idempotence} Idempotence: $A_pA_p = A_p$,\phantom{$\big|\parentSet{p}{i}{\bba}{}\big|=\big|\parentSet{p}{j}{\bba}{}\big|$}
\item \label{it:cardinality invariance} Equal number of non-zero elements: $\big|\parentSet{p}{i}{\bba}{}\big|=\big|\parentSet{p}{j}{\bba}{}\big|$,
\item \label{it:unique paths} For all $i_p,j_p\in[N]$ and  $((i_p,\ldots,i_q),(j_p,\ldots,j_q))\in \pathsFrom{i_p}{\bba_{p+1:q}}{}\times \pathsFrom{j_p}{\bba_{p+1:q}}{}$, where $0\leq p < q \leq m$ and $(i_p,i_q) = (j_p,j_q)$, one has $(i_p,\ldots,i_q) = (j_p,\ldots,j_q)$.
\end{enumerate}
\end{assumption}
\begin{remark}
The conditions \ref{it:symmetry}--\ref{it:idempotence} are standard matrix properties. Condition \ref{it:cardinality invariance} states that each row in each element of $\bba$ has the same number of non-zero elements. Property \ref{it:unique paths} states that given any two vertices of the graph $\graph{\bba}{}$ with column indices $i_p$ and $i_q$, there exists at most one directed path between those  vertices. This condition is closely related to the existence of unique paths between any vertices in an undirected tree graph (see, e.g.~\cite{lauritzen96}).
\end{remark}
We are then ready to state the second main result on the conditional second moment, whose proof if given in \secref{sec:proofs for collision analysis} of the \ref{suppA}.
%
%
\begin{proposition}
\label{prop:tensor product formulation}
Fix $N,m\geq 1$ and $\varphi\in \boundMeas{\ss}$. If $\bba = \bba^{(N,m)}$ satisfies \assref{ass:A_k extra}, then
\begin{align*}
&\dfrac{m}{N}\E\Bigg[\Bigg(\displaystyle\sum_{\varrho \in [Nm]} X^{(N,m)}_\varrho\Bigg)^2 \Bigg| \calF^{(N,m)}_0\Bigg]
= \dfrac{1}{N^2}\displaystyle\sum_{i} g^{2}(\bxi{i}{0}{}{})\cvarphi^2_{N}(\bxi{i}{0}{}{})\\
&\rule{.3\textwidth}{0pt}+ \dfrac{1}{N^4}\displaystyle\sum_{i}\sum_{j\neq i}g(\bxi{i}{0}{}{})\cvarphi^2_{N}(\bxi{i}{0}{}{})g(\bxi{j}{0}{}{})\specCollisionPairs{i}{j}{\bba}{4}\\
&\rule{.3\textwidth}{0pt}+ \dfrac{1}{N^4}\displaystyle\sum_{i}\sum_{j\neq i}g(\bxi{i}{0}{}{})\cvarphi_{N}(\bxi{i}{0}{}{})g(\bxi{j}{0}{}{})\cvarphi_{N}(\bxi{j}{0}{}{})\specNonCollisionPairsFrom{\bba}{m}{i}{j},
\end{align*}
where 
for all $i,j\in [N]$ such that $i\neq j$
\begin{align*}
\specCollisionPairs{i}{j}{\bba}{}
&=
\sum_{k=1}^m \Big|\lowerPartFrom{k}{u_{0}}{\bba}\Big|^2\ind{}\Big(j \in \collisionStartSetSh{k}{i}{\bba}\Big)\Big|\primeParentSet{k}{i}{\bba}{}\Big|, \\
\specNonCollisionPairsFrom{\bba}{m}{i}{j} &= \sum_{k=1}^m \Big(N^2 - \Big|\lowerPartFrom{k}{u_{0}}{\bba}\Big|^2\Big|{\primeParentSet{k}{i}{}{}}\Big|\Big)\ind{}\Big(j \in \collisionStartSetSh{k}{i}{\bba}\Big),
\end{align*}
where $u_{0} \in [N]$, and for all $0\leq k < m$ and $i\in[N]$, $\lowerPartFrom{k}{i}{\bba} \defeq \pathsFrom{i}{\bba_{k+1:m}}$ and $\lowerPartFrom{m}{i}{\bba} \defeq \{i\}$.
\end{proposition}

It is now apparent that in order to study the asymptotic behaviour of the conditional second moment,  one needs to study the quantities $\specCollisionPairs{i}{j}{\bba}{4}$ and
$\specNonCollisionPairsFrom{\bba}{m}{i}{j}$. This involves detailed combinatorial analysis, specific to each of the two butterfly resampling schemes.

\section{Analysis part II - LLN and CLT for butterfly resampling algorithms}
\label{sec:convergence analysis part 2}

The next step towards proving the LLN and CLT for particle filters deploying the butterfly resampling, is to prove the corresponding results for a single application of butterfly resampling.


\subsection{Radix-\texorpdfstring{$r$}{r} algorithm}
\label{sec:lln and clt fixed radix}

Throughout \secref{sec:lln and clt fixed radix}, $r\geq2$ is a fixed integer and for each $m\geq 1$ we assume $\bba^{(r^m,m)} = \radixMatrices{r}{m}$ as defined in \eqnref{eq:def fixed radix matrices}.

\begin{theorem}\label{thm:butterfly_lln_fixed_radix}
No matter what the distribution of the input random variables $(\xi_{\mathrm{in}}^i)_{i\in[r^m]}$ is, for any $\varphi\in\boundMeas{\ss}$,
\begin{equation*}
\sqrt{\frac{m}{r^m}}\sum_{\varrho\in[r^m m]}X_\varrho^{(r^m,m)}\almostsurely{m\rightarrow\infty}{\P} 0.
\end{equation*}
\end{theorem}
\begin{proof}
By \lemmaref{lem:modular congruence realtion fixed}, $\radixMatrices{r}{m}$ satisfies \assref{ass:A_k} and hence we can apply \propref{prop:intro_to_aug_resampling} to give, for any $p\geq1$,
\begin{equation*}
\E\left[\left|\sqrt{\frac{m}{r^m}}\sum_{\varrho\in[r^m m]}X_\varrho^{(r^m,m)}\right|^{p}\right]\leq b_p \left(\frac{m}{r^m}\right)^{p/2} \|g\|_{\infty}^p\text{osc}(\varphi)^p,
\end{equation*}
and the claim then follows from the Borel-Cantelli lemma.
\end{proof}

We note that the following result has as a hypothesis a bound on errors associated with certain subsets of the input random variables $(\xi_{\mathrm{in}}^i)_{i\in[N]}$, which is unusual compared to similar results for multinomial resampling, e.g. \citep{smc:the:C04}.
\begin{theorem}\label{thm:butterfly_clt_fixed_radix}
If for all $\varphi\in\boundMeas{\ss}$
there exists $b(\varphi)\in\real$ such that for some $\mu\in\pmeasure{\ss}$, and for all $m\geq 1$, $\ka \in [m]$ and $q\in [r^{\ka-1}]$
\begin{equation}\label{eq:required conv in prob}
\E\left[\abs{\frac{1}{r^{m-\ka+1}}\sum_{i\in [r^{m-\ka+1}]} \varphi(\xi^{J(i)}_\mathrm{in}) - \mu(\varphi)}^{2}\right]^{\frac{1}{2}}
\leq b(\varphi)\sqrt{\frac{m-\ka}{r^m} + \frac{1}{r^{m-\ka+1}}},
\end{equation}
where $J(i) \defeq i+(q-1)r^{m-\ka+1}$ for all $i \in [r^{m-\ka+1}]$, then for any $\varphi\in\boundMeas{\ss}$ and any $u\in\real$,
\begin{equation*}
\E\Bigg[\exp\Bigg(iu\sum_{\varrho\in[r^m m]}X_\varrho^{(r^m,m)}\Bigg)\Bigg|\F_0^{(r^m,m)}\Bigg]\inprob{m\rightarrow\infty}\exp\left(-(u^2/2)\sigma^2(\varphi)\right),
\end{equation*}
where
\begin{equation*}
\sigma^2(\varphi)=(1-r^{-1})\mu\left(g\left(\varphi-\frac{\mu(g\varphi)}{\mu(g)}\right)^2\right)\mu(g).
\end{equation*}
\end{theorem}
\begin{proof}
In order to apply \theref{thm:Douc and Moulines}, by the discussion in \secref{sec:martingale array mapping}, we need to verify conditions \eqref{eq:douc_cond_mart}-\eqref{eq:douc_cond_var}. Condition \eqref{eq:douc_cond_mart} holds immediately by \propref{prop:generalized martingale decomposition}\ref{eq:zero expectations}. To check \eqref{eq:douc_cond_neg}, we have by \eqref{eq:boundedness of X} that
\begin{eqnarray*}
&&\sum_{\varrho\in[r^m m]} \E \bigg[ \Big(X_\varrho^{(r^m,m)}\Big)^2 \mathbb{I}\Big\{\Big|X_\varrho^{(r^m,m)}\Big|\geq \epsilon\Big\} \,\Big|\, \F_{\varrho-1}^{(r^m,m)} \bigg]\nonumber\\
&&\qquad\leq \infnorm{g}^2\osc{\varphi}^2 \ind{}\bigg(\frac{\infnorm{g}\osc{\varphi}}{\sqrt{m r^m}} \geq \epsilon\bigg) \surely{m\rightarrow\infty} 0.
\end{eqnarray*}
It remains to verify \eqref{eq:douc_cond_var}, i.e.,
\begin{equation}
\sum_{\varrho\in[r^m m]} \E \left[ \left.  \left(X_\varrho^{(r^m,m)}\right)^2 \right| \F_{\varrho-1}^{(r^m,m)} \right] \inprob{m\rightarrow\infty} \sigma^2(\varphi)\text{, for some}~\sigma^2(\varphi)>0.\label{eq:clt_fixed_radix_cond_var_conv}
\end{equation}
To do this, we first use \propref{prop:generalized martingale decomposition}\ref{eq:zero expectations} and the tower property of conditional expectations to obtain the decomposition:
\begin{align}
&\sum_{\varrho\in[r^m m]} \E \Bigg[  \Bigg(X_\varrho^{(r^m,m)}\Bigg)^2 \Bigg| \F_{\varrho-1}^{(r^m,m)} \Bigg] \nonumber\\
&= \E \Bigg[ \Bigg(\sum_{\varrho\in[r^m m]} X_\varrho^{(r^m,m)}\Bigg)^2\Bigg|\F_0^{(r^m,m)}\Bigg]+\sum_{\varrho\in[r^m m]}Z_\varrho^{(r^m,m)},\label{eq:2nd moment part}
\end{align}
where
\begin{equation}\label{eq:def Z}
Z_\varrho^{(r^m,m)}\defeq\E \bigg[ \bigg(X_\varrho^{(r^m,m)}\bigg)^2 \bigg| \F_{\varrho-1}^{(r^m,m)} \bigg]- \E \bigg[ \bigg(X_\varrho^{(r^m,m)}\bigg)^2\bigg|\F_0^{(r^m,m)}\bigg].
\end{equation}
By \propref{prop:fixed radix satisfies assumptions} in \secref{sec:conditional independece structure fixed radix} of the \ref{suppA}, $\radixMatrices{r}{m}$ satisfies \assref{ass:A_k extra} and hence Propositions \ref{prop:tens_prod_formula}-\ref{prop:tensor product formulation} together with the hypothesis \eqnref{eq:required conv in prob} can be used to establish Proposition \ref{prop:var_conv_fixed_radix} in \secref{sub:collision_analysis_fixed_radix} of the \ref{suppA}, from which it follows that
  \[
  \E \Bigg[ \Bigg(\sum_{\varrho\in[r^m m]} X_\varrho^{(r^m,m)}\Bigg)^2\Bigg|\F_0^{(r^m,m)}\Bigg] \inprob{m\rightarrow\infty}\sigma^2(\varphi).
  \]
   Proposition \ref{prop:var_vanish_fixed_radix}, in \secref{sub:independence_analysis_fixed_radix} of the \ref{suppA}, shows that $Z_\varrho^{(r^m,m)}\inprob{m\rightarrow\infty} 0$. This establishes \eqnref{eq:clt_fixed_radix_cond_var_conv} and the proof of the theorem is complete.
\end{proof}

\subsection{Mixed radix-\texorpdfstring{$r$}{r} algorithm}
\label{sec:lln and clt mixed radix}

For the mixed radix-$r$ algorithm, we fix, throughout \secref{sec:lln and clt mixed radix}, $m=2$ and $r\geq 2$, and for all $c\geq 1$ we assume $\bba^{(rc,2)} = \mradixMatrices{r}{c}$ as defined in \eqnref{eq:mradix matrix definition}. Analogous to \theref{thm:butterfly_lln_fixed_radix}, we have by \propref{prop:intro_to_aug_resampling}:

\begin{theorem}\label{thm:butterfly_lln_mixed_radix}
No matter what the distribution of the input random variables $(\xi_{\mathrm{in}}^i)_{i\in[rc]}$ is, for any $\varphi\in\boundMeas{\ss}$,
\begin{equation*}
\sqrt{\frac{2}{rc}}\sum_{\varrho\in[2rc]}X_\varrho^{(rc,2)}\almostsurely{c\rightarrow\infty}{\P} 0.
\end{equation*}
\end{theorem}
\begin{proof}
Similar to the proof of \theref{thm:butterfly_lln_fixed_radix}.
\end{proof}
Similarly as in the case of the radix-$r$ algorithm, a hypothesis on the errors associated with certain sub-populations of the input random variables plays a role in the CLT for the mixed radix-$r$ algorithm. 
\begin{theorem}\label{thm:butterfly_clt_mixed_radix}
If for all $\varphi\in\boundMeas{\ss}$
there exists $b(\varphi)\in\real$ such that for some $\mu\in\pmeasure{\ss}$, and for all $c\geq 1$, $\ka \in \{1,2\}$ and $q \in [r^{\ka-1}]$
\begin{equation}\label{eq:required conv in prob mixed}
\E\left[\abs{\frac{r^{\ka-1}}{rc}\sum_{i\in[cr^{2-\ka}]}\varphi(\xi^{J(i)}_0) - \mu(\varphi)}^{2}\right]^{\frac{1}{2}}
\leq b(\varphi)\sqrt{\frac{2-\ka}{rc} + \frac{r^{\ka-1}}{rc}}.
\end{equation}
where $J(i) = i + (q-1)cr^{2-\ka}$ for all $i \in [cr^{2-\ka}]$,
then for any $\varphi\in\boundMeas{\ss}$ and any $u\in\real$,
\begin{equation*}
\E\left[\left.\exp\left(iu\sum_{\varrho\in[2rc]}\mrbfX{\varrho}{rc}{2}{\varphi}\right)\right|\mrbfF{0}{rc}{2}\right]\inprob{c\to\infty}\exp\left(-(u^2/2)\sigma^2(\varphi)\right),
\end{equation*}
where
\begin{equation*}
\sigma^2(\varphi)=\Bigg(1-\frac{1}{2r}\Bigg)\mu\Bigg(g\bigg(\varphi-\frac{\mu(g\varphi)}{\mu(g)}\bigg)^2\Bigg)\mu(g).
\end{equation*}
\end{theorem}
\begin{proof}

The proof is similar to that of \theref{thm:butterfly_clt_fixed_radix}, with the exceptions that we use \propref{prop:mixed radix satisfies assumptions} in \secref{sec:conditional independece structure mixed radix} of the \ref{suppA} instead of \propref{prop:fixed radix satisfies assumptions}, \propref{prop:var_conv_mixed_radix} of \secref{sec:convergence of the conditional variance mixed} in the \ref{suppA} instead of \propref{prop:var_conv_fixed_radix}, and \mbox{\propref{prop:var_vanish_mixed_radix}} of \secref{sec:independence analysis mixed} in the \ref{suppA} instead of \propref{prop:var_vanish_fixed_radix}. Also, the hypothesis \eqnref{eq:required conv in prob mixed} as well as Propositions \ref{prop:tens_prod_formula} and \ref{prop:tensor product formulation} are needed in the proof of \propref{prop:var_conv_mixed_radix}.
\end{proof}

\section{Analysis part III - particle filters}
\label{sec:convergence of particle filters}

\newcommand{\amC}{b_n(\varphi,p)}
\newcommand{\amChat}{\hat{b}_n(\varphi,p)}

Finally, we address the proofs of the two main results of the paper, Theorems \ref{thm:radix-r} and \ref{thm:mixed radix-r}. To extend the results of \secref{sec:convergence analysis part 2} to the particle filter we need to ensure that the hypotheses of Theorems \ref{thm:butterfly_clt_fixed_radix} and \ref{thm:butterfly_clt_mixed_radix} are valid and that their validity is preserved throughout the filtering sequence. The next result, when applied with appropriate $\bba^{(N,m)}$, $\calI$, $d$ and $J$, allows us to do this, by quantifying the errors associated with certain sub-populations of the particle system.



\begin{proposition}\label{prop:subsample absolute moment convergence}
Fix $N,m \geq 1$, $\ka\in[m]$ and a partition $\calI$, and let $(\zeta_n^i,\hat{\zeta}_n^i)_{n\geq 0,i\in[r^m]}$ be the random variables associated with the augmented resampling particle filter deploying matrices $\bba^{(N,m)}$. If the triple $(\bba^{(N,m)},\calI,d)$ satisfies \assref{ass:partition and cond indep}, then for all $n\geq 0$, $p>1$, $\varphi \in \boundMeas{\ss}$ there exist $\amC,\amChat \in \real{}$, depending only on $n$, $p$ and $\varphi$, such that for all $J \in \partitionMap{\calI}$
\begin{equation*}
\E\left[\abs{\frac{1}{\card{\calI}}\sum_{i\in[\card{\calI}]} \varphi(\zeta^{J(i)}_n) - \pi_n(\varphi)}^p\right]^{\frac{1}{p}}
\leq \amC\sqrt{\frac{m-\ka}{N} + \frac{1}{\card{\calI}}},
\end{equation*}
and
\begin{equation*}
\E\left[\abs{\frac{1}{\card{\calI}}\sum_{i\in[\card{\calI}]} \varphi(\hat{\zeta}^{J(i)}_n) - \hat{\pi}_n(\varphi)}^p\right]^{\frac{1}{p}}
\leq \amChat\sqrt{\frac{m-\ka}{N} + \frac{1}{\card{\calI}}}.
\end{equation*}
\end{proposition}
The strategy of the proof is an induction, showing that if the first bound holds for the entire population given as input to the augmented resampling algorithm, then the second bound holds for certain blocks in the output of the resampling, including the entire population, and moreover that this bound is preserved in the mutation step of the particle filter. The proof is given in \secref{sec:subsample convergence supp} of the \ref{suppA}.


  The steps required to complete the proofs of Theorems \ref{thm:radix-r} and \ref{thm:mixed radix-r}, in outline, follow those of \cite{smc:the:C04}. An inductive argument is used to show that the LLN and CLT are preserved at each time step. Although some of the scaling in the CLT's is unusual, the proof techniques are standard and so the proofs are given in \secref{sec:subsample convergence supp} of the \ref{suppA}. 

\begin{supplement}[id=suppA]
  \sname{Supplement}
  \stitle{``Butterfly resampling: asymptotics for particle filters with constrained interactions''}
  \slink[url]{}
\end{supplement}


\newpage

\addtocounter{partcount}{1}

\renewcommand{\theequation}{S.\arabic{equation}}
\renewcommand{\thesection}{\Alph{section}}
\renewcommand{\thepage}{~\arabic{page}~}

\setcounter{equation}{0}
\setcounter{figure}{0}
\setcounter{page}{1}
\setcounter{section}{0}
\begin{frontmatter}

\title{Supplementary material for ``Butterfly resampling: asymptotics for particle filters with constrained interactions''}
\runtitle{Butterfly resampling (suppl.)}

\begin{aug}
\author{\fnms{Kari} \snm{Heine}
},
\author{\fnms{Nick} \snm{Whiteley}
},
\author{\fnms{A.~Taylan} \snm{Cemgil}
}
\and
\author{\fnms{Hakan} \snm{G\"ulda{\c s}}
}



\address{Department of Mathematics\\
University of Bristol\\
University Walk \\
Bristol \\
BS8 1TW \\
\printead{e3}\\
\phantom{E-mail:\ }\printead*{e4}}

\address{Department of Computer Engineering\\
Bo{\u g}azi{\c c}i University\\
34342 Bebek\\
Istanbul\\
\printead{e1}\\
\phantom{E-mail:\ }\printead*{e2}}

\runauthor{K.~Heine et al.}
\end{aug}

\end{frontmatter}

\section{Proofs for Sections \ref{subsec:distributions} and \ref{subsec:martingale}}
\label{sec:preparatory results}

\begin{proof}[Proof of \lemmaref{lem:conditional marginal}]
By \assref{ass:A_k} and the definition of $\bxi{i}{k}{}{}$ in \algrefmy{alg:augmented_resampling} we can assume that for all $k\in [m]$ and $i \in [N]$
\begin{equation}\label{eq:rigorus formulation of xi}
\bxi{i}{k}{}{}= \bxi{I_k^i}{k-1}{}{},\text{ where } I_k^i\thicksim \frac{1}{V^i_k}\sum_j A_k^{ij}V_{k-1}^j\delta_{j},
\end{equation}
and the $I^i_k$ are independent given $\xi_0$. Let $\ell_0 = i$. By the law of total probability, conditional independence of $I_k^i$ and the one step conditional independence, for all $0\leq k \leq m-1$,
\begin{align}
& \P\left(\bxi{i}{m}{}{}\in S\mids\xi_0,\ldots,\xi_{m-k-1}\right) \nonumber\\
&= \sum_{(\ell_1,\ldots,\ell_{k})} \E\left[\Bigg(\prod_{q=0}^{k-1}\ind{}(I_{m-q}^{\ell_q}=\ell_{q+1})\Bigg) \ind{}(\bxi{\ell_{k}}{m-k}{}{} \in S)\mids \xi_0,\ldots,\xi_{m-k-1}\right]\nonumber\\
&= \sum_{\ell_k}\P\big(\bxi{\ell_{k}}{m-k}{}{} \in S\big| \xi_0,\ldots,\xi_{m-k-1}\big)\sum_{(\ell_1,\ldots,\ell_{k-1})} \prod_{q=0}^{k-1}\P\big(I_{m-q}^{\ell_q}=\ell_{q+1} \big| \xi_0\big).\label{eq:messy thing}
\end{align}
By \eqnref{eq:rigorus formulation of xi}, $\P\left(I_{m-q}^{\ell_q}=\ell_{q+1} \mids \xi_0\right) = {(V^{\ell_q}_{m-q})^{-1}} A_{m-q}^{\ell_q \ell_{q+1}}V^{\ell_{q+1}}_{m-q-1}$ yielding
\begin{align}
\sum_{(\ell_1,\ldots,\ell_{k-1})}\prod_{q=0}^{k-1}\P\left(I_{m-q}^{\ell_q}=\ell_{q+1} \mids \xi_0\right)
&= \frac{V^{\ell_{k}}_{m-k}}{V^{\ell_0}_m}\sum_{(\ell_1,\ldots,\ell_{k-1})}\Bigg(\prod_{q=0}^{k-1} A_{m-q}^{\ell_q \ell_{q+1}}\Bigg) \nonumber\\
&= \frac{V^{\ell_{k}}_{m-k}}{V^{\ell_0}_m}\Bigg(\prod_{q=0}^{k-1} A_{m-q}\Bigg)^{\ell_0\ell_k}.\label{eq:product}
\end{align}
Because by \eqnref{eq:rigorus formulation of xi} $$\P\big(\bxi{\ell_{k}}{m-k}{}{} \in S\,\big|\, \xi_0,\ldots,\xi_{m-k-1}\big) = {(V^{\ell_k}_{m-k})^{-1}}\sum_j A_{m-k}^{\ell_k j}V_{m-k-1}^{j}\ind{}(\xi^j_{m-k-1}\in S),$$ by substituting \eqnref{eq:product} into \eqnref{eq:messy thing} we have
\begin{align*}
&\P\big(\bxi{\ell_0}{m}{}{}\in S\,\big|\,\xi_0,\ldots,\xi_{m-k-1}\big) \\
&= \frac{1}{V^{\ell_0}_m}\sum_jV_{m-k-1}^{j}\ind{}(\xi^j_{m-k-1}\in S)\sum_{\ell_k} \Bigg(\prod_{q=0}^{k-1} A_{m-q}\Bigg)^{\ell_0\ell_k}A_{m-k}^{\ell_k j}.
\end{align*}
from which the claim follows by recalling that $\ell_0 = i$.
\end{proof}

\begin{proof}[Proof of \propref{prop:generalized martingale decomposition}]

For brevity, let us write $J_{i} \defeq J(i)$ for all $i\in[\card{\calI}]$. Then, by defining
\begin{align*}
A &\defeq \frac{1}{\card{\calI}}\sum_{i_{m}\in[\card{\calI}]} V^{J_{i_m}}_m\cvarphi_{N}(\bxi{J_{i_m}}{m}{}{})
 - \frac{1}{N}\sum_{i_{m-\ka}} V_{m-\ka}^{i_{m-\ka}}\cvarphi_{N}(\bxi{i_{m-\ka}}{m-\ka}{}{}), \\
B &\defeq \sum_{q \in [m-\ka]} \Bigg(
 \frac{1}{N}\sum_{i_q} V_{q}^{i_q}\cvarphi_{N}(\bxi{i_q}{q}{}{})  -
 \frac{1}{N}\sum_{i_{q-1}} V_{q-1}^{i_{q-1}}\cvarphi_{N}(\bxi{i_{q-1}}{q-1}{}{}) \Bigg),
\end{align*}
we have the telescoping decomposition
\begin{equation}\label{eq:martingale telescope}
\frac{1}{\card{\calI}}\sum_{i_{m}\in[\card{\calI}]} V^{J_{i_m}}_m\cvarphi_{N}(\bxi{J_{i_m}}{m}{}{}) = A + B.
\end{equation}
By \assref{ass:A_k}, the matrices $\aMatrices{N}{m}$ are doubly stochastic and by \ref{it:is partition} and \ref{it:equivalence classes} of \assref{ass:partition and cond indep}, we have by writing $\cvarphi_{m}^i\defeq \cvarphi_{N}(\bxi{i}{m}{}{})$ for brevity
\begin{align}
A
&= \frac{1}{\card{\calI}}\sum_{i_{m}=1}^{\card{\calI}} V^{J_{i_m}}_m\cvarphi^{J_{i_m}}_{m}
 - \frac{1}{N}\sum_{j}\sum_{i_{m-\ka}}  \Bigg(\prod_{q=0}^{\ka-1}A_{m-q}\Bigg)^{j i_{m-\ka}}V_{m-\ka}^{i_{m-\ka}}\cvarphi^{{i_{m-\ka}}}_{m-\ka}\nonumber\\
&= \frac{1}{\card{\calI}}\sum_{i_{m}=1}^{\card{\calI}} V^{J_{i_m}}_m\cvarphi^{J_{i_m}}_{m} \nonumber \\
&\quad-~\frac{1}{N}\sum_{j=1}^{\card{\calI}}\frac{N}{\card{\calI}}
\sum_{i_{m-\ka}}  \Bigg(\prod_{q=0}^{\ka-1}A_{m-q}\Bigg)^{J_{j} i_{m-\ka}}V_{m-\ka}^{i_{m-\ka}}\cvarphi^{{i_{m-\ka}}}_{m-\ka}\nonumber\\
&= \frac{1}{\card{\calI}}\sum_{i_{m}=1}^{\card{\calI}} V^{J_{i_m}}_m\label{eq:A part}\\
&\quad\times~\Bigg(\cvarphi^{J_{i_m}}_{m} - \frac{1}{V^{J_{i_m}}_{m}} \sum_{i_{m-\ka}} \Bigg(\prod_{q=0}^{\ka-1}A_{m-q}\Bigg)^{J_{i_m} i_{m-\ka}}V_{m-\ka}^{i_{m-\ka}}\cvarphi^{{i_{m-\ka}}}_{m-\ka}\Bigg).\nonumber
\end{align}
Similarly for $B$ we have
\begin{align}
B
&= \sum_{q \in [m-\ka]} \Bigg(
 \frac{1}{N}\sum_{i_q} V_{q}^{i_q}\cvarphi_{N}(\bxi{i_q}{q}{}{})  -
 \frac{1}{N}\sum_{j}\sum_{i_{q-1}}A_q^{j i_{q-1}} V_{q-1}^{i_{q-1}}\cvarphi_{N}(\bxi{i_{q-1}}{q-1}{}{})\Bigg)\nonumber\\
&=
\sum_{q\in[{m-\ka}]}\sum_{i_q}\frac{V_q^{i_q}}{N} \Bigg(\overline{\varphi}_{N}(\xi_q^{i_q})-\frac{1}{V^{i_q}_q}\sum_{i_{q-1}} A_q^{i_q i_{q-1}} V^{i_{q-1}}_{q-1}\overline{\varphi}_{N}(\xi_{q-1}^{i_{q-1}}) \Bigg),\label{eq:B part}
\end{align}
and by combining \eqnref{eq:martingale telescope}, \eqnref{eq:A part} and \eqnref{eq:B part} we have established \eqnref{it:martingale representation}. Using \eqnref{eq:V_defn_proofs} and \lemmaref{lem:facts_about_Vs}, we can establish \eqnref{it:another martingale representation}:
\begin{align*}
&\Bigg(\frac{1}{N}\sum_{i}g(\xi_{\text{in}}^i)\Bigg)\Bigg(\frac{1}{\card{\calI}}\sum_{i\in[\card{\calI}]}\varphi(\xi_{\text{out}}^{J_{i}})\Bigg)-\frac{1}{N}\sum_{i}g(\xi_{\text{in}}^i)\varphi(\xi_{\text{in}}^i)\\
&\qquad=\frac{1}{\card{\calI}}\sum_{i\in[\card{\calI}]} V_m^{J_{i}} \varphi(\xi_m^{J_{i}})-\frac{1}{N}\sum_i V_0^i \varphi(\xi_0^i)\\
&\qquad=\frac{1}{\card{\calI}}\sum_{i\in[\card{\calI}]} V_m^{J_{i}} \Bigg(\varphi(\xi_m^{J_{i}})-\frac{\frac{1}{N}\sum_j V_0^j \varphi(\xi_0^j)}{\frac{1}{N}\sum_j V_0^j }\Bigg)\\
&\qquad=\frac{1}{\card{\calI}}\sum_{i\in[\card{\calI}]} V_m^{J_{i}}\overline{\varphi}_N(\xi_m^{J_{i}}).
\end{align*}
Using \eqnref{eq:varphi_bar_defn}, \eqnref{eq:piecewise sigma fields}, \eqnref{eq:explicit X} and \lemmaref{lem:facts_about_Vs}\ref{it:measurability of V} we find that \ref{eq:measurability} holds.
For $\varrho \leq (m-\ka)N$, \ref{eq:zero expectations} follows from \eqref{eq:P_in_terms_of_V_proofs} and the one step conditional independence.
For $(m-\ka)N < \varrho \leq (m-\ka)N + \card{\calI}$, \ref{eq:zero expectations} follows from \lemmaref{lem:conditional marginal} and \assref{ass:partition and cond indep}\ref{it:cond indep}. Finally \eqnref{eq:boundedness of X} holds by \lemmaref{lem:facts_about_Vs}\ref{it:boundedness of V}.
\end{proof}

%
\section{Proofs for \secref{sec:conditional variance and collision analysis main}}
\label{sec:proofs for collision analysis}

\renewcommand{\Gam}[2]{\Gamma_{{#1},{#2}}}
\renewcommand{\cPhi}{\overline{\Phi}}
\renewcommand{\aprod}[4]{\prod_{q=#1}^{#2} A_{q+1}^{#3_{q+1},#3_{q}}A_{q+1}^{#4_{q+1},#4_{q}}}
\renewcommand{\ccomp}[4]{\mathcal{C}_{{#3}_{#1:#2},{#4}_{#1:#2}}}

\begin{proof}[Proof of \propref{prop:tens_prod_formula}] 
We will fix $N\geq2$ and $m\geq1$. 
For all $0 \leq p,q \leq m$ we define $i_{p:q}\defeq(i_{p},\ldots,i_{q})\in [N]^{q-p+1}$, $j_{p:q} \defeq(j_{p},\ldots,j_{q}) \in [N]^{q-p+1}$ and for all $i_{p:q},j_{p:q} \in [N]^{q-p+1}$ we define
\begin{equation*}
\ccomp{p}{q}{i}{j} \defeq \mathcal{C}_{\mathbb{I}[i_{p}=j_{p}]}\cdots \mathcal{C}_{\mathbb{I}[i_{q}=j_{q}]},
\end{equation*}
where $\mathcal{C}_0$ and $\mathcal{C}_{1}$ are as defined in the statement of the proposition.
For all $k\in[m]$ and $i\in [N]$ we define also the measure
\begin{equation*}
\Gam{k}{i} \defeq \sum_{j}A_k^{ij}V_{k-1}^{j}\delta_{\xi^{j}_{k-1}}.
\end{equation*}
To proceed, we will in fact prove a more general result:
\begin{align*}
&\E\Bigg[\Bigg(\dfrac{1}{N}\displaystyle\sum_{i}V_{m}^{i}\delta_{\xi_{m}^{i}}\Bigg) \otimes \Bigg(\dfrac{1}{N}\displaystyle\sum_{i}V_{m}^{i}\delta_{\xi_{m}^{i}}\Bigg) (\Phi)\,\Bigg|\,\F_{0}^{(N,m)}\Bigg]\\
&
=\displaystyle\sum_{\left(i_{0},j_{0},\ldots,i_{m},j_{m}\right)}\left(\dfrac{1}{N^{2}}\prod_{k=0}^{m-1}A_{k+1}^{i_{k+1}i_{k}}A_{k+1}^{j_{k+1}j_{k}}\right)g(\xi_{0}^{i_{0}})g(\xi_{0}^{j_{0}})\ccomp{1}{m}{i}{j}(\Phi)\big(\xi_{0}^{i_{0}},\xi_{0}^{j_{0}}\big),\nonumber
\end{align*}
where $\Phi\in \boundMeas{\ss^2}$. The claim then follows by noting that when $\Phi = \cvarphi_{N}^{\otimes 2}$, then by \eqnref{it:martingale representation}
\begin{align}
\frac{m}{N}\Bigg(\sum_{\varrho\in[Nm]} X_\varrho^{(N,m)}\Bigg)^2 &= \left(\frac{1}{N}\sum_{i}V_{m}^{i}\cvarphi_{N}(\xi_{m}^{i})\right)^2 \nonumber\\
&=\left(\frac{1}{N}\sum_{i}V_{m}^{i}\delta_{\xi_{m}^{i}}\right) \otimes \left(\frac{1}{N}\sum_{i}V_{m}^{i}\delta_{\xi_{m}^{i}}\right) \left(\Phi\right).\label{eq:tensor generalisation}
\end{align}

We first derive an expression for $\E\big[\big(V_{k}^{i_{k}}\delta_{\xi_{k}^{i_{k}}}\big)\otimes\big(V_{k}^{j_{k}}\delta_{\xi_{k}^{j_{k}}}\big)\big(\Phi\big)\big|\xi_{0},\ldots,\xi_{k-1}\big]$,
where $k\in[m]$ and $\Phi\in \boundMeas{\ss^2}$. In the case $i_{k}=j_{k}$, and by writing $\cPhi(x) = \Phi(x,x)$
\begin{align*}
\E\left[\left.\left(V_{k}^{i_{k}}\delta_{\xi_{k}^{i_{k}}}\right)\otimes\left(V_{k}^{j_{k}}\delta_{\xi_{k}^{j_{k}}}\right)\left(\Phi\right)\right|\xi_{0},\ldots,\xi_{k-1}\right]
&=\left(V_{k}^{i_{k}}\right)^{2}\frac{\Gam{k}{i_k}(\cPhi)}{\Gam{k}{i_k}(1)}\\
&=\Gam{k}{i_k}(1)\Gam{k}{i_k}(\cPhi)\\
&=(\Gam{k}{i_k})^{\otimes 2}\left(\mathcal{C}_{1}(\Phi)\right),
\end{align*}
and in the case $i_{k}\neq j_{k}$,
\begin{align*}
&\E\left[\left.\Big(V_{k}^{i_{k}}\delta_{\xi_{k}^{i_{k}}}\Big)\otimes\left(V_{k}^{j_{k}}\delta_{\xi_{k}^{j_{k}}}\right)\left(\Phi\right)\right|\xi_{0},\ldots,\xi_{k-1}\right]\\
&=\left(V_{k}^{i_{k}}\frac{\Gam{k}{i_k}}{\Gam{k}{i_k}(1)}\right)\otimes\left(V_{k}^{j_{k}}\frac{\Gam{k}{j_k}}{\Gam{k}{j_k}(1)}\right)\left(\Phi\right)\\
&=(\Gam{k}{i_k} \otimes\Gam{k}{j_k})\left(\Phi\right),
\end{align*}
so in any case,
\begin{align}
&\E\left[\left.\left(V_{k}^{i_{k}}\delta_{\xi_{k}^{i_{k}}}\right)\otimes\left(V_{k}^{j_{k}}\delta_{\xi_{k}^{j_{k}}}\right)\left(\Phi\right)\right|\xi_{0},\ldots,\xi_{k-1}\right]\nonumber\\
&= (\Gam{k}{i} \otimes\Gam{k}{j})\left(\mathcal{C}_{\mathbb{I}[i_{k}=j_{k}]}\Phi\right).\label{eq:intermediate-1-1}
\end{align}
The proof now proceeds by a backward induction. Our first application of (\ref{eq:intermediate-1-1}) is with $k=m$ to initialize this induction, with the identity:
\begin{align*}
& \E\Bigg[\Bigg(\frac{1}{N}\sum_{i_{m}}V_{m}^{i_{m}}\delta_{\xi_{m}^{i_{m}}}\Bigg)\otimes\Bigg(\frac{1}{N}\sum_{j_{m}}V_{m}^{j_{m}}\delta_{\xi_{m}^{j_{m}}}\Bigg)\left(\Phi\right)\Bigg|\xi_{0},\ldots,\xi_{m-1}\Bigg]\\
&=\frac{1}{N^{2}}\sum_{(i_{m},j_{m})}\left(\Gam{m}{i_{m}} \otimes \Gam{m}{j_{m}}\right)\left(\mathcal{C}_{\mathbb{I}[i_{m}=j_{m}]}\Phi\right).
\end{align*}
The inductive hypothesis is that at rank $k$, with $1\leq k \leq m$, the following holds:
\begin{align*}
& \E\Bigg[\Bigg(\frac{1}{N}\sum_{i_{m}}V_{m}^{i_{m}}\delta_{\xi_{m}^{i_{m}}}\Bigg)\otimes\Bigg(\frac{1}{N}\sum_{j_{m}}V_{m}^{j_{m}}\delta_{\xi_{m}^{j_{m}}}\Bigg)\left(\Phi\right)\Bigg|\xi_{0},\ldots,\xi_{k}\Bigg] \\
&=\frac{1}{N^{2}}\sum_{\left(i_{k},j_{k},\ldots,i_{m},j_{m}\right)}\Bigg(\Bigg(\aprod{k}{m-1}{i}{j}\Bigg)\\
&\quad\times~\left(V_{k}^{i_{k}}\delta_{\xi_{k}^{i_{k}}}\otimes V_{k}^{j_{k}}\delta_{\xi_{k}^{j_{k}}}\right)\left(\ccomp{k+1}{m}{i}{j}(\Phi)\right)\Bigg).
\end{align*}
By \eqnref{eq:intermediate-1-1} we have
\begin{align*}
&\E\left[\left.\left(V_{k}^{i_{k}}\delta_{\xi_{k}^{i_{k}}}\otimes V_{k}^{j_{k}}\delta_{\xi_{k}^{j_{k}}}\right)\left(\ccomp{k+1}{m}{i}{j}(\Phi)\right)\right|\xi_{0},\ldots,\xi_{k-1}\right] \nonumber \\
&= (\Gam{k}{i_k} \otimes\Gam{k}{j_k})\left(\ccomp{k}{m}{i}{j}(\Phi)\right),
\end{align*}
and therefore at rank $k-1$, applying the tower property of conditional expectation gives
\begin{align*}
&\E\Bigg[\Bigg(\frac{1}{N}\sum_{i_{m}}V_{m}^{i_{m}}\delta_{\xi_{m}^{i_{m}}}\Bigg)\otimes\Bigg(\frac{1}{N}\sum_{j_{m}}V_{m}^{j_{m}}\delta_{\xi_{m}^{j_{m}}}\Bigg)\left(\Phi\right)\Bigg|\xi_{0},\ldots,\xi_{k-1}\Bigg]\\
&=
\frac{1}{N^{2}}\sum_{(i_{k},j_{k},\ldots,i_{m},j_{m})}\Bigg(\Bigg(\aprod{k}{m-1}{i}{j}\Bigg)\\
&\quad \times ~\E\Big[\Big(V_{k}^{i_{k}}\delta_{\xi_{k}^{i_{k}}}\otimes V_{k}^{j_{k}}\delta_{\xi_{k}^{j_{k}}}\Big)\Big(\ccomp{k+1}{m}{i}{j}(\Phi)\Big)\Big|\xi_{0},\ldots,\xi_{k-1}\Big]\Bigg)\\
&=
\frac{1}{N^{2}}\sum_{(i_{k-1},j_{k-1},\ldots,i_{m},j_{m})}\Bigg(\Bigg(\aprod{k-1}{m-1}{i}{j}\Bigg)\\
&\quad \times ~ \left(V_{k-1}^{i_{k-1}}\delta_{\xi_{k-1}^{i_{k-1}}}\right)\otimes\left(V_{k-1}^{j_{k-1}}\delta_{\xi_{k-1}^{j_{k-1}}}\right)\Big(\ccomp{k}{m}{i}{j}(\Phi)\Big)\Bigg).
\end{align*}
That is, the hypothesis then also holds at rank $k-1$. Thus the induction is complete, and so we can conclude that for $k=1$,
\begin{align*}
& \E\Bigg[\Bigg(\frac{1}{N}\sum_{i_{m}}V_{m}^{i_{m}}\delta_{\xi_{m}^{i_{m}}}\Bigg)\otimes\Bigg(\frac{1}{N}\sum_{j_{m}}V_{m}^{j_{m}}\delta_{\xi_{m}^{j_{m}}}\Bigg)\left(\Phi\right)\Bigg|\xi_{0}\Bigg]\\
&=\frac{1}{N^{2}}\sum_{(i_{0},j_{0},\ldots,i_{m},j_{m})}\Bigg(\Bigg(\prod_{q=0}^{m-1}A_{q+1}^{i_{q+1}i_{q}}A_{q+1}^{j_{q+1}j_{q}}\Bigg)\\
&\quad \times~ \left(V_{0}^{i_{0}}\delta_{\xi_{0}^{i_{0}}}\right)\otimes\left(V_{0}^{j_{0}}\delta_{\xi_{0}^{j_{0}}}\right)\left(\ccomp{1}{m}{i}{j}(\Phi)\right)\Bigg)\\
&=\sum_{(i_{0},j_{0},\ldots,i_{m},j_{m})}\Bigg(\frac{1}{N^{2}}\Bigg(\prod_{q=0}^{m-1}A_{q+1}^{i_{q+1}i_{q}}A_{q+1}^{j_{q+1}j_{q}}\Bigg)\\
& \quad \times~ g(\xi_{0}^{i_{0}})g(\xi_{0}^{j_{0}})\left(\ccomp{1}{m}{i}{j}(\Phi)\right)\left(\xi_{0}^{i_{0}},\xi_{0}^{j_{0}}\right)\Bigg),
\end{align*}
as required.
\end{proof}

The proof of \propref{prop:tensor product formulation} consists of several technical results which we state and prove first while the actual proof of \propref{prop:tensor product formulation} is postponed to the end of this section. First we establish some key implications of \assref{ass:A_k extra} that will be found useful throughout the remainder of the work.

%
%
\begin{lemma}\label{lem:general properties of path sets}
If $\bba = \bba^{(N,m)}$ satisfies \assref{ass:A_k extra} for some $N,m\geq 1$, then for all $((i_0,\ldots,i_m),(j_0,\ldots,j_m))\in \paths{\bba}{}^2$, $i,j \in [N]$, and $k\in [m]$
\begin{enumerate}[itemsep=0pt, topsep=5pt, partopsep=0pt,label={{(\roman*)}}]
	\item \label{it:sym com property} $i \in \primeParentSet{k}{j}{\bba}{4}$ if and only if $j \in \primeParentSet{k}{i}{\bba}{4}$,
	\item \label{it:equal prime parent sets} If $j\in \primeParentSet{k}{i}{\bba}{4}$ then $\primeParentSet{k}{i}{\bba}{4} = \primeParentSet{k}{j}{\bba}{4}$, if $j\notin \primeParentSet{k}{i}{\bba}{4}$, then $\primeParentSet{k}{i}{\bba}{4} \cap \primeParentSet{k}{j}{\bba}{4} = \emptyset$,
	\item \label{it:positive diagonals} $A_k^{ii}>0$\vphantom{$\primeParentSet{k}{j}{\bba}{4}$},
	\item \label{it:subset property} If $q \leq k$ and $j\in\primeParentSet{k}{i}{\bba}{}$, then $\primeParentSet{q}{j}{\bba}{}\subset\primeParentSet{k}{i}{\bba}{}$,
	\item \label{it:collision now or never} If $j_0 \in \primeParentSet{k}{i_0}{\bba}{4}$, then either $i_k = j_k$ or for all $q\geq k$, $i_q \neq j_q$,
	\item \label{it:complement decomposition} $\bigcup_{k=1}^m \collisionStartSetSh{k}{i}{\bba} = [N] \setminus \{i\}$, and for all $k,k'\in[m]$ such that $k\neq k'$, $\collisionStartSetSh{k}{i}{\bba} \cap \collisionStartSetSh{k'}{i}{\bba} =\emptyset$.
\end{enumerate}
\end{lemma}
\begin{proof}
\ref{it:sym com property} follows from \eqnref{eq:prime parents} and parts \ref{it:symmetry} and \ref{it:commutativity} of \assref{ass:A_k extra}.

To check \ref{it:equal prime parent sets}, suppose that $j \in \primeParentSet{k}{i}{\bba}{}$, and there exists $u \in [N]$ such that $u\in \primeParentSet{k}{i}{\bba}{}$. Then, by parts \ref{it:commutativity}, \ref{it:idempotence}, \ref{it:symmetry} of \assref{ass:A_k extra} and \assref{ass:A_k}, $A_{k:1}^{ju} = \left(A_{k:1}A_{k:1}\right)^{ju} = \sum_\ell A_{k:1}^{j\ell}A_{k:1}^{u\ell}\geq A_{k:1}^{ji}A_{k:1}^{i u}> 0$ where the last inequality holds by assumption. Thus $u \in \primeParentSet{k}{j}{\bba}{}$ proving $\primeParentSet{k}{i}{\bba}{}\subset\primeParentSet{k}{j}{\bba}{}$. The converse inclusion follows from the symmetry of the arguments. For the case $j \notin \primeParentSet{k}{i}{\bba}{}$ we assume there exists $u \in \primeParentSet{k}{i}{\bba}{4}\cap \primeParentSet{k}{j}{\bba}{4}$, from which it follows that $0<A_{k:1}^{iu}A_{k:1}^{uj}\leq \sum_\ell A_{k:1}^{i\ell}A_{k:1}^{\ell j} = \left(A_{k:1}A_{k:1}\right)^{ij} = A_{k:1}^{ij}\iff j \in \primeParentSet{k}{i}{\bba}{}$ constituting a contradiction, which completes the proof of \ref{it:equal prime parent sets}.

To prove \ref{it:positive diagonals} we have by \assref{ass:A_k} and parts \ref{it:idempotence} and \ref{it:symmetry} of \assref{ass:A_k extra}, $A_k^{ii} = (A_kA_k)^{ii} = \sum_{\ell} A_k^{i\ell}A_k^{i\ell} > 0$.

To prove \ref{it:subset property}, suppose that for some $q\leq k$ and $j\in\primeParentSet{k}{i}{\bba}{}$ we take $u\in\primeParentSet{q}{j}{\bba}{}$. By \ref{it:positive diagonals}, there exists $(i_q,\ldots,i_k)\in\paths{\bba_{q+1:k}}{}$ such that $i_q = i_k = j$. Since $u\in\primeParentSet{q}{j}{\bba}{}$ there exists $(i_0,\ldots,i_k) \in \paths{\bba_{1:k}}{}$ such that $i_0=u$ and $i_k=j$ implying, by \ref{it:equal prime parent sets}, that $u\in\primeParentSet{k}{j}{\bba}{}=\primeParentSet{k}{i}{\bba}{}$, hence proving  \ref{it:subset property}.

To prove \ref{it:collision now or never}, we observe that by \ref{it:equal prime parent sets}, $j_0 \in \primeParentSet{k}{i_0}{\bba}{4}$ implies $\primeParentSet{k}{i_0}{\bba}{4} = \primeParentSet{k}{j_0}{\bba}{4}$. Because $(j_0,\ldots,j_m)\in\paths{\bba}{m}$, we have $j_0 \in \primeParentSet{k}{j_k}{\bba}{4}$ and by \ref{it:sym com property}, $j_k \in \primeParentSet{k}{j_0}{\bba}{4}=\primeParentSet{k}{i_0}{\bba}{4}$. Hence,  by \ref{it:sym com property}, $i_0 \in \primeParentSet{k}{j_k}{\bba}{4}$ and there exists $(i',j') \in \paths{\bba}{m}^2$, where $i' = (i'_0,\ldots,i'_m)$ and $j' = (j'_0,\ldots,j'_m)$, such that $i'_0 = j'_0 = i_0$, $i'_k = i_k$, and $j'_k = j_k$. Now suppose that $i_k \neq j_k$ and there exists $q \geq k$ such that $i_q = j_q$. By the existence of $(i',j')$, we can construct paths $i'' = (i'_0,\ldots,i'_k,i_{k+1},\ldots,i_m)$ and $j'' = (j'_0,\ldots,j'_k,j_{k+1},\ldots,j_m)$ for which we have $i''_0 = j''_0$, $i''_k \neq j''_k$ and $i''_q = j''_q$ contradicting \assref{ass:A_k extra}\ref{it:unique paths} which completes the proof of \ref{it:collision now or never}.

To prove \ref{it:complement decomposition}, we observe that by \ref{it:subset property} and \ref{it:positive diagonals}, $\primeParentSet{k-1}{i}{\bba}{}\subset\primeParentSet{k}{i}{\bba}{}$. Moreover, by definition $\primeParentSet{0}{i}{\bba}{} = \{i\}$ and by \assref{ass:A_k}, $\primeParentSet{m}{i}{\bba}{} = [N]$. Therefore it is a matter of elementary set operations to check that
\begin{equation*}
\bigcup_{k=1}^m \collisionStartSetSh{k}{i}{\bba} = \bigcup_{k=1}^m \collisionStartSet{k}{i}{\bba} = \primeParentSet{m}{i}{\bba}{}\setminus \{i\} = [N] \setminus \{i\}.
\end{equation*}
Empty intersections follow straightforwardly by definition \eqnref{eq:collision start sets} and the fact that $\primeParentSet{k-1}{i}{\bba}{}\subset\primeParentSet{k}{i}{\bba}{}$.
\end{proof}


We start proving \propref{prop:tensor product formulation} by writing $i_{0:m} = (i_0,\ldots,i_m)$ and $j_{0:m} = (j_0,\ldots,j_m)$ for brevity. Then, for any $N,m\geq 1$, the set $\paths{\bba}{m}^2$, where $\bba = \bba^{(N,m)}$, can be decomposed in three disjoint sets
\begin{equation}\label{eq:collision set definitions}
\arraycolsep=1.4pt
\begin{array}{rcl}
\clsnSetn{1}{m}{\bba} &\defeq& \big\{(i_{0:m},j_{0:m}) \in \paths{\bba}{m}^2: i_0 = j_0 \big\}, \rule[-3mm]{0pt}{4mm}\\
\clsnSetn{2}{m}{\bba} &\defeq& \textstyle{\bigcup_{k=1}^m} \left\{(i_{0:m},j_{0:m}) \in \paths{\bba}{m}^2: i_0 \neq j_0,~ i_{k}=j_{k} \right\}, \rule[-3mm]{0pt}{4mm}\\
\clsnSetn{3}{m}{\bba} &\defeq& \textstyle{\bigcap_{k=1}^m} \left\{(i_{0:m},j_{0:m}) \in \paths{\bba}{m}^2: i_0 \neq j_0,~ i_{k}\neq j_{k} \right\}.\rule[-3mm]{0pt}{4mm}
\end{array}
\end{equation}
Clearly, the sets $\clsnSetn{1}{m}{\bba}$, $\clsnSetn{2}{m}{\bba}$ and $\clsnSetn{3}{m}{\bba}$ form a partition of $\paths{\bba}{m}^2$.
\begin{lemma}\label{lem:partition}
Fix $N,m\geq 1$ and $\bba = \bba^{(N,m)}$. The sets $\clsnSetn{2}{m}{\bba}$ and $\clsnSetn{3}{m}{\bba}$ admit the decompositions:
\begin{align}
\clsnSetn{2}{m}{\bba} &= \bigcup_{i\in[N]}\bigcup_{\substack{j\in[N]\\j\neq i}}\bigcup_{k\in[m]}\bigcup_{u\in[N]} \specCollisionPairsFrom{\bba}{m}{k}{u}{i}{j}. \label{eq:partition formulation for d2}\\
\clsnSetn{3}{m}{\bba} &= \bigcup_{i\in[N]}\bigcup_{\substack{j\in[N]\\j\neq i}} \bigg(\big(\pathsFrom{i}{\bba}\times\pathsFrom{j}{\bba}\big) \setminus \bigcup_{k\in[m]}\bigcup_{u\in[N]}\specCollisionPairsFrom{\bba}{m}{k}{u}{i}{j}\bigg),\label{eq:partition formulation for d3}
\end{align}
where for all $k\in[m]$ and $i,j,u \in [N]$
\begin{align*}
&\specCollisionPairsFrom{\bba}{m}{k}{u}{i}{j} \defeq\\
& \quad\displaystyle\bigcap_{q=0}^{k-1} \left\{((i'_0,\ldots,i'_m),(j'_0,\ldots,j'_m)) \in \pathsFrom{i}{\bba}\times \pathsFrom{j}{\bba}: i'_{k}=j'_{k}=u,~i'_{q}\neq j'_{q}\right\},  	
\end{align*}
and for all $(k,u,i,j) \neq (k',u',i',j')$, $\specCollisionPairsFrom{\bba}{m}{k}{u}{i}{j} \cap \specCollisionPairsFrom{\bba}{m}{k'}{u'}{i'}{j'} = \emptyset$.
\end{lemma}
\begin{proof}
First we observe that for the sets in \eqnref{eq:partition formulation for d2} the inclusion $\supset$ is trivial by definition. Then take a pair $(i_0,\ldots,i_m)$ and $(j_0,\ldots,j_m)$ belonging to $\clsnSetn{2}{m}{\bba}$. Then there exists $p = \min(q \in [m]: i_q = j_q)$, and thus the pair also belongs to $\specCollisionPairsFrom{\bba}{}{p}{i_p}{i_0}{j_0}$ and therefore also the inclusion $\subset$ holds, establishing \eqnref{eq:partition formulation for d2}. By elementary set theory it follows by \eqnref{eq:collision set definitions} that $\clsnSetn{3}{m}{\bba} =\big( \bigcup_{i\in[N]}\bigcup_{j\in[N],j\neq i} \big(\pathsFrom{i}{\bba}\times\pathsFrom{j}{\bba}\big)\big)\setminus \clsnSetn{2}{m}{\bba}$ and since for all $k\in[m]$ and $u\in[N]$, $\specCollisionPairsFrom{\bba}{m}{k}{u}{i}{j} \subset \pathsFrom{i}{\bba}\times\pathsFrom{j}{\bba}$, \eqnref{eq:partition formulation for d3} can be checked by elementary set theory.

To prove that the sets $\specCollisionPairsFrom{\bba}{m}{k}{u}{i}{j}$ are disjoint, assume that
\begin{equation}\label{eq:non empty intersection}
((i_0,\ldots,i_m),(j_0,\ldots,j_m)) \in \specCollisionPairsFrom{\bba}{m}{k}{u}{i}{j} \cap \specCollisionPairsFrom{\bba}{m}{k'}{u'}{i'}{j'}.
\end{equation}
If $i\neq i'$ and \eqnref{eq:non empty intersection} was true, then $i = i_0 = i' \neq i$, and similarly for $j$ and $j'$. In the case $k \neq k'$, since we are not assuming anything about the values of $u$, $u'$, $i$, $i'$, $j$ and $j'$, we can assume without loss of generality that $k < k'$. Now if \eqnref{eq:non empty intersection} was true, then $i_{k} = j_{k}$ and $i_{k} \neq j_{k}$, which is a contradiction. Finally it suffices to consider the case $k=k'$ and $u\neq u'$. If \eqnref{eq:non empty intersection} was true, then one must have $u = i_k = j_k = u' \neq u$ which is a contradiction completing the proof.
\end{proof}

The cardinality of a set can be evaluated by constructing a bijection between the set in question and some other set with known cardinality. For this purpose, we have the following result. Note that throughout the remainder of this document, for given $N,m\geq 1$, $\bba = \bba^{(N,m)}$, $0\leq k \leq m$ and $u\in[N]$, we let $\lowerPartFrom{k}{u}{\bba}$ be as defined in the statement of \propref{prop:tensor product formulation}.
\begin{lemma}\label{eq:bijectivity h}
Fix $N,m\geq 1$ and $\bba = \bba^{(N,m)}$. For all $i,j,u\in [N]$, such that $i\neq j$ and $k\in[m]$, define
\begin{align*}
&\upperPartFrom{\bba}{}{k}{u}{i}{j} \defeq \\
&~\bigcap_{q=0}^{k-1}\left\{((i'_0,\ldots,i'_k),(j'_0,\ldots,j'_k))\in \pathsFrom{i}{\bba_{1:k}}\times \pathsFrom{j}{\bba_{1:k}}:i'_k=j'_k=u,~i'_q\neq j'_q\right\}, 
\end{align*}
and let the mapping
\begin{equation*}
\kappa: \specCollisionPairsFrom{\bba}{}{k}{u}{i}{j} \to \upperPartFrom{\bba}{}{k}{u}{i}{j} \times \lowerPartFrom{k}{u}{\bba} \times \lowerPartFrom{k}{u}{\bba},
\end{equation*}
be defined as
\begin{equation*}
\kappa:(i_{0:m},j_{0:m}) \mapsto \left(((i_0,\ldots,i_{k}),(j_0,\ldots,j_{k})),(i_{k},\ldots,i_{m}),(j_{k},\ldots,j_{m})\right),
\end{equation*}
where $i_{0:m} \defeq (i_0,\ldots,i_m)$, $j_{0:m} \defeq (j_0,\ldots,j_m)$.
Then $\kappa$ is a bijection.
\end{lemma}
\begin{proof}
By the definitions of $\lowerPartFromNoArg{\bba}$, $\specCollisionPairsFromNoArg{\bba}$ and $\upperPartFromNoArg{\bba}$, for any $(i_{0:m},j_{0:m}) \in \specCollisionPairsFrom{\bba}{}{k}{u}{i}{j}$, where $i,j,u\in[N]$ such that $i\neq j$ and $k\in[m]$
\begin{align*}
& \left(((i_0,\ldots,i_{k}),(j_0,\ldots,j_{k})),(i_{k},\ldots,i_{m}),(j_{k},\ldots,j_{m})\right) \\
&\qquad \in \upperPartFrom{\bba}{}{k}{u}{i}{j} \times \lowerPartFrom{k}{u}{\bba} \times \lowerPartFrom{k}{u}{\bba}.
\end{align*}
If $(i_{0:m},j_{0:m}) \neq (i'_{0:m},j'_{0:m}) \in \specCollisionPairsFrom{\bba}{}{k}{u}{i}{j}$ then $\kappa(i_{0:m},j_{0:m}) \neq \kappa(i'_{0:m},j'_{0:m})$, from which we conclude that $\kappa$ is an injection. To see that $\kappa$ is a surjection, take any
\begin{align*}
&(((i_0,\ldots,i_{k}),(j_0,\ldots,j_{k})),(i'_k,\ldots,i'_{m}),(j'_k,\ldots,j'_{m})) \\
&\qquad \in \upperPartFrom{\bba}{}{k}{u}{i}{j} \times \lowerPartFrom{k}{u}{\bba} \times \lowerPartFrom{k}{u}{\bba}.
\end{align*}
Then $i_0 = i$, $j_0 = j$, and by the definitions of $\lowerPartFromNoArg{\bba}$ and $\upperPartFromNoArg{\bba}$, $i_{k} = j_{k} = i'_{k} = j'_{k} = u$, for all $0 \leq p < k$, $i_{p} \neq j_{p}$, and $(i_0,\ldots,i_k,i'_{k+1}\ldots,i'_{m})\in \pathsFrom{i}{\bba}$, $(j_0,\ldots,j_k,j'_{k+1}\ldots,j'_{m})\in \pathsFrom{j}{\bba}$. From these observations we conclude by the definition of $\specCollisionPairsFrom{\bba}{}{k}{u}{i}{j}$ that
\begin{equation*}
((i_0,\ldots,i_k,i'_{k+1}\ldots,i'_{m}),(j_0,\ldots,j_k,j'_{k+1}\ldots,j'_{m})) \in \specCollisionPairsFrom{\bba}{}{k}{u}{i}{j},
\end{equation*}
and hence $\kappa$ is a surjection.
\end{proof}
By using the bijectivity result, \lemmaref{eq:bijectivity h}, we can find an expression for the cardinalities of the sets $\specCollisionPairsFrom{\bba}{m}{k}{u}{i}{j}$ in terms of the cardinalities of the sets $\lowerPartFrom{k}{u}{\bba}{}$ as defined in the statement of \propref{prop:tensor product formulation}. This is established by the following result.
\begin{lemma}
\label{lem:collision cardinality}
\label{lem:noncolliding cardinality}
If $\bba = \bba^{(N,m)}$ satisfies \assref{ass:A_k extra} for some $N,m\geq 1$, then for all $k\in[m]$ and $i,j,u \in [N]$ such that $i\neq j$
\begin{equation*}
\Big|\specCollisionPairsFrom{\bba}{m}{k}{u}{i}{j}\Big| =
\Big|\lowerPartFrom{k}{u}{\bba}{}\Big|^2\ind{}\Big(u \in \primeParentSet{k}{i}{\bba}{}\Big)\ind{}\Big(j \in \collisionStartSetSh{k}{i}{\bba}\Big).
\end{equation*}
\end{lemma}
\begin{proof}
First we prove the part $\specCollisionPairsFrom{\bba}{m}{k}{u}{i}{j} = \emptyset$ if $(u,j) \notin \primeParentSet{k}{i}{\bba}{}\times\collisionStartSetSh{k}{i}{\bba}$. If $u\notin \primeParentSet{k}{i}{\bba}{}$, then by \lemmaref{lem:general properties of path sets}\ref{it:sym com property} $i\notin \primeParentSet{k}{u}{\bba}{}$ and hence $\specCollisionPairsFrom{\bba}{m}{k}{u}{i}{j} = \emptyset$. Next, if $j \in \primeParentSet{k-1}{i}{\bba}{4}$, then by \lemmaref{lem:general properties of path sets}\ref{it:collision now or never}, for all $(i'_0,\ldots,i'_m)\in\pathsFrom{i}{\bba}$, $(j'_0,\ldots,j'_m)\in\pathsFrom{j}{\bba}$, either $i'_{k-1} = j'_{k-1}$ or $i'_{k} \neq j'_{k}$ and hence $\specCollisionPairsFrom{\bba}{m}{k}{u}{i}{j}=\emptyset$. For $j \notin \primeParentSet{k}{i}{\bba}{4}$, suppose that $\specCollisionPairsFrom{\bba}{m}{k}{u}{i}{j}\neq\emptyset$. In this case, $i,j \in \primeParentSet{k}{u}{\bba}{4}$ and by \lemmaref{lem:general properties of path sets}\ref{it:sym com property} $u \in \primeParentSet{k}{i}{\bba}{4}\cap\primeParentSet{k}{j}{\bba}{4}$ and hence by \lemmaref{lem:general properties of path sets}\ref{it:equal prime parent sets} $j \in \primeParentSet{k}{i}{\bba}{4}$, which concludes the proof for $(u,j) \notin \primeParentSet{k}{i}{\bba}{}\times\collisionStartSetSh{k}{i}{\bba}$.

Next we prove that $\specCollisionPairsFrom{\bba}{m}{k}{u}{i}{j} \neq \emptyset$, if $(u,j) \in \primeParentSet{k}{i}{\bba}{}\times\collisionStartSetSh{k}{i}{\bba}$. Take $(u,j) \in \primeParentSet{k}{i}{\bba}{}\times\collisionStartSetSh{k}{i}{\bba}$. Because $u\in\primeParentSet{k}{i}{\bba}{}$, then by \lemmaref{lem:general properties of path sets}\ref{it:equal prime parent sets}, $\primeParentSet{k}{i}{\bba}{} = \primeParentSet{k}{u}{\bba}{}$, and because $j \in \primeParentSet{k}{i}{\bba}{} = \primeParentSet{k}{u}{\bba}{}$, then by \lemmaref{lem:general properties of path sets}\ref{it:sym com property}, 
$i,j \in \primeParentSet{k}{u}{\bba}{}$ from which we conclude that there exists $(i'_0,\ldots,i'_m)\in\pathsFrom{i}{\bba}$ and $(j'_0,\ldots,j'_m) \in \pathsFrom{j}{\bba}$ such that $i'_k = j'_k = u$ and $i'_0 =i$ and $j'_0 =j$. Suppose then that 
$i'_{k-1} = j'_{k-1}$. This would imply that $i,j \in \primeParentSet{k-1}{i'_{k-1}}{\bba}{4}$ and, by \lemmaref{lem:general properties of path sets}\ref{it:equal prime parent sets}, $\primeParentSet{k-1}{i}{\bba}{4} = \primeParentSet{k-1}{j}{\bba}{4} = \primeParentSet{k-1}{i'_{k-1}}{\bba}{4}$ and hence $j \in \primeParentSet{k-1}{i}{\bba}{4}$ which is a contradiction implying that $i'_{k-1} \neq j'_{k-1}$. By \assref{ass:A_k extra}\ref{it:unique paths} we then deduce that $i'_q \neq j'_q$ for all $q < k$ and hence $((i'_0,\ldots,i'_m),(j'_0,\ldots,j'_m))\in \specCollisionPairsFrom{\bba}{m}{k}{u}{i}{j}$ which can therefore not be empty.

Finally, by \lemmaref{eq:bijectivity h}, for nonempty $\specCollisionPairsFrom{\bba}{m}{k}{u}{i}{j}$ we have
\begin{equation*}
\card{\specCollisionPairsFrom{\bba}{m}{k}{u}{i}{j}} = \card{\upperPartFrom{\bba}{m}{k}{u}{i}{j}}\card{\lowerPartFrom{k}{u}{\bba}}\card{ \lowerPartFrom{k}{u}{\bba}},
\end{equation*}
and by \assref{ass:A_k extra}\ref{it:unique paths} we have $\card{\upperPartFrom{\bba}{m}{k}{u}{i}{j}}=1$ which concludes the proof.
%
\end{proof}

By \lemmaref{lem:collision cardinality} we observe that in order to have explicit expressions for the cardinalities of $\specCollisionPairsFrom{\bba}{m}{k}{u}{i}{j}$, it suffices to have expressions for the cardinalities of the sets $\lowerPartFrom{k}{u}{\bba}$. In order to evaluate these cardinalities, we follow the principle mentioned earlier of constructing appropriate bijections to sets with known cardinalities, according to the following result.
\begin{lemma}\label{lem:bijectivity g}
Suppose that $\bba = \bba^{(N,m)}$ satisfies \assref{ass:A_k extra} for some $N,m\geq 1$. For all $i\in [N]$ and $k\in [m]$, let $r_k \defeq \big|\parentSet{k}{i}{\bba}{}\big|$ and let
$\phi^i_k:\parentSet{k}{i}{\bba}{} \to [r_k]$ be arbitrary bijections. Then for any $u\in [N]$ and $k\in [m]$, the mapping $\gamma: \lowerPartFrom{k}{u}{\bba} \to  \{u\} \times [r_{k+1}]\times\cdots\times[r_m]$, defined as $\gamma: (i_{0},\ldots,i_{m-k}) \mapsto (c_0,\ldots,c_{m-k})$, where $c_0 = i_0$ and for all $0 \leq p < m-k$, $c_{p+1} = \phi_{p+k+1}^{i_p}(i_{p+1})$, is a bijection.
\end{lemma}
\begin{proof}
From the definition of $\lowerPartFromNoArg{\bba}$ and \assref{ass:A_k extra}\ref{it:symmetry}, it follows that for given $(i_0,\ldots,i_{m-k}) \in \lowerPartFrom{k}{u}{\bba}$, one has $i_{p+1} \in \parentSet{p+k+1}{i_{p}}{\bba}{m}$ for all $0 \leq p < m-k$. It then follows that $c_{p+1} = \phi_{p+k+1}^{i_p}(i_{p+1}) \in [r_{p+k+1}]$ for all $0 \leq p < m-k$ and thus $\gamma(i_0,\ldots,i_{m-k}) \in \{u\} \times [r_{k+1}]\times\cdots\times[r_m]$.

For $((i_0,\ldots,i_{m-k}),(i'_0,\ldots,i'_{m-k}))\in\lowerPartFrom{k}{u}{\bba}^2$ such that $(i_0,\ldots,i_{m-k})\neq(i'_0,\ldots,i'_{m-k})$, one can take $q = \max(p \in \{0,\ldots,m-k\}: i_p=i'_p)$ for which $i_{q+1}\neq i'_{q+1}$. By the bijectivity of $\phi_{q+k+1}^{i_q}$, one has $\phi_{q+k+1}^{i_q}(i_{q+1}) \neq \phi_{q+k+1}^{i_q}(i'_{q+1}) = \phi_{q+k+1}^{i'_q}(i'_{q+1})$. From this it follows that $\gamma(i_0,\ldots,i_{m-k}) \neq \gamma(i'_0,\ldots,i'_{m-k})$ proving that $\gamma$ is injection.

For given $0\leq  p < m-k$, $c \in [r_{p+k+1}]$ and $i \in [N]$, one has $(\phi_{p+k+1}^{i})^{-1}(c) \in \parentSet{p+k+1}{i}{\bba}{}$ and hence if for any given $(c_0,\ldots,c_{m-k}) \in \{u\} \times [r_{k+1}]\times\cdots\times[r_m]$, $(i_0,\ldots,i_{m-k})$ is defined recursively as $i_0 = c_0$ and $i_{p+1} = (\phi_{p+k+1}^{i_p})^{-1}(c_{p+1})$ for all $0 \leq p < m-k$, then $(i_0,\ldots,i_{m-k}) \in \lowerPartFrom{k}{n}{\bba}$ and $\gamma(i_0,\ldots,i_{m-k}) = (c_0,\ldots,c_{m-k})$, which completes the proof.
\end{proof}

\begin{remark}\label{rem:invariance}
By \lemmaref{lem:bijectivity g}, the primary implication of \assref{ass:A_k extra}\ref{it:cardinality invariance} becomes clear. Effectively it implies that that the cardinalities of $\lowerPartFrom{k}{u}{\bba}$ are independent of $u$ and, as a corollary of \lemmaref{lem:bijectivity g}, we have for any $u \in [N]$ and $(i_{k+1},\ldots,i_{m}) \in [N]^{m-k}$,
\begin{equation*}
\Big|\lowerPartFrom{k}{u}{\bba}\Big| = \prod_{q=k+1}^m \Big|\parentSet{q}{i_{q}}{\bba}{}\Big|,
\end{equation*}
which is simple to evaluate given the explicit definition of $\bba$.
\end{remark}

We have now all the ingredients to prove \propref{prop:tensor product formulation}.

\begin{proof}[Proof of \propref{prop:tensor product formulation}]
Throughout the proof we will use the notations $i_{p:q} \defeq (i_p,\ldots,i_q) \in [N]^{q-p+1}$ and $j_{p:q} \defeq (j_p,\ldots,j_q)\in [N]^{q-p+1}$, for all $0\leq p\leq q\leq m$. First note that by \propref{prop:tens_prod_formula} 
\begin{align}
\label{eq:decomposable sum}
& \displaystyle\frac{m}{N}\E\Bigg[\Bigg(\sum_{\varrho\in [Nm]} X_\varrho^{(N,m)}\Bigg)^2\Bigg|\F_0^{(N,m)}\Bigg] =
\\
&\displaystyle\sum_{\left(i_{0},j_{0},...,i_{m},j_{m}\right)}\frac{1}{N^{2}}\Bigg(\prod_{k=0}^{m-1}A_{k+1}^{i_{k+1}i_{k}}A_{k+1}^{j_{k+1}j_{k}}\Bigg)
g(\xi_{0}^{i_{0}})g(\xi_{0}^{j_{0}})\ccomp{1}{m}{i}{j}\big(\cvarphi_{N}^{\otimes2}\big)\big(\xi_{0}^{i_{0}},\xi_{0}^{j_{0}}\big).\nonumber
\end{align}
By \assref{ass:A_k extra}\ref{it:unique paths}, there exists at most one sequence $(i'_0,\ldots,i'_m)$ in $\paths{\bba}{}$ for which $(i_0,i_m) = (i'_0,i'_m)$. Therefore, by \assref{ass:A_k extra}\ref{it:symmetry} and \assref{ass:A_k}, we have
\begin{equation*}
\prod_{k=0}^{m-1} A_{k+1}^{i'_{k+1}i'_{k}} =\sum_{(i_1,\ldots,i_{m-1})}\prod_{k=0}^{m-1} A_{k+1}^{i_{k+1}i_{k}} = \left(\prod_{k=0}^{m-1} A_{k+1}\right)^{i_mi_0} = \frac{1}{N},
\end{equation*}
and hence $\prod_{k=0}^{m-1}A_{k+1}^{i_{k+1}i_{k}}A_{k+1}^{j_{k+1}j_{k}} = N^{-2}\ind{}(((i_0,\ldots,i_m),(j_0,\ldots,j_m))	\in	\paths{\bba}{}^2)$,
and from \eqnref{eq:decomposable sum} we then have
\begin{align}
&\frac{m}{N}\E\Bigg[\Bigg(\sum_{\varrho\in [Nm]} X_\varrho^{(N,m)}\Bigg)^2\Bigg|\F_0^{(N,m)}\Bigg]\nonumber\\
&= \frac{1}{N^4}\sum_{\left(i,j\right)\in \paths{\bba}{m}^2}g(\xi_{0}^{i_{0}})g(\xi_{0}^{j_{0}})\left(\ccomp{1}{m}{i}{j}\left(\cvarphi_{N}^{\otimes2}\right)\right)\left(\xi_{0}^{i_{0}},\xi_{0}^{j_{0}}\right).\label{eq:decomposable sum 2}
\end{align}
%

By \lemmaref{lem:partition}, \assref{ass:A_k extra}\ref{it:cardinality invariance} and \lemmaref{lem:collision cardinality} (see also \remref{rem:invariance}),
\begin{equation}\label{eq:collision cardinality sum 2}
\Bigg|\bigcup_{k\in[m]}\bigcup_{u\in[N]} \specCollisionPairsFrom{\bba}{m}{k}{u}{i_0}{j_0}\Bigg| = \sum_{k=1}^m\sum_{u = 1}^{N}\card{\specCollisionPairsFrom{\bba}{m}{k}{u}{i_0}{j_0}} = \specCollisionPairs{i_0}{j_0}{\bba}{}.
\end{equation}
By \assref{ass:A_k}, for all $u\in[N]$, $\primeParentSet{m}{u}{\bba}{} = [N]$, and therefore for all $i\in[N]$, there exists a sequence $(i_0,\ldots,i_m)\in\pathsFrom{i}{\bba}$ such that $i_m = u$. On the other hand, by \assref{ass:A_k extra}\ref{it:unique paths} there exist at most one such sequence from which we conclude that $\big|\lowerPartFrom{0}{i}{\bba}\big| = N$. Therefore, by the definition of $\lowerPartFromNoArg{\bba}$,
\begin{equation}\label{eq:L set cardinality}
\big|\pathsFrom{i}{\bba}\times\pathsFrom{j}{\bba}\big| = \big|\lowerPartFrom{0}{i}{\bba}\big|\big|\lowerPartFrom{0}{j}{\bba}\big| = N^2.
\end{equation}

Using the fact that by \lemmaref{lem:general properties of path sets}\ref{it:complement decomposition} for all $i_{0} \neq j_{0}$, $\sum_{k=1}^m \ind{}(j_{0} \in \collisionStartSetSh{k}{i_{0}}{\bba}) = 1$, we then have by \eqnref{eq:L set cardinality}, \lemmaref{lem:partition}, \lemmaref{lem:noncolliding cardinality} and \eqnref{eq:collision cardinality sum 2} that
\begin{align}
&\Bigg|\big(\pathsFrom{i_{0}}{\bba}\times\pathsFrom{j_{0}}{\bba}\big) \setminus \bigcup_{k\in[m]}\bigcup_{u\in[N]}\specCollisionPairsFrom{\bba}{m}{k}{u}{i_{0}}{j_{0}}\Bigg| \nonumber\\
&= N^2\sum_{k=1}^m \ind{}\big(j_{0} \in \collisionStartSetSh{k}{i_{0}}{\bba}\big)  - \sum_{k=1}^m\sum_{u = 1}^{N}\card{\specCollisionPairsFrom{\bba}{m}{k}{u}{i_{0}}{j_{0}}}\nonumber\\
&= \specNonCollisionPairsFrom{\bba}{m}{i_{0}}{j_{0}}.\label{eq:collision cardinality sum 3}
\end{align}
By \eqnref{eq:collision operator} and \eqnref{eq:collision set definitions}
\begin{align*}
&g(\xi_{0}^{i_{0}})g(\xi_{0}^{j_{0}})\big(\ccomp{1}{m}{i}{j}\big(\cvarphi_{N}^{\otimes2}\big)\big)\big(\xi_{0}^{i_{0}},\xi_{0}^{j_{0}}\big) \\
&=
\begin{cases}
g^{2}(\bxi{i_0}{0}{}{})\cvarphi^2_{N}(\bxi{i_0}{0}{}{}), & (i_{0:m},j_{0:m})\in\clsnSetn{1}{m}{\bba}, \rule[-2mm]{0pt}{3mm}\\
g(\bxi{i_0}{0}{}{})\cvarphi^2_{N}(\bxi{i_0}{0}{}{})g(\bxi{j_0}{0}{}{}), & (i_{0:m},j_{0:m})\in\clsnSetn{2}{m}{\bba},  \rule[-2mm]{0pt}{3mm}\\
g(\bxi{i_0}{0}{}{})\cvarphi_{N}(\bxi{i_0}{0}{}{})g(\bxi{j_0}{0}{}{})\cvarphi_{N}(\bxi{j_0}{0}{}{}), & (i_{0:m},j_{0:m})\in\clsnSetn{3}{m}{\bba}.\rule[-2mm]{0pt}{3mm}
\end{cases}
\end{align*}
For the set $\clsnSetn{1}{m}{\bba}$ we have the disjoint decomposition $$\clsnSetn{1}{m}{\bba} = \textstyle\bigcup_{u=1}^{N} \pathsFrom{u}{\bba}\times \pathsFrom{u}{\bba}.$$
By \lemmaref{lem:partition}, \eqnref{eq:L set cardinality}, \eqnref{eq:collision cardinality sum 2} and \eqnref{eq:collision cardinality sum 3} we have
\begin{align*}
\sum_{(i_{0:m},j_{0:m}) \in\clsnSetn{1}{m}{\bba}}g^{2}(\bxi{i_0}{0}{}{})\cvarphi^2_{N}(\bxi{i_0}{0}{}{})
&= N^2\sum_{i_0}g^{2}(\bxi{i_0}{0}{}{})\cvarphi^2_{N}(\bxi{i_0}{0}{}{}),\\
\sum_{(i_{0:m},j_{0:m}) \in\clsnSetn{2}{m}{\bba}}g(\bxi{i_0}{0}{}{})\cvarphi^2_{N}(\bxi{i_0}{0}{}{})g(\bxi{j_0}{0}{}{})
&= \sum_{i_0}\sum_{j_0\neq i_0}g(\bxi{i_0}{0}{}{})\cvarphi^2_{N}(\bxi{i_0}{0}{}{})g(\bxi{j_0}{0}{}{})\specCollisionPairs{i_0}{j_0}{\bba}{4}
\end{align*}
and
\begin{align*}
&\sum_{(i_{0:m},j_{0:m}) \in\clsnSetn{3}{m}{\bba}}g(\bxi{i_0}{0}{}{})\cvarphi_{N}(\bxi{i_0}{0}{}{})g(\bxi{j_0}{0}{}{})\cvarphi_{N}(\bxi{j_0}{0}{}{}) \\
&= \sum_{i_0}\sum_{j_0\neq i_0}g(\bxi{i_0}{0}{}{})\cvarphi_{N}(\bxi{i_0}{0}{}{})g(\bxi{j_0}{0}{}{})\cvarphi_{N}(\bxi{j_0}{0}{}{})\specNonCollisionPairsFrom{\bba}{m}{i_0}{j_0}. 
\end{align*}
The proof is completed by substituting the last three equations into \eqnref{eq:decomposable sum 2}.
\end{proof}

\section{Proofs for \secref{sec:lln and clt fixed radix}}
\label{sec:proofs for fixed radix supp}

In this section, we essentially focus on establishing the condition \eqnref{eq:douc_cond_var} of \theref{thm:Douc and Moulines} for the radix-$r$ algorithm. Because of the lengthy analysis, this task is divided into the three subsequent sections. In \secref{sec:conditional independece structure fixed radix} we establish that the specific choice of matrices $\bba^{(r^m,m)} = \radixMatrices{r}{m}$ associated with the radix-$r$ algorithm enables us to construct partitions, call them $\radixSamplePartition{r}{m}{\ka}$, such that for any given $d \in [m]$, the triple $\big(\radixMatrices{r}{m},\radixSamplePartition{r}{m}{\ka},d\big)$ satisfies all the required conditions, namely Assumptions \ref{ass:partition and cond indep} and \ref{ass:A_k extra}, that we need to establish \eqnref{eq:douc_cond_var}. The task then becomes two fold due to the structure of the proof of \theref{thm:butterfly_clt_fixed_radix} where the sum in \eqnref{eq:douc_cond_var} is decomposed into two parts. For the first part, in \secref{sub:collision_analysis_fixed_radix}, the limit is shown to be exactly as desired and in \secref{sec:independece analysis fixed radix} the remainder part of the decomposition is shown to vanish by further analysis of the conditional independence structure of the radix-$r$ algorithm.

\subsection{Conditional independence structure of the radix-\texorpdfstring{$r$}{r} algorithm}
\label{sec:conditional independece structure fixed radix}
      

\begin{proposition}
\label{prop:fixed radix satisfies assumptions}
The matrices $\radixMatrices{r}{m}$ satisfy \assref{ass:A_k extra} for all $r\geq 2$ and $m\geq 1$. Moreover, define for all $r\geq 2$, $m\geq 1$ and $\ka\in[m]$
\begin{equation}\label{eq:def radix partitions}
\arraycolsep=1.4pt
\begin{array}{rcl}
\radixSamplePartition{r}{m}{\ka} &\defeq& \big\{\radixSamplePartitionElement{r^m}{m}{\ka}{u}:u \in [r^{m-\ka+1}]\big\},\\[.2cm]
\radixSamplePartitionElement{r^m}{m}{\ka}{u} &\defeq&  \big\{u+(q-1)r^{m-\ka+1}:q\in [r^{\ka-1}]\big\},\quad u\in[r^{m-\ka+1}].
\end{array}
\end{equation}
Then the triple $(\radixMatrices{r}{m},\radixSamplePartition{r}{m}{\ka},d)$ satisfies \assref{ass:partition and cond indep} for all $r\geq 2$, $m\geq 1$ and $\ka\in[m]$.
\end{proposition}
The proof is divided into several technical lemmata that we will prove first. The proof of \propref{prop:fixed radix satisfies assumptions} itself is postponed to the end of this section.

\begin{lemma}\label{lem:reverse product formula}
Fix $m\geq 1$, $r\geq 2$ and $\bba = \radixMatrices{r}{m}$. Then for all $k\in [m]$
\begin{enumerate}[itemsep=4pt, topsep=5pt, partopsep=0pt,label={{(\roman*)}}]
\item \label{it:tail product} $\prod_{q=0}^{k-1}A_{m-q} = \ones_{1/r^{k}}\otimes I_{r^{m-k}}$,
\item \label{it:front product} $\prod_{q=1}^{k}A_{q} = I_{r^{m-k}}\otimes \ones_{1/r^{k}}$.
\end{enumerate}
\end{lemma}
\begin{proof}
For both cases, the proof is by induction and the case $k = 1$ is obvious by \eqnref{eq:def fixed radix matrices}. To check \ref{it:tail product}, assume then that $\prod_{q=0}^{k-2}A_{m-q} = \ones_{1/r^{k-1}}\otimes I_{r^{m-(k-1)}}$ for some $k>1$. Then by \eqnref{eq:def fixed radix matrices} and the associativity and the mixed product property \eqnref{eq:mixed product property} we have
\begin{align*}
\textstyle\prod_{q=0}^{k-1}A_{m-q} 
&= \textstyle\big(\prod_{q=0}^{k-2}A_{m-q}\big)A_{m-k+1}\\
&= \big(\ones_{1/r^{k-1}}\otimes I_{r^{m-k+1}}\big)\big(I_{r^{k-1}}\otimes \ones_{1/r}\otimes I_{r^{m-k}}\big)\\
&= \big(\ones_{1/r^{k-1}}I_{r^{k-1}}\big)\otimes\big(I_{r^{m-k+1}} \big( \ones_{1/r}\otimes I_{r^{m-k}}\big)\big)\\
&= \ones_{1/r^{k-1}} \otimes \ones_{1/r}\otimes I_{r^{m-k}}\\
&= \ones_{1/r^k} \otimes I_{r^{m-k}},
\end{align*}
concluding the proof of \ref{it:tail product}. The part \ref{it:front product} follows from the proof of \lemmaref{lem:modular congruence realtion fixed}.
\end{proof}
We introduce the following additional set notation for all $N,m\geq 1$, $\bba=\bba^{(N,m)}$, $k\in[m]$ and $i\in[N]$
\begin{equation*}
\tailProdSet{m}{i}{k}{\bba}\defeq \Big\{j\in[N]:\Big(\textstyle\prod_{q=0}^{k-1} A_{m-q}\Big)^{ij} \neq 0 \Big\}.
\end{equation*}
By \eqnref{eq:parent columns} and \eqnref{eq:prime parents}, these sets admit the following special cases for all $i \in [N]$
\begin{equation}
\tailProdSet{m}{i}{1}{\bba} = \parentSet{m}{i}{\bba}{},\qquad \tailProdSet{m}{i}{m}{\bba} = \primeParentSet{m}{i}{\bba}{}.
\label{eq:part prod spec cases}
\end{equation}
\begin{lemma}\label{lem:nonzero elements}
Fix $m\geq 1$, $r\geq 2$ and $\bba = \radixMatrices{r}{m}$. Then for all $k\in[m]$ and $i\in[r^m]$
\begin{equation}\label{eq:explicit ancestors}
\parentSet{k}{i}{\bba}{m} = 
\left\{\big((i-1) \bmod r^{k-1}\big) + (q-1)r^{k-1} + r^k\floor{\frac{i-1}{r^k}}+1 : q\in[r]\right\}, 
\end{equation}
and
\begin{align}
\tailProdSet{m}{i}{k}{\bba}
&= 
\bigg\{\big((i-1)\bmod r^{m-k}\big) + (q-1)r^{m-k} + 1: q\in[r^k]\bigg\},\label{eq:tailProdSet fixed}\\
\primeParentSet{k}{i}{\bba}{m} 
&= \left\{r^k\floor{\frac{i-1}{r^k}}+q: q \in [r^k]\right\}.\label{eq:origin set}
\end{align}
Moreover, if $u_1,u_2 \in \parentSet{k}{i}{\bba}{m}$ and $u_1 \neq u_2$, then $\primeParentSet{k-1}{u_1}{\bba}{m}\cap \primeParentSet{k-1}{u_2}{\bba}{m} = \emptyset$.
\end{lemma}
\begin{proof}
We start with the element-wise definition of the Kronecker product. For any $N_1\times N_2$ matrix $A$, $M_1 \times M_2$ matrix $B$ and $0\leq \alpha < N_1M_1 $ and $0\leq \beta < N_2M_2$, we have
\begin{equation}\label{eq:kronecker inverse definition}
\left(A \otimes B\right)^{\alpha+1,\beta+1} = A^{\floor{\frac{\alpha}{M_1}}+1,\floor{\frac{\beta}{M_2}}+1}B^{(\alpha \bmod M_1)+1, (\beta \bmod M_2)+1}. 
\end{equation}
By the definition in \eqnref{eq:def fixed radix matrices}, the associativity of the Kronecker product, and two applications of \eqnref{eq:kronecker inverse definition}, we have for all $0 \leq \alpha < r^{m}$ and $0 \leq \beta < r^{m}$
\begin{align*}
A_k^{\alpha+1, \beta+1} &= I_{r^{m-k}}^{\floor{\frac{\alpha}{r^k}}+1,\floor{\frac{\beta}{r^k}}+1} \ones_{1/r}^{\big(\big\lfloor\frac{\alpha}{r^{k-1}}\big\rfloor \bmod r\big)+1,\big(\big\lfloor\frac{\beta}{r^{k-1}}\big\rfloor \bmod r\big)+1}\\
&\quad \times~I_{r^{k-1}}^{(\alpha \bmod r^{k-1})+1,(\beta \bmod r^{k-1})+1},
\end{align*}
where also the facts that $\lfloor { \lfloor {\alpha}/{r^{k-1}} \rfloor }/{r} \rfloor =  \lfloor\alpha/{r^{k}}\rfloor$ and $\lfloor { \lfloor {\beta}/{r^{k-1}} \rfloor }/{r} \rfloor =  \lfloor\beta/{r^{k}}\rfloor$ have been used. From this, by considering only the diagonal elements of the identity matrices, we have readily
\begin{align*}
&\parentSet{k}{\alpha+1}{\bba}{} = \\
&\left\{i \in [r^m] :\bigg\lfloor\frac{\alpha}{r^k}\bigg\rfloor = \bigg\lfloor\frac{i-1}{r^k}\bigg\rfloor,~ \big(\alpha \bmod r^{k-1}\big) = \big((i-1) \bmod r^{k-1}\big)\right\}.
\end{align*}
To prove the '$\supset$' part of the equation \eqnref{eq:explicit ancestors}, suppose that 
\begin{equation}
\beta = \big(\alpha \bmod r^{k-1}\big) + (q-1)r^{k-1} + r^k\big\lfloor{\alpha}/{r^k}\big\rfloor, \qquad q \in [r].
\label{eq:choice of beta}
\end{equation}
It is then simple to check by substituting the $\beta$ specified by \eqnref{eq:choice of beta} that $\lfloor{\beta}/{r^k}\rfloor = \lfloor{\alpha}/{r^k}\rfloor$ and 
\begin{equation*}
\big(\beta \bmod r^{k-1}\big) = \beta - \bigg\lfloor{\frac{\beta}{r^{k-1}}}\bigg\rfloor r^{k-1} = \big(\alpha \bmod r^{k-1}\big). 
\end{equation*}
To prove the converse inclusion, suppose that $\lfloor{\alpha}/{r^k}\rfloor = \lfloor{\beta}/{r^k}\rfloor$ and $\big(\alpha \bmod r^{k-1}\big) = \big(\beta \bmod r^{k-1}\big)$. 
Then one can check that
\begin{equation*}
\beta 
= \big(\alpha \bmod r^{k-1}\big) + r^{k}\floor{\frac{\alpha}{r^{k}}} + r^{k-1}\left(\floor{\frac{\beta}{r^{k-1}}}\bmod r\right), 
\end{equation*}
and since $\big(\lfloor{\beta}/{r^{k-1}}\rfloor\bmod r\big) + 1\in [r]$, the claim follows.

To prove \eqnref{eq:tailProdSet fixed} we have by \eqnref{eq:kronecker inverse definition} and  \lemmaref{lem:reverse product formula} for all $0 \leq \alpha,\beta < r^m$
\begin{align*}
\textstyle\left(\prod_{p=0}^{k-1}A_{m-p}\right)^{\alpha+1,\beta+1} &= \left(\ones_{1/r^k}\right)^{\floor{\frac{\alpha}{r^{m-k}}}+1,\floor{\frac{\beta}{r^{m-k}}}+1} \\
&\quad\times~\left(I_{r^{m-k}}\right)^{(\alpha \bmod r^{m-k}) + 1,(\beta \bmod r^{m-k}) +1},
\end{align*}
from which we have readily that $$\tailProdSet{m}{\alpha+1}{k}{\bba} = \big\{i \in [r^m]:\big(\alpha \bmod r^{m-k}\big) = \big((i-1) \bmod r^{m-k}\big)\big\}.$$ Take $i \in \tailProdSet{m}{\alpha+1}{k}{\bba}$, for which $i = \big(\alpha \bmod r^{m-k}\big) + \floor{{(i-1)}/{r^{m-k}}} r^{m-k} + 1$
and since $i\in[r^m]$, we have $\floor{(i-1)/{r^{m-k}}}+1 \in [r^k]$ and therefore '$\subset$' holds for \eqnref{eq:tailProdSet fixed}. To prove the converse inclusion, suppose that $i = \big(\alpha \bmod r^{m-k}\big) + (q-1)r^{m-k} + 1$, where $q \in [r^k]$. Then, by the substitution of this particular choice of $i$ one can check that $\big((i-1) \bmod r^{m-k}\big) = \big(\alpha \bmod r^{m-k}\big)$.
The equation \eqnref{eq:origin set} follows analogously by \eqnref{eq:kronecker inverse definition} and \lemmaref{lem:reverse product formula}\ref{it:front product}.

To check the empty intersection, by \eqnref{eq:explicit ancestors} and the assumption that $u_1 \neq u_2$, we have for $\ell\in\{1,2\}$
\begin{equation}
u_\ell = \big((i-1) \bmod r^{k-1}\big) + q_\ell r^{k-1} + r^k\floor{\frac{i-1}{r^k}}+1, \label{eq:sas1}
\end{equation}
where $0\leq q_1, q_2 < r$ and $q_1 \neq q_2$. Without loss of generality, we can assume $q_1 < q_2$ and by \eqnref{eq:origin set}
it suffices to show that 
\begin{equation*}
r^{k-1}\floor{\frac{u_2-1}{r^{k-1}}}+1 - r^{k-1}\floor{\frac{u_1-1}{r^{k-1}}+1} > 0,
\end{equation*}
which follows from elementary calculations using \eqnref{eq:sas1}.
\end{proof}

\begin{lemma}\label{lem:explicit analysis regarding partitions}
Fix $m\geq 1$, $r\geq 2$, $\ka\in[m]$, $\bba = \radixMatrices{r}{m}$, 
and 
let $\radixSamplePartitionElement{r^m}{m}{\ka}{u}$ be as in \eqnref{eq:def radix partitions} for all $u\in[r^{m-\ka+1}]$.
%
\begin{enumerate}[itemsep=4pt, topsep=5pt, partopsep=0pt,label={{(\roman*)}}]
\item \label{it:partition property 1 fixed} If $u_{1},u_{2} \in [r^{m-\ka+1}]$, $u_1\neq u_2$, $d>1$ and $(i,j)\in \radixSamplePartitionElement{r^m}{m}{\ka}{u_1} \times \radixSamplePartitionElement{r^m}{m}{\ka}{u_2}$, then \[\tailProdSet{m}{i}{\ka-1}{\bba}\cap \tailProdSet{m}{j}{\ka-1}{\bba} = \emptyset.\]
\item \label{it:partition property 2 fixed} If $u \in [r^{m-\ka+1}]$ and $i,j \in \radixSamplePartitionElement{r^m}{m}{\ka}{u}$, then $\tailProdSet{m}{i}{\ka}{\bba} = \tailProdSet{m}{j}{\ka}{\bba}$.
\end{enumerate}
\end{lemma}
\begin{proof}
To prove \ref{it:partition property 1 fixed} we have by \eqnref{eq:def radix partitions} $i = u_1 + (q_1-1)r^{m-\ka+1}$ and $j = u_2 + (q_2-1)r^{m-\ka+1}$ for some $q_1,q_2\in [r^{\ka-1}]$, from which it follows that $\big((i-1)\bmod r^{m-\ka+1}\big) = u_1-1$ and $\big((j-1)\bmod r^{m-\ka+1}\big) = u_2-1$. Now, suppose that $\tailProdSet{m}{i}{\ka-1}{\bba}\cap \tailProdSet{m}{j}{\ka-1}{\bba}\neq \emptyset$. Then, by \eqnref{eq:tailProdSet fixed} one must have
\begin{align*}
q'_1 - q'_2 
&= \frac{1}{r^{m-\ka+1}}\left(\big((i-1)\bmod r^{m-\ka+1}\big) - \big((j-1)\bmod r^{m-\ka+1}\big)\right)\\
&= \frac{1}{r^{m-\ka+1}}\left(u_1 - u_2\right),
\end{align*}
for some $q'_1,q'_2 \in [r^{\ka-1}]$. Since $u_1,u_2\in[r^{m-\ka+1}]$ and $u_1\neq u_2$, $(u_1 - u_2)r^{-m+\ka-1} \in (-1,1)\setminus \{0\}$ while $q'_1 - q'_2\in\integers$, which is a contradiction proving \ref{it:partition property 1 fixed}.

To prove \ref{it:partition property 2 fixed}, we observe that if $i \in \radixSamplePartitionElement{r^m}{m}{\ka}{u}$, then by \eqnref{eq:def radix partitions} $i = u + (q-1)r^{m-\ka+1}$ where $q \in [r^{\ka-1}]$ and thus
\begin{align*}
\big((i-1)\bmod r^{m-\ka} \big)
&= u - 1 + (q-1)r^{m-\ka+1} 
- \floor{\frac{u-1}{r^{m-\ka}} + (q-1)r}r^{m-\ka} \\
&= u - 1 - \floor{\frac{u-1}{r^{m-\ka}}}r^{m-\ka}, 
\end{align*}
Since, the same can be repeated for $j \in \radixSamplePartitionElement{r^m}{m}{\ka}{u}$, we have $\big((i-1)\bmod r^{m-\ka}\big) = \big((j-1)\bmod r^{m-\ka}\big)$
which, by \eqnref{eq:tailProdSet fixed} is sufficient for \ref{it:partition property 2 fixed} to hold.
\end{proof}

\begin{lemma}
\label{lem:cond indep propagation}
Fix $N,m\geq 1$, $\bba = \bba^{(N,m)}$, $\calG \subset \calF$ and $k,k'\in[m]$. If $\bba$ satisfies \assref{ass:A_k} and for some $i,j\in[N]$
\begin{equation}\label{eq:codnitional independence criteria}
\arraycolsep=1.4pt
\begin{array}{ll}
\parentSet{k}{i}{\bba}{} \cap \parentSet{k'}{j}{\bba}{} &= \emptyset,\\[.1cm]
 \parentSet{k}{i}{\bba}{} \times \parentSet{k'}{j}{\bba}{} &\subset \ciSetGen{k-1}{k'-1}{\calG}{N},
\end{array}
\end{equation} 
then $(i,j) \in \ciSetGen{k}{k'}{\calG}{N}$.
\end{lemma}
\begin{proof}
By \assref{ass:A_k} we can use \eqnref{eq:rigorus formulation of xi}, \eqnref{eq:codnitional independence criteria} and the law of total probability, for all $i,j\in[N]$ and $S^i,S^j \in \sss$
\begin{align*}
& \P\big(\xi^{i}_k \in S^i,~ \xi^{j}_{k'} \in S^j \big| \calG\big)\\
&= \displaystyle\sum_{\ell_i\in\parentSet{k}{i}{\bba}{}}\sum_{\ell_j\in\parentSet{k'}{j}{\bba}{}} \P\big(I^{i}_k = \ell_i,~ \xi^{\ell_i}_{k-1} \in S^i,~ I^j_{k'} = \ell_j,~ \xi^{\ell_j}_{k'-1} \in S^j \big| \calG\big)\\ 
&= \displaystyle\sum_{\ell_i\in\parentSet{k}{i}{\bba}{}} \P\big(I^{i}_k = \ell_i,~ \xi^{\ell_i}_{k-1} \in S^i \big| \calG\big)\sum_{\ell_j\in\parentSet{k'}{j}{\bba}{}} \P\big(I^j_{k'} = \ell_j,~ \xi^{\ell_j}_{k'-1} \in S^j \big| \calG\big)\\ 
&= \P\big(\xi^{i}_k \in S^i\,\big|\, \calG\big)\P\big( \xi^{j}_{k'} \in S^j \,\big|\, \calG\big),	
\end{align*}
concluding the proof.
\end{proof}
To prove \ref{it:cond indep} of \assref{ass:partition and cond indep} we need the following conditional independence result.
\begin{lemma}\label{lem:final conditional independence}
Fix $N,m\geq 1$ and $\bba = \bba^{(N,m)}$. If $\bba$ satisfies \assref{ass:A_k extra} and $\tailProdSet{m}{i}{\ka-1}{\bba}\cap\tailProdSet{m}{j}{\ka-1}{\bba} = \emptyset$ for some $i,j \in [N]$ and $1<\ka\leq m$, then $\xi^{i}_m\independent \xi^j_m\big| \xi_0,\ldots,\xi_{m-\ka}$.
\end{lemma}
\begin{proof}
In the case $\ka=2$, by \eqnref{eq:part prod spec cases}, $\tailProdSet{m}{i}{\ka-1}{\bba}\cap \tailProdSet{m}{j}{\ka-1}{\bba} = \parentSet{m}{i}{\bba}{}\cap \parentSet{m}{j}{\bba}{} = \emptyset$ and because by the one step conditional independence we also have 
\begin{eqnarray*}
\parentSet{m}{i}{\bba}{}\times \parentSet{m}{j}{\bba}{} \subset \ciSet{m-1}{\xi_0,\ldots,\xi_{m-2}}{N},
\end{eqnarray*}
the claim holds by \lemmaref{lem:cond indep propagation} for $\ka=2$. 

In order to prove the claim for $2<\ka\leq m$, we first show that if for any $k' \in [m-2]$ and $\calG \subset \calF$, one has
\begin{equation}
\begin{array}{rl}
&\tailProdSet{m}{i}{k+1}{\bba}\cap \tailProdSet{m}{j}{k+1}{\bba} = \emptyset,\\[.1cm]
&\tailProdSet{m}{i}{k+1}{\bba}\times \tailProdSet{m}{j}{k+1}{\bba} \subset \ciSet{m-k-1}{\calG}{N},\label{it:ass1}
\end{array}
\end{equation}
where $k=k'$, then \eqnref{it:ass1} is also true for $k = k'-1$.
To do this, we first observe that by \eqnref{it:ass1} and \assref{ass:A_k extra}\ref{it:commutativity}
\begin{align}
\sum_{\ell}  A_{m-k:m}^{i\ell}A_{m-k:m}^{j \ell} 
&= {\sum_{\substack{(u,v) \in \tailProdSet{m}{i}{k}{\bba}\times\tailProdSet{m}{j}{k}{\bba}\\ u\neq v }}A_{m-k+1:m}^{iu}A_{m-k+1:m}^{j v}\sum_{\ell} A_{m-k}^{u\ell}A_{m-k}^{v\ell}}\nonumber \\
&\quad +~{
\sum_{u}A_{m-k+1:m}^{iu}A_{m-k+1:m}^{j u}\sum_{\ell} A_{m-k}^{u\ell}A_{m-k}^{u\ell}}=0.\label{eq:propagation of separation}
\end{align}
In the second sum of the decomposition, by \assref{ass:A_k}, $\sum_{\ell} A_{m-k}^{u\ell}A_{m-k}^{u\ell}>0$, and from this we conclude that $\sum_{\ell}\textstyle\left(\prod_{q=0}^{k-1} A_{m-q}\right)^{i\ell}\textstyle\left(\prod_{q=0}^{k-1} A_{m-q}\right)^{j\ell} = 0$, which is equivalent to $\tailProdSet{m}{i}{k}{\bba} \cap \tailProdSet{m}{j}{k}{\bba} = \emptyset$, proving the first part of \eqnref{it:ass1} for $k = k'-1$.
To prove the second part, we show that for all $(p,q) \in \tailProdSet{m}{i}{k}{\bba} \times \tailProdSet{m}{j}{k}{\bba}$
\begin{equation}
\begin{array}{rl}
&\parentSet{m-k}{p}{\bba}{} \cap \parentSet{m-k}{q}{\bba}{} = \emptyset, \\[.1cm]
&\parentSet{m-k}{p}{\bba}{}\times\parentSet{m-k}{q}{\bba}{} \subset \ciSet{m-k-1}{\calG}{N}.
\label{eq:part 1 for cond indep}
\end{array}
\end{equation}
To see this, we observe that in the first sum of the decomposition \eqnref{eq:propagation of separation}, $A_{m-k+1:m}^{iu}>0$ and $A_{m-k+1:m}^{j v}>0$, and hence by the non-negativity of the matrices $\aMatrices{N}{m}$, one must also have $\sum_{\ell} A_{m-k}^{u\ell}A_{m-k}^{v\ell}=0$ which is equivalent to $\parentSet{m-k}{u}{\bba}{} \cap \parentSet{m-k}{v}{\bba}{} = \emptyset$. This establishes the first part of \eqnref{eq:part 1 for cond indep}. To prove the second part of \eqnref{eq:part 1 for cond indep},
one can check that by definitions $\tailProdSet{m}{i}{k+1}{\bba} = \bigcup_{\ell\in \tailProdSet{m}{i}{k}{\bba}}  \parentSet{m-k}{\ell}{\bba}{}$ and hence by \eqnref{it:ass1}, also the conditional independence in \eqnref{eq:part 1 for cond indep} holds for all $(p,q) \in \tailProdSet{m}{i}{k}{\bba} \times \tailProdSet{m}{j}{k}{\bba}$.
%
Finally the conditional independence in \eqnref{it:ass1} for $k = k'-1$ follows by \eqnref{eq:part 1 for cond indep} and \lemmaref{lem:cond indep propagation}.


By assumption, $\tailProdSet{}{i}{\ka-1}{\bba}\cap\tailProdSet{}{j}{\ka-1}{\bba}=\emptyset$ and by the one step conditional independence \eqnref{it:ass1} holds for $k=\ka-2$ and $\calG = \sigma(\xi_0,\ldots,\xi_{m-\ka})$. By \eqnref{eq:part prod spec cases}, $\tailProdSet{}{i}{1}{\bba}=\parentSet{m}{i}{\bba}{}$ and $\tailProdSet{}{j}{1}{\bba} =\parentSet{m}{j}{\bba}{}$. Thus by the backward induction enabled by \eqnref{it:ass1} we have,  
%
\begin{eqnarray*}
\parentSet{m}{i}{\bba}{}\cap\parentSet{m}{j}{\bba}{}&=&\emptyset,\\
\parentSet{m}{i}{\bba}{}\times \parentSet{m}{j}{\bba}{} &\subset& \ciSet{m-1}{\xi_{0},\ldots,\xi_{m-\ka}}{N},
\end{eqnarray*}
from which the claim then follows by \lemmaref{lem:cond indep propagation}.
\end{proof}

\begin{proof}[Proof of \propref{prop:fixed radix satisfies assumptions}]
First we prove that $\bba = \radixMatrices{r}{m}$ satisfies \assref{ass:A_k extra}. 
\assref{ass:A_k} and parts \ref{it:symmetry}, \ref{it:idempotence} of \assref{ass:A_k extra} follow from the proof of \lemmaref{lem:modular congruence realtion fixed}. \assref{ass:A_k extra}\ref{it:commutativity} can be checked by using the mixed product property \eqnref{eq:mixed product property}. \assref{ass:A_k extra}\ref{it:cardinality invariance} follows from \eqnref{eq:explicit ancestors}.


To prove the only non-trivial condition \assref{ass:A_k extra}\ref{it:unique paths}, we prove that if there are $(i_0,\ldots,i_m)$ and $(j_0,\ldots,j_m)$ in  $\paths{\bba}{}$ such that for some $p\in[m]$, one has $i_p = j_p$ and $i_{p-1} \neq j_{p-1}$ then $i_q \neq j_q$ for all $q < p$. To do this, suppose that $i_q = j_q$ for some $q<p\in[m]$. From the definition \eqnref{eq:prime parents} it follows that for any $k\in [m]$, $\primeParentSet{k-1}{i_{k-1}}{\bba}{m} \subset \primeParentSet{k}{i_k}{\bba}{m}$. Therefore $\primeParentSet{p-1}{i_{p-1}}{\bba}{m} \supset \primeParentSet{q}{i_q}{\bba}{m} = \primeParentSet{q}{j_q}{\bba}{m} \subset \primeParentSet{p-1}{j_{p-1}}{\bba}{m}$, which is a contradiction with $\primeParentSet{p-1}{i_{p-1}}{\bba}{m} \cap \primeParentSet{p-1}{j_{p-1}}{\bba}{m} = \emptyset$, which we know by \lemmaref{lem:nonzero elements} since $i_{p-1},j_{p-1}\in\parentSet{p}{i_p}{\bba}{}$ and $i_{p-1} \neq j_{p-1}$.

It remains to prove that $(\bba,\radixSamplePartition{r}{m}{\ka},d)$ satisfies \assref{ass:partition and cond indep}. \assref{ass:partition and cond indep}\ref{it:is partition} follows from \eqnref{eq:def radix partitions}, and \assref{ass:partition and cond indep}\ref{it:equivalence classes} follows from \lemmaref{lem:explicit analysis regarding partitions}\ref{it:partition property 2 fixed} since $\big(\prod_{k=0}^{\ka-1}A_{m-k}\big)^{ij} \in \{0,r^{-\ka}\}$. To verify \assref{ass:partition and cond indep}\ref{it:cond indep}, we observe first that for $\ka=1$, the claim follows trivially by the one step conditional independence. For $1<\ka\leq m$ we observe that by \lemmaref{lem:explicit analysis regarding partitions}\ref{it:partition property 1 fixed}, if $(i,j) \in \radixSamplePartitionElement{r^m}{m}{\ka}{u_1}\times \radixSamplePartitionElement{r^m}{m}{\ka}{u_2}$ where $(u_1,u_2)\in [r^{m-\ka+1}]^2$ such that $u_1 \neq u_2$,  then $\tailProdSet{m}{i}{\ka-1}{\bba} \cap \tailProdSet{m}{j}{\ka-1}{\bba} = \emptyset$, and the claim thus follows from \lemmaref{lem:final conditional independence}.
\end{proof}

\subsection{Convergence of the conditional variance}
\label{sub:collision_analysis_fixed_radix}

The main result of this section is the following proposition whose proof is postponed to the end of this section.
\begin{proposition}\label{prop:var_conv_fixed_radix}
Under the hypotheses of \theref{thm:butterfly_clt_fixed_radix},
\begin{equation*}
\E \Bigg[\Bigg(\sum_{\varrho=1}^{r^m m}X_\varrho^{(r^m,m)}\Bigg)^2\Bigg|\F_0^{(r^m,m)}\Bigg]
\inprob{m\rightarrow\infty}\left(1-\frac{1}{r}\right)\mu\Bigg(g\bigg(\varphi-\frac{\mu(g\varphi)}{\mu(g)}\bigg)^2\Bigg)\mu(g).
\end{equation*}
\end{proposition}

\newcommand{\I}[2]{\calI^{#1}_{#2}}
\newcommand{\ffi}{\varphi(\xi^i_0)}
\newcommand{\ffj}{\varphi'(\xi^j_0)}
\newcommand{\mfi}{\mu(\varphi)}
\newcommand{\mfj}{\mu(\varphi')}
\newcommand{\ffisq}{\varphi^2(\xi^i_0)}

In order to prove \propref{prop:var_conv_fixed_radix} we need the following auxiliary result which is the main application of the block-wise absolute second moment bound hypothesis in \theref{thm:butterfly_clt_fixed_radix}.

\begin{lemma}
\label{eq:convergence of weigted subsample integrals}
Under the hypotheses of \theref{thm:butterfly_clt_fixed_radix},
for all $\varphi,\varphi'\in\boundMeas{\ss}$
\begin{equation*}
e_m(\varphi,\varphi') \defeq \frac{1}{m}\sum_{k=1}^m\frac{1}{r^m}\sum_{i}\frac{1}{\card{\collisionStartSetSh{k}{i}{\bba(m)}}}\sum_{j\in \collisionStartSetSh{k}{i}{\bba(m)}}
\ffi\ffj - \mfi\mfj \inprob{m\to\infty} 0,
\end{equation*}
where $\bba(m)=\radixMatrices{r}{m}$.
\end{lemma}
\begin{proof}
By defining 
\begin{align*}
A_m &\defeq \frac{1}{m}\sum_{k=1}^m\frac{1}{r^m}\sum_{i}\ffi\frac{1}{\Big|\collisionStartSetSh{k}{i}{\bba(m)}\Big|}\sum_{j\in \collisionStartSetSh{k}{i}{\bba(m)}}
\Big(\ffj-\mfj\Big),\\
B_m &\defeq \Bigg(\frac{1}{r^m}\sum_{i}\ffi\Bigg)\mfj - \mfi\mfj,
\end{align*}
we have the decomposition $\abs{e_m(\varphi,\varphi')} = \abs{A_m + B_m} \leq \abs{A_m} + \abs{B_m}$.
From the hypotheses of \theref{thm:butterfly_clt_fixed_radix} it follows that if we set $\ka=1$ and $q=1$ in \eqnref{eq:required conv in prob}, then for all $m\geq 1$
\begin{equation}\label{eq:lp convergence}
\E\Bigg[\Bigg|\frac{1}{r^{m}}\sum_{j} \varphi(\xi^j_0) - \mu(\varphi)\Bigg|^{2}\Bigg]^{\frac{1}{2}} 
\leq b(\varphi)\sqrt{\frac{m}{r^m}},
\end{equation}
implying that $\abs{B_m}$ converges to zero in probability as $m\to\infty$. To prove the same for $\abs{A_m}$ we apply triangle inequality, Cauchy-Schwartz inequality and Jensen's inequality, yielding
\begin{align}
& \E\left[\abs{A_m}\right]\nonumber\\
&\leq \frac{1}{m}\sum_{k=1}^m\E\left[\sqrt{\Bigg|\frac{1}{r^m}\sum_{i}\ffi\frac{1}{\big|\collisionStartSetSh{k}{i}{\bba(m)}\big|}\sum_{j\in \collisionStartSetSh{k}{i}{\bba(m)}}
\left(\ffj-\mfj\right)\Bigg|^2}\right]\nonumber\\
&\leq\frac{\infnorm{\varphi}}{m}\sum_{k=1}^m \E\left[\sqrt{\frac{1}{r^m}\sum_i\Bigg(\frac{1}{\big|\collisionStartSetSh{k}{i}{\bba(m)}\big|}\sum_{j\in \collisionStartSetSh{k}{i}{\bba(m)}}
\left(\ffj-\mfj\right)\Bigg)^2}\right]\nonumber\\
&\leq \frac{\norm{\varphi}_\infty}{m}\sum_{k=1}^{m} \sqrt{\frac{1}{r^m}\sum_{i}\E\left[\Bigg(\frac{1}{\big|\collisionStartSetSh{k}{i}{\bba(m)}\big|}\sum_{j\in\collisionStartSetSh{k}{i}{\bba(m)}}
\left(\ffj-\mfj\right)\Bigg)^2\right]}.\label{eq:intermediate upperbound 1}
\end{align}
By reversing the summation order in the last sum, we need to consider the sets $\collisionStartSetSh{m-k+1}{i}{\bba(m)}$. Using \eqnref{eq:origin set} one can check that 
\begin{align}
&\collisionStartSetSh{m-k+1}{i}{\bba(m)} = \{j + (q(p)-1)r^{m-k}: p \in [r]\setminus \{p^{\ast}\},~j \in [r^{m-k}]\} \label{eq:decomposition of I}
\end{align}
where
\begin{equation*}
p^{\ast} \defeq \bigg(\floor{\frac{(i-1)}{r^{m-k}}}\bmod r\bigg) + 1, \quad q(p) \defeq r\floor{\frac{(i-1)}{r^{m-k+1}}} + p
\end{equation*}
from which we readily have
\begin{equation}\label{eq:card collisionStartSetShort}
\big|\collisionStartSetSh{m-k+1}{i}{\bba(m)}\big| = (r-1)r^{m-k}.
\end{equation}
Note that $p^{\ast}$ and $q(p)$ both depend on $m$, $k$ and $i$ but in the following we will consider these quantities for fixed $m$, $k$ and $i$ only.

Because $\big|\collisionStartSetSh{1}{i}{\bba(m)}\big| = r-1$ we have 
\begin{equation}\label{eq:case k is m}
\frac{1}{m} \sqrt{\frac{1}{r^m}\sum_i\E\left[\Bigg(\frac{1}{\big|\collisionStartSetSh{1}{i}{\bba(m)}\big|}\sum_{j\in \collisionStartSetSh{1}{i}{\bba(m)}}
\left(\ffj-\mfj\right)\Bigg)^2\right]} 
\leq 2\frac{\infnorm{\varphi'}}{m},
\end{equation}
for the term $k=m$ in the sum in \eqnref{eq:intermediate upperbound 1}. 

For all $m\geq 1$, $k \in [m-1]$, $i\in [r^m]$ and $p \in [r]$, we have $k+1\in [m]$ and $q(p)\in[r^{(k+1)-1}]$. 
Hence by \eqnref{eq:decomposition of I}, \eqnref{eq:card collisionStartSetShort}, Minkowski's inequality and \eqnref{eq:required conv in prob} that
\begin{align*}
&\E\left[\Bigg(\frac{1}{\big|\collisionStartSetSh{m-k+1}{i}{\bba(m)}\big|}\sum_{j\in \collisionStartSetSh{m-k+1}{i}{\bba(m)}}
\left(\ffj-\mfj\right)\Bigg)^2\right]^{\frac{1}{2}} \\
&\leq
\frac{1}{r-1}\sum_{p\in[r]\setminus \{p^{\ast}\}}\E\left[\Bigg(\frac{1}{r^{m-(k+1)+1}}\sum_{j \in [r^{m-(k+1)+1}]}
\left(\varphi'(\xi^{J(j)}_0)-\mu(\varphi)\right)\Bigg)^2\right]^{\frac{1}{2}}\\
&\leq b(\varphi')\sqrt{\frac{m-k-1}{r^m} + \frac{1}{r^{m-k}}}.
\end{align*}
where $J(j) = j + (q(p) - 1 )r^{m-(k+1)-1}$.
By substituting this and \eqnref{eq:case k is m} into \eqnref{eq:intermediate upperbound 1} we have for all $m\geq 1$
\begin{eqnarray*}
\E\left[\abs{A_m}\right] \leq 2\frac{\norm{\varphi}^2_\infty}{m} + 
b(\varphi')\norm{\varphi}_\infty\frac{1}{m}\sum_{k=1}^{m-1} \sqrt{\frac{m-k-1}{r^m} + \frac{1}{r^{m-k}}},
\end{eqnarray*}
and by Cauchy-Schwartz inequality we have
\begin{align*}
&\left(\frac{1}{m}\sum_{k=1}^{m-1} \sqrt{\frac{m-k-1}{r^m} + \frac{1}{r^{m-k}}}\right)^2\\
&\leq {\frac{1}{m}\sum_{k=1}^{m-1} \left(\frac{m-k-1}{r^m} + \frac{1}{r^{m-k}}\right)}\\
&\leq{\frac{(m-1)(m-2)}{mr^m} + \frac{1}{m}\left(\frac{1-r^{-m}}{1-r^{-1}}-1\right)}\surely{m\to\infty} 0,
\end{align*}
implying that $\E\left[\abs{A_m}\right]$ converges to zero as $m\to\infty$ which concludes the proof.
\end{proof}
\begin{proof}[Proof of \propref{prop:var_conv_fixed_radix}]
Because $r\geq 2$ is assumed fixed, let us write $\bba(m) = \radixMatrices{r}{m}$. First we observe that by \eqnref{eq:origin set}, \eqnref{eq:explicit ancestors} and \lemmaref{lem:bijectivity g} we have for all $k\in [m]$, $i,u_0\in[r^m]$
\begin{equation*}
\big|\primeParentSet{k}{i}{\bba(m)}{}\big| = r^k,\qquad \big|\parentSet{k}{i}{\bba(m)}{}\big| = r,\qquad \big|\lowerPartFrom{k}{u_0}{\bba(m)}\big| = r^{m-k},
\end{equation*}
and by \propref{prop:fixed radix satisfies assumptions} we can apply \propref{prop:tensor product formulation} and by substitution we have
\begin{align}\label{eq:fixed radix conditional second moment}
&\E\left[\left(\displaystyle\sum_{\varrho=1}^{mr^m} X^{(r^m,m)}_\varrho\right)^2 \mids \calF^{(r^m,m)}_0\right] 
= \dfrac{r^m}{m}\dfrac{1}{r^{2m}}\displaystyle\sum_{i} g^{2}(\bxi{i}{0}{}{})\cvarphi^2_{r^m}(\bxi{i_0}{0}{}{}) \\
&\quad+~ \dfrac{r^m}{m}\dfrac{1}{r^{2m}}\displaystyle\sum_{k=1}^m\sum_{i}\sum_{j\neq i}g(\bxi{i}{0}{}{})\cvarphi^2_{r^m}(\bxi{i}{0}{}{})g(\bxi{j}{0}{}{})\ind{}\Big(j\in \collisionStartSetSh{k}{i}{\bba(m)}\Big)r^{-k}\nonumber\\
&\quad+~ \dfrac{r^m}{m}\dfrac{1}{r^{2m}}\displaystyle\sum_{k=1}^m\sum_{i}\sum_{j\neq i}g(\bxi{i}{0}{}{})\cvarphi_{r^{m}}(\bxi{i}{0}{}{})g(\bxi{j}{0}{}{})\cvarphi_{r^{m}}(\bxi{j}{0}{}{})\ind{}\Big(j\in \collisionStartSetSh{k}{i}{\bba(m)}\Big)(1-r^{-k}).\nonumber
\end{align}
The three nested sums on the r.h.s.~will each be considered separately. For the first sum, we have by \eqnref{eq:lp convergence} and the continuous mapping theorem
\begin{equation}\label{eq:first nested sum}
\frac{r^m}{m}\frac{1}{r^{2m}}\sum_{i} g^{2}(\bxi{i}{0}{}{})\cvarphi^2_{r^m}(\bxi{i}{0}{}{})\inprob{m\to\infty} 0.
\end{equation}
For the second sum we see by normalizing the nested sums and by using \eqnref{eq:card collisionStartSetShort} that
\begin{align}
&\frac{r^m}{m}\frac{1}{r^{2m}}\sum_{k=1}^m\sum_{i}\sum_{j\neq i}g(\bxi{i}{0}{}{})\cvarphi^2_{r^m}(\bxi{i}{0}{}{})g(\bxi{j}{0}{}{})\ind{}\Big(j \in \collisionStartSetSh{k}{i}{\bba(m)}\Big)r^{-k} \nonumber\\
&=\left(1-\frac{1}{r}\right)\frac{1}{m}\sum_{k=1}^m\frac{1}{r^m}\sum_{i} \frac{1}{\card{\collisionStartSetSh{k}{i}{\bba(m)}}}\sum_{j\in\collisionStartSetSh{k}{i}{\bba(m)}}g(\bxi{i}{0}{}{})\cvarphi^2_{r^m}(\bxi{i}{0}{}{})g(\bxi{j}{0}{}{})\nonumber \\
& \inprob{m\to\infty} \left(1-\frac{1}{r}\right)\mu(g\cvarphi^2)\mu(g),\label{eq:second part limit statement}
\end{align}
where the convergence follows from \lemmaref{eq:convergence of weigted subsample integrals} and several applications of \eqnref{eq:lp convergence} and continuous mapping theorem.

For the third sum we define
\begin{align*}
A_m &\defeq \frac{r^m}{m}\frac{1}{r^{2m}}\sum_{k=1}^m\sum_{i}\sum_{j}g(\bxi{i}{0}{}{})\cvarphi_{r^m}(\bxi{i}{0}{}{})g(\bxi{j}{0}{}{})\cvarphi_{r^m}(\bxi{j}{0}{}{})\ind{}\Big(j \in \collisionStartSetSh{k}{i}{\bba(m)}\Big),\\
B_m &\defeq \frac{r^m}{m}\frac{1}{r^{2m}}\sum_{k=1}^m\sum_{i}\sum_{j}g(\bxi{i}{0}{}{})\cvarphi_{r^m}(\bxi{i}{0}{}{})g(\bxi{j}{0}{}{})\cvarphi_{r^m}(\bxi{j}{0}{}{})\ind{}\Big(j \in \collisionStartSetSh{k}{i}{\bba(m)}\Big)r^{-k},
\end{align*}
and show that
\begin{eqnarray}
&&\frac{r^m}{m}\frac{1}{r^{2m}}\sum_{k=1}^m\sum_{i}\sum_{j\neq i}g(\bxi{i}{0}{}{})\cvarphi_{r^m}(\bxi{i}{0}{}{})g(\bxi{j}{0}{}{})\cvarphi_{r^m}(\bxi{j}{0}{}{})\ind{}\Big(j \in \collisionStartSetSh{k}{i}{\bba(m)}\Big)(1-r^{-k}) \nonumber\\
&&\qquad = A_m + B_m \inprob{m\to\infty} 0.\label{eq:third nested sum}
\end{eqnarray}
To do this, for $A_m$ we use \lemmaref{lem:general properties of path sets}\ref{it:complement decomposition} by which
\begin{align}
A_m 
&= \frac{r^{m}}{m}\frac{1}{r^{2m}}\sum_{i}\sum_{j\neq i}g(\bxi{i}{0}{}{})\cvarphi_{r^m}(\bxi{i}{0}{}{})g(\bxi{j}{0}{}{})\cvarphi_{r^m}(\bxi{j}{0}{}{})\nonumber\\
&= \frac{r^{m}}{m}\frac{1}{r^{2m}}\Bigg(\sum_{i}g(\bxi{i}{0}{}{})\cvarphi_{r^m}(\bxi{i}{0}{}{})\Bigg(\sum_{j}g(\bxi{j}{0}{}{})\cvarphi_{r^m}(\bxi{j}{0}{}{}) - g(\bxi{i}{0}{}{})\cvarphi_{r^m}(\bxi{i}{0}{}{})\Bigg)\Bigg)\nonumber\\
&=
\frac{r^{m}}{m}\left(\frac{1}{r^m}\sum_{i}g(\bxi{i}{0}{}{})\cvarphi_{r^m}(\bxi{i}{0}{}{})\right)^2 - 
\frac{1}{m}\frac{1}{r^m}\sum_{j}g^2(\bxi{j}{0}{}{})\cvarphi_{r^m}^2(\bxi{j}{0}{}{}).\label{eq:tmp final form}
\end{align}
Because for the first term in \eqnref{eq:tmp final form} we have ${r^{-m}}\sum_{i}g(\xi_0^i)\cvarphi_{r^m}(\xi_0^i) = 0$
and for the second term we have ${r^{-m}}\sum_{i}g^2_{0}(\xi_0^i)\cvarphi^2_{r^m}(\xi_0^i)\leq \infnorm{g}^2\osc{\varphi}^2$.
Hence we see that $\abs{A_m}$ converges to zero in probability as $m\to\infty$. For $B_m$ we have similarly as for \eqnref{eq:second part limit statement} that
\begin{equation*}
B_m \inprob{m\to\infty} \left(1-\frac{1}{r}\right)\mu(g\cvarphi)\mu(g\cvarphi) = 0.
\end{equation*}
The proof is completed by combining \eqnref{eq:first nested sum}, \eqnref{eq:second part limit statement}, \eqnref{eq:third nested sum} and \eqnref{eq:fixed radix conditional second moment}.
\end{proof}

\subsection{Approximation of the conditional variance and independence analysis}\label{sub:independence_analysis_fixed_radix}
\label{sec:independece analysis fixed radix}

The main result of this section is the following proposition, which is the last remaining part in completing the proof of \theref{thm:butterfly_clt_fixed_radix}.
\begin{proposition}\label{prop:var_vanish_fixed_radix}
Under the hypotheses of \theref{thm:butterfly_clt_fixed_radix},
\begin{equation*}
\sum_{\varrho=1}^{r^m m}\bigg(\E \bigg[ \Big(X_\varrho^{(r^m,m)}\Big)^2 \bigg| \F_{\varrho-1}^{(r^m,m)} \bigg]- \E \bigg[ \Big(X_\varrho^{(r^m,m)}\Big)^2\bigg|\F_0^{(r^m,m)}\bigg]\bigg)\inprob{m\rightarrow\infty} 0.
\end{equation*}
\end{proposition}
\begin{proof}
We take $Z_\varrho^{(r^m,m)}$ to be as defined in \eqnref{eq:def Z} of the proof of \theref{thm:butterfly_clt_fixed_radix}. By Markov's inequality, for any $\epsilon>0$
\begin{align}
&\P\Bigg(\bigg|\sum_{\varrho\in[r^m m]}Z_\varrho^{(r^m,m)}\bigg|\geq \epsilon\Bigg) \leq \dfrac{1}{\epsilon^2}\E\Bigg[\bigg(\sum_{\varrho\in[r^m m]}Z_\varrho^{(r^m,m)}\bigg)^2\Bigg]\nonumber\\
&=\dfrac{1}{\epsilon^2}\sum_{\varrho=1}^{mr^m}\E\bigg[\left(Z_\varrho^{(r^m,m)}\right)^2\bigg]\label{eq:fixed_radix_P_of_Z} 
+\dfrac{1}{\epsilon^2}\sum_{\varrho=1}^{mr^m}\sum_{\varrho'\neq \varrho}\E\bigg[Z_\varrho^{(r^m,m)} Z_{\varrho'}^{(r^m,m)}\bigg].
\end{align}
By \propref{prop:generalized martingale decomposition}, $\big|X_\varrho^{(r^{m},m)}\big|\leq (r^m m)^{-1/2}\infnorm{g}\osc{\varphi}$. Therefore
for any $\varrho,\varrho'\in[r^m m]$,
\begin{equation}
\left|Z_\varrho^{(r^m,m)} Z_{\varrho'}^{(r^m,m)}\right| \leq 4  (r^m m)^{-2} \infnorm{g}^4 \osc{\varphi}^4.\label{eq:Z_times_Z_bounded}
\end{equation}
Therefore the first term on the r.h.s.~of \eqref{eq:fixed_radix_P_of_Z} converges to zero as $m\rightarrow  \infty$. It remains to establish the convergence of the second term. This is not equally straightforward as the number of cross terms is of order $(r^m m)^2$ and therefore the reasoning applied to the first term does not work without additional delicacy. The key step, which we shall take next, is to establish that a suitably large proportion of the terms $\E\left[Z_\varrho^{(r^m,m)} Z_{\varrho'}^{(r^m,m)}\right]$ are in fact zero.

To proceed, we observe that if $\varrho,\varrho'\in[r^m m]$ are such that $Z_\varrho^{(r^m,m)}$ and $Z_{\varrho'}^{(r^m,m)}$ are conditionally independent given $\F_0^{(r^m,m)}$, then by the tower property and \propref{prop:generalized martingale decomposition}\ref{eq:zero expectations}, $\E\left[Z_\varrho^{(r^m,m)} Z_{\varrho'}^{(r^m,m)}\right]=0$.

There are altogether $m^2r^{2m}-m r^m$ pairs $\big(Z_\varrho^{(r^m,m)}, Z_{\varrho'}^{(r^m,m)}\big)$ with $\varrho\neq \varrho'$, and by \lemmaref{lem:independent pair count fixed}, there are at most
\begin{equation}\label{eq:ipair bound}
a_m = m^2 r^{2m}-m r^m - r(r-1)\sum_{i=0}^{m-2} (i+1)^2 r^{m+i}
\end{equation}
pairs which are not conditionally independent given $\F_0^{(r^m,m)}$. Therefore in order to establish that the second term on the r.h.s.~of \eqref{eq:fixed_radix_P_of_Z} converges to zero as $m\rightarrow\infty$, it is enough to once again apply $\eqref{eq:Z_times_Z_bounded}$, and check that
\begin{equation}
\lim_{m\rightarrow 0} \frac{a_m}{m^2 r^{2m}}= 0.
\label{eq:ipair limit}
\end{equation}
By shifting the summation index, reversing the summation order and expanding the square expression, we have
\begin{eqnarray}\label{eq:ipair expansion}
r\sum_{i=0}^{m-2} (i+1)^2 r^{m+i} &=& r^{2m}\sum_{i=1}^{m-1}(m-i)^2r^{-i}\nonumber \\
&=& r^{2m}\Bigg(m^2\sum_{i=1}^{m-1}\frac{1}{r^i} - 2m\sum_{i=1}^{m-1}\frac{i}{r^i} + \sum_{i=1}^{m-1}\frac{i^2}{r^i}\Bigg),
\end{eqnarray}
where each of the three sums converges to a finite value as $m\to\infty$. By elementary calculations one can then check that \eqnref{eq:ipair limit} follows by combining \eqnref{eq:ipair bound} and \eqnref{eq:ipair expansion}.
\end{proof}


%
\begin{figure}
\begin{center}
\tikzset{
    vertex/.style = {
    	draw,
	    circle,
      fill      = white,
      outer sep = 2pt,
      inner sep = 2pt,
    }
}

\tikzset{
    mainvertex/.style = {
    	draw,
		circle,
        fill = black,
        outer sep = 2pt,
        inner sep = 2pt,
    }
}

\tikzset{
    demovertex/.style = {
    	draw,
		circle,
        fill = gray,
        outer sep = 2pt,
        inner sep = 2pt,
    }
}

\tikzset{
    highlightnode/.style = {
    	draw,
			circle,
      fill = black,
      outer sep = 2pt,
      inner sep = 2pt,
    }
}

\tikzset{
    demoedge/.style = {
		-stealthnew,
		shorten <=.15cm, 
		shorten >=.15cm,
		arrowhead=2mm,
		line width=2pt    
		}
}

\tikzset{
    basicedge/.style = {
		-stealthnew,
		color = black,
		shorten <=.15cm, 
		shorten >=.15cm,
		arrowhead=1mm,
		line width=.5pt    
		}
}

\begin{tikzpicture}

\def \vstep {-.75cm}
\def \hstep {.75cm}
\def \r {2}
\def \m {3} 
\def \N {8}
\def \X {4}
\def \Y {2}

\node[draw,minimum height=4*-\vstep, minimum width=7.9*\hstep, %
  fill = black!10!,
  rounded corners=3pt] at (4.5*\hstep,1.5*\vstep) {}; %
\node[draw,minimum height=4*-\vstep, minimum width=7.9*\hstep, %
  fill = black!10!,
  rounded corners=3pt] at (12.5*\hstep,1.5*\vstep) {}; %
  
\node[draw,minimum height=2.75*-\vstep, minimum width=3.7*\hstep, %
  fill = black!20!,
  rounded corners=2pt] at (2.5*\hstep,1*\vstep) {}; %
\node[draw,minimum height=2.75*-\vstep, minimum width=3.7*\hstep, %
  fill = black!20!,
  rounded corners=2pt] at (6.5*\hstep,1*\vstep) {}; %
\node[draw,minimum height=2.75*-\vstep, minimum width=3.7*\hstep, %
  fill = black!20!,
  rounded corners=2pt, line width=2pt] at (10.5*\hstep,1*\vstep) {}; %
\node[draw,minimum height=2.75*-\vstep, minimum width=3.7*\hstep, %
  fill = black!20!,
  rounded corners=2pt] at (14.5*\hstep,1*\vstep) {}; %

\node[draw,minimum height=1.5*-\vstep, minimum width=1.5*\hstep, %
  fill = black!30!,
  rounded corners=1pt] at (1.5*\hstep,.5*\vstep) {}; %
\node[draw,minimum height=1.5*-\vstep, minimum width=1.5*\hstep, %
  fill = black!30!,
  rounded corners=1pt] at (3.5*\hstep,.5*\vstep) {}; %
\node[draw,minimum height=1.5*-\vstep, minimum width=1.5*\hstep, %
  fill = black!30!,
  rounded corners=1pt] at (5.5*\hstep,.5*\vstep) {}; %
\node[draw,minimum height=1.5*-\vstep, minimum width=1.5*\hstep, %
  fill = black!30!,
  rounded corners=1pt] at (7.5*\hstep,.5*\vstep) {}; %
\node[draw,minimum height=1.5*-\vstep, minimum width=1.5*\hstep, %
  fill = black!30!,
  rounded corners=1pt] at (9.5*\hstep,.5*\vstep) {}; %
\node[draw,minimum height=1.5*-\vstep, minimum width=1.5*\hstep, %
  fill = black!30!,
  rounded corners=1pt] at (11.5*\hstep,.5*\vstep) {}; %
\node[draw,minimum height=1.5*-\vstep, minimum width=1.5*\hstep, %
  fill = black!30!,
  rounded corners=1pt] at (13.5*\hstep,.5*\vstep) {}; %
\node[draw,minimum height=1.5*-\vstep, minimum width=1.5*\hstep, %
  fill = black!30!,
  rounded corners=1pt] at (15.5*\hstep,.5*\vstep) {}; %

\foreach \i in {0,...,4}
{
	\foreach \j in {1,...,16}
	{
  		\node[vertex] (\i\j) at (\j*\hstep,\i*\vstep) {};
	}
}	

\foreach \i in {0}
{
	\foreach \j in {1,...,16}
	{
  		\node[circle] (\i\j) at (\j*\hstep,-1*\vstep) {$\pgfmathprintnumber{\j}$};
	}
}	

\foreach \i in {0,...,4}
{
	\foreach \j in {0}
	{
  		\node[circle] (\i\j) at (\j*\hstep,\i*\vstep) {$\pgfmathprintnumber{\i}$};
	}
}	

\foreach \i in {0,...,3}
{
	\foreach \j in {1,...,16}
	{	
		\foreach \k in {0,...,1}
		{
			\draw[basicedge]
			let 
				\n1 = {int(mod((\j-1),\r^(\i))+\k*\r^(\i)+\r^(\i+1)*int(floor((\j-1)/\r^(\i+1))))+1}, 
				\n2 = {int(\i+1)} 
			in 
				(\j*\hstep,\i*\vstep) -- (\n1*\hstep,\n2*\vstep);
		}
	}
}	


\end{tikzpicture}
\end{center}
\caption{All the subsets $\vertexset{\bba}{}(k,w)$, where $\bba = \radixMatrices{2}{4}$, $k\in[4]$ and $w \in [r^{4-k}]$ depicted by rectangles, and the set $\vertexset{\bba}{}(2,3)$ is highlighted by the rectangle with thick border.}
\label{fig:subgraphs}
\end{figure}
Before stating the next result, it is worth recalling the graph theoretical interpretation of the conditional independence structure of the augmented resampling algorithm defined in \eqnref{eq:def vertex set} and \eqnref{eq:def edge set} in \secref{sec:conditional variance and collision analysis main}. The following result establishes the conditional independence of specific subsets of vertices of the graph $\graph{\bba}{}$, where $\bba = \radixMatrices{r}{m}$. For all $r\geq 2$, $m\geq 1$, $k \in [m]$ and $w \in [r^{m-k}]$ these subsets are defined as
\begin{equation}
\vertexset{\bba}{}(k,w)
\defeq \big\{\xi^i_q: 0 \leq q \leq k,~(w-1)r^{k}< i \leq wr^{k}\big\}, \label{eq:sub vertices}
\end{equation} 
where $\bba = \radixMatrices{r}{m}$.  See \figref{fig:subgraphs} for an illustrations of these sets.
\begin{lemma}\label{lem:conditional independence in graph}
Fix $m\geq 1$, $r\geq 2$, $\bba = \radixMatrices{r}{m}$, $k\in [m]$, $w_{1},w_{2} \in [r^{m-k}]$, such that $w_{1}\neq w_{2}$ and $u_1,u_2\in[r^m]$ and $0\leq q_1,q_2\leq m$ such that $(\xi^{u_1}_{q_1},\xi^{u_2}_{q_2}) \in \vertexset{\bba}{}(k,w_{1})\times\vertexset{\bba}{}(k,w_{2})$. Then $\xi^{u_1}_{q_1} \independent \xi^{u_2}_{q_2} \,|\,\xi_0$.
\end{lemma}
\begin{proof}
For notational purposes, let us assume an ordering of the elements of $\vertexset{\bba}{}$ such that $\xi^{i_1}_{k_1} \leq \xi^{i_2}_{k_2}$ if and only if $k_1 \leq k_2$ or $k_1 = k_2$ and $i_1\leq i_2$. Then for any subset $V \subset \vertexset{\bba}{2}$ of size $p\in [(m+1)r^m]$ we define $x_V \defeq (x^{i_1}_{k_1},\ldots,x^{i_p}_{k_p})$ where $(i_\ell)_{\ell=1}^p$ and $(k_\ell)_{\ell=1}^p$ are such that $V=\{\xi^{i_\ell}_{k_\ell}:1 \leq \ell \leq p\}$ and $\xi^{i_\ell}_{k_\ell} \leq \xi^{i_{\ell+1}}_{k_{\ell+1}}$. We also use $\pa{V} \defeq \{\xi^j_q : q=k-1,~~j\in\parentSet{k}{i}{\bba}{},~\xi^{i}_k\in V\}$ to denote the graph theoretical parents of the elements of $V$. 
Moreover we define, for $i\in\{1,2\}$, $\vertexset{\bba}{}^{i} \defeq \{\xi^j_p: A_{p+1:q_i}^{u_ij}\neq 0
,~0\leq p<q_i\}\cup\{\xi^{u_i}_{q_i}\}$.
In simple terms, $\vertexset{\bba}{}^{1}$ and $\vertexset{\bba}{}^{2}$ are the ancestor sets of $\xi^{u_{1}}_{q_{1}}$ and $\xi^{u_{2}}_{q_{2}}$, respectively.

First we show that $\vertexset{\bba}{}^{1} \cap \vertexset{\bba}{}^{2}=\emptyset$. To do this, assume that there exists $\xi^{u^\ast}_{q^\ast} \in \vertexset{\bba}{}^{1} \cap \vertexset{\bba}{}^{2}$, where $q^\ast \leq \min(q_1,q_2)$ and $u^\ast \in [r^m]$. By \lemmaref{lem:general properties of path sets}\ref{it:positive diagonals} there exists $(i_0,\ldots,i_{q_1}) \in \paths{\bba_{1:q_1}}{}$ and $(j_0,\ldots,j_{q_2}) \in \paths{\bba_{1:q_2}}{}$ such that $i_0 = j_0 = u^\ast$, $i_{q_1} = u_1$ and $j_{q_2} = u_2$. Hence $u^\ast \in \primeParentSet{q_1}{u_1}{\bba}{} \cap \primeParentSet{q_2}{u_2}{\bba}{}$. On the other hand, by \lemmaref{lem:nonzero elements}, $\primeParentSet{k}{u_1}{\bba}{} = \{(w_{1}-1)r^{k}+1,\ldots,w_{1}r^{k}\}$ and $\primeParentSet{k}{u_2}{\bba}{} = \{(w_{2}-1)r^{k}+1,\ldots,w_{2}r^{k}\}$ and since $q_{1},q_{2}\leq k$, by \lemmaref{lem:general properties of path sets}\ref{it:subset property}, $\primeParentSet{q_1}{u_1}{\bba}{}\subset\primeParentSet{k}{u_1}{\bba}{}$ and $\primeParentSet{q_2}{u_2}{\bba}{}\subset\primeParentSet{k}{u_2}{\bba}{}$. Since $w_{1} \neq w_{2}$, $\primeParentSet{k}{u_1}{\bba}{} \cap \primeParentSet{k}{u_2}{\bba}{} = \emptyset$ and thus $\primeParentSet{q_1}{u_1}{\bba}{} \cap \primeParentSet{q_2}{u_2}{\bba}{} = \emptyset$, which is a contradiction proving that $\vertexset{\bba}{}^{1} \cap \vertexset{\bba}{}^{2}=\emptyset$.

The conditional distribution of $(\xi_k)_{0\leq k \leq m}$ given $\xi_\mathrm{in}$ factorizes according to the graph $\graph{\bba}{}$ with conditional densities 
\begin{equation*}
\phi_{\xi^i_k}(x_{\{\xi^i_k\}\cup \pa{\xi^i_k}}) = 
\begin{cases}
\dfrac{1}{V^i_k}\sum_{j\in\parentSet{k}{i}{\bba}{}}A_k^{ij}V^j_{k-1}\delta_{x^{j}_{k-1}}(x^i_k),& (k,i)\in[m]\times[N]\\
\delta_{\xi_{0}^i}(x^{i}_{0}), & k=0,~i\in[N].
\end{cases}
\end{equation*}
By definition, the sets $\vertexset{\bba}{}^{1}$ and $\vertexset{\bba}{}^{2}$ are ancestral (see, e.g.~\cite{lauritzen96}) and having established that $\vertexset{\bba}{}^{1} \cap \vertexset{\bba}{}^{2}=\emptyset$, we can apply \cite[Corollary 3.23]{lauritzen96} to yield the claimed conditional independence.
\end{proof}

\begin{lemma}\label{lem:independent pair count fixed}
Fix $m > 1$, $r\geq 2$ and $\bba = \radixMatrices{r}{m}$. Then for all $1<k\leq m$, $w \in [r^{m-k}]$
\begin{equation*}
\big|\ipair{\bba}(k,w)\big| \geq r(r-1)\sum_{i=0}^{k-2}(i+1)^2r^{k+i},
\end{equation*}
where for all $r\geq 2$, $m\geq 1$, $k \in [m]$ and $w \in [r^{m-k}]$ 
$$\ipair{\bba}(k,w) \defeq \left\{(\xi^{i_1}_{k_1},\xi^{i_2}_{k_2}) \in \vertexset{\bba}{}(k,w)^2 : k_1,k_2 \in [m],~\xi^{i_1}_{k_1}\independent \xi^{i_2}_{k_2}\,\big|\, \calF^{(r^m,m)}_0\right\}.$$
\end{lemma}
\begin{proof}
The proof is by induction over $k>1$. First we observe that for any $1<k\leq m$ and $w \in [r^{m-k}]$, there are $r$ subsets $\vertexset{\bba}{}(k-1,w_{k-1})\subset \vertexset{\bba}{}(k,w)$, where $wr - (r-1)\leq w_{k-1}\leq wr$.

By \lemmaref{lem:conditional independence in graph}, if $(\xi,\xi') \in \vertexset{\bba}{}(1,w_{1})\times \vertexset{\bba}{}(1,w'_{1})$ where $wr - (r-1)\leq w_{1},w'_{1}\leq wr$ and $w_1 \neq w'_1$, then $(\xi,\xi') \in \ipair{\bba}(2,w)$. By \eqnref{eq:sub vertices} one can check that $\card{\vertexset{\bba}{}(k,w)\setminus\{\xi^i_0:i\in[r^m]\}} = kr^k$ and hence $$\card{\vertexset{\bba}{}(1,w_{1})\setminus\{\xi^i_0:i\in[r^m]\}} =\card{\vertexset{\bba}{}(1,w'_{1})\setminus\{\xi^i_0:i\in[r^m]\}} = r.$$ Therefore the first element of the pair $(\xi,\xi')$ can be chosen among the $r$ elements of $r$ sets and the second element from the $r$ elements of the remaining $r-1$ sets implying that, when $k=2$, for all $w \in [r^{m-2}]$
\begin{equation*}
\card{\ipair{\bba}(2,w)} \geq r(r-1)r^{2} + \sum_{i=wr-(r-1)}^{wr}\card{\ipair{\bba}(1,i)} \geq r(r-1)r^{2},
\end{equation*}
where the second inequality follows from the simplifying observation that for all $w \in [r^{m-2}]$ and $wr - (r-1)\leq i \leq wr$, one trivially has $\card{\ipair{\bba}(1,i)}\geq 0$. This completes the proof for $k=2$.

Let us then assume that the claim holds for some $2\leq k < m$. Therefore each of the $r$ subsets $\vertexset{\bba}{}(k,w_{k})$ of $\vertexset{\bba}{}(k+1,w)$, where $w\in[r^{m-(k+1)}]$ and $wr - (r-1)\leq w_{k}\leq wr$, admits at least
\begin{equation}\label{eq:old pairs}
a_1 = r(r-1)\sum_{i=0}^{k-2}(i+1)^2r^{k+i},
\end{equation}
pairs of vertices that are conditionally independent given $\calF^{(r^m,m)}_0$. By applying \lemmaref{lem:conditional independence in graph} again, similarly as above, and by observing that there are
\begin{equation*}
a_2 = r(kr^k)(r-1)(kr^k)
\end{equation*}
pairs of vertices $(\xi,\xi')\in \vertexset{\bba}{}(k,w_{k})\times \vertexset{\bba}{}(k,w'_{k})$ where $wr-(r-1)\leq w_{k},w'_{k}\leq wr$ and $w_{k} \neq w'_{k}$. From this together with \eqnref{eq:old pairs} we conclude that
\begin{align*}
\big|\ipair{\bba}(k+1,w)\big| &\geq r(kr^k)(r-1)(kr^k) + r \cdot r(r-1)\sum_{i=0}^{k-2}(i+1)^2r^{k+i} \\
&= r(r-1)\sum_{i=0}^{(k+1)-2}(i+1)^2r^{k+1+i},
\end{align*}
completing the proof.
\end{proof}

\section{Proofs for \secref{sec:lln and clt mixed radix}}

In this section we undertake the task of establishing the condition \eqnref{eq:douc_cond_var} of \theref{thm:Douc and Moulines} for the mixed radix-$r$ algorithm. Because the proof of 
\theref{thm:butterfly_clt_mixed_radix} is similar to that of \theref{thm:butterfly_clt_fixed_radix}, also the structure of this section is analogous to \mbox{\secref{sec:proofs for fixed radix supp}}.


\subsection{Conditional independence structure of the mixed radix-\texorpdfstring{$r$}{r} algorithm}
\label{sec:conditional independece structure mixed radix}

\begin{proposition}\label{prop:mixed radix satisfies assumptions}
The matrices $\mradixMatrices{r}{c}$ satisfy \assref{ass:A_k extra} for all $r \geq 2$ and $c\geq 1$. Moreover, define for all $r \geq 2$, $c\geq 1$ and $\ka\in\{1,2\}$
\begin{equation}
\arraycolsep=1.4pt
\begin{array}{rcl}
\mradixSamplePartition{r}{c}{\ka} &\defeq& \big\{\radixSamplePartitionElement{rc}{c}{\ka}{u}:u \in [cr^{2-\ka}]\big\},\\[0.2cm]
\radixSamplePartitionElement{rc}{c}{\ka}{u} &\defeq& \big\{u+(q-1)cr^{2-\ka}:q\in [r^{\ka-1}]\big\},\quad u \in [cr^{2-\ka}].\label{eq:def mradix partitions}
\end{array}
\end{equation}
Then the triple $(\mradixMatrices{r}{c},\mradixSamplePartition{r}{c}{\ka},d)$ satisfies \assref{ass:partition and cond indep} for all $r\geq 2$, $c\geq 1$ and $\ka\in\{1,2\}$.
\end{proposition}

Before the proof of \propref{prop:mixed radix satisfies assumptions}, we state the following technical result establishing explicit expression for the sets needed in the collision analysis in the case of the mixed radix-$r$ algorithm.
\begin{lemma}\label{lem:nonzero elements mixed}
Fix $r\geq 2$, $c\geq 1$ and $\bba = \mradixMatrices{r}{c}$. For all $i \in [rc]$
\begin{eqnarray}
\parentSet{1}{i}{\bba}{4} &=& \left\{c\floor{\frac{i-1}{c}}+q: q \in [c]\right\}, \label{eq:explicit parents mixed radix 1}\\
\parentSet{2}{i}{\bba}{} &=& \bigg\{\big((i-1)\bmod c\big) + (q-1)c + 1 : q\in [r]\bigg\},\label{eq:explicit parents mixed radix}
\end{eqnarray}
and $\primeParentSet{1}{i}{\bba}{4} = \parentSet{1}{i}{\bba}{4}$, $\primeParentSet{2}{i}{\bba}{4} = [rc]$.
\end{lemma}
\begin{proof}
By the element-wise definition \eqnref{eq:kronecker inverse definition} of the Kronecker product and \eqnref{eq:mradix matrix definition} it follows similarly as in the proof of \lemmaref{lem:nonzero elements} that
$\parentSet{1}{\alpha+1}{\bba}{4} = \{j \in [rc]:\floor{{\alpha}/c} = \floor{(j-1)/c}\}$. From this \eqnref{eq:explicit parents mixed radix 1} follows by elementary calculation. Equation \eqnref{eq:explicit parents mixed radix} can be verified exactly as in the proof of \eqnref{eq:tailProdSet fixed} in \lemmaref{lem:nonzero elements}. The identity $\primeParentSet{1}{i}{\bba}{4} = \parentSet{1}{i}{\bba}{4}$ follows immediately by definition and finally the claim $\primeParentSet{2}{i}{\bba}{4} = [rc]$ holds because $A_1A_2 = \ones_{1/rc}$.
\end{proof}

\begin{proof}[Proof of \propref{prop:mixed radix satisfies assumptions}]
To prove that $\bba = \mradixMatrices{r}{c}$ satisfies \assref{ass:A_k extra} we observe, as in the proof of \propref{prop:fixed radix satisfies assumptions}, that the only non-trivial property  is \assref{ass:A_k extra}\ref{it:unique paths}, which follows similarly as in the proof of \propref{prop:fixed radix satisfies assumptions} by using \lemmaref{lem:nonzero elements mixed}.

\assref{ass:partition and cond indep}\ref{it:is partition} can be checked with elementary calculation using \eqnref{eq:mradix matrix definition} and \eqnref{eq:def mradix partitions}. \assref{ass:partition and cond indep}\ref{it:equivalence classes} is verified simply by noting that for $\ka = 1$, $u\in[rc]$, $\big|\radixSamplePartitionElement{rc}{c}{\ka}{u}\big| = 1$, and for $\ka=2$ one has $\prod_{k=0}^{\ka-1}A_{m-k} = \ones_{1/rc}$. 

To check \assref{ass:partition and cond indep}\ref{it:cond indep} we first note that for $\ka=1$ the claim follows from the one step conditional independence. For $\ka=2$, one can check using \eqnref{eq:explicit parents mixed radix} of \lemmaref{lem:nonzero elements mixed} analogously to the proof of \lemmaref{lem:explicit analysis regarding partitions}\ref{it:partition property 1 fixed}, that for all $(i,j) \in \radixSamplePartitionElement{rc}{c}{2}{u_1}\times \radixSamplePartitionElement{rc}{c}{2}{u_2}$ where $(u_1,u_2)\in [c]^2$ such that $u_1 \neq u_2$, one has $\parentSet{2}{i}{\bba}{} \cap \parentSet{2}{j}{\bba}{} = \emptyset$. By the one step conditional independence, we have $\parentSet{2}{i}{\bba}{} \times \parentSet{2}{j}{\bba}{} \subset \ciSet{1}{\xi_0}{rc}$ and hence \assref{ass:partition and cond indep}\ref{it:cond indep} follows from \lemmaref{lem:cond indep propagation}.
\end{proof}

\subsection{Convergence of the conditional variance}
\label{sec:convergence of the conditional variance mixed}

The main result of this section is the following proposition. 
\begin{proposition}\label{prop:var_conv_mixed_radix}
Under the hypotheses of \theref{thm:butterfly_clt_mixed_radix}
\begin{equation*}
\E\Bigg[\bigg(\sum_{\varrho =1}^{2rc} \mrbfX{\varrho}{rc}{2}{\varphi}\bigg)^2\Bigg|\mrbfF{0}{rc}{2}\Bigg] \inprob{c\to\infty} \left(1-\frac{1}{2r}\right)\mu\Bigg(g\bigg(\varphi-\frac{\mu(g\varphi)}{\mu(\varphi)}\bigg)^2\Bigg)\mu(g).
\end{equation*}
\end{proposition}


We have the following result, which serves a purpose analogous to \lemmaref{eq:convergence of weigted subsample integrals} in the case of radix-$r$ algorithm, although it is somewhat different by nature.
\begin{lemma}
\label{lem:subsample convergence for mixed radix}
Under the hypotheses of \theref{thm:butterfly_clt_mixed_radix}, if for all $i\in [rc]$, $q_i \in [r^{k-1}]$, for some $k \in \{1,2\}$, then for all $\varphi,\varphi'\in\boundMeas{\ss}$ 
\begin{equation*}
\frac{1}{rc}\sum_{i}\frac{r^{k-1}}{rc}
\sum_{j\in [cr^{2-k}]}\varphi(\xi^i)\varphi'(\xi^{J(j)})
- \mu(\varphi)\mu(\varphi')\inprob{c\to\infty} 0.
\end{equation*}
where $J(j) = j + (q_{i}-1)cr^{2-k}$ for all $j \in [cr^{2-k}]$.
\end{lemma}
\begin{proof}
By defining
\begin{align*}
A_c &\defeq \frac{1}{rc}\sum_i \varphi(\xi^i) \frac{r^{k-1}}{rc}\sum_{j\in [cr^{2-k}]}(\varphi'(\xi^{J(j)}) - \mu(\varphi')),\\
B_c &\defeq \frac{1}{rc}\sum_i \varphi(\xi^i)\mu(\varphi') - \mu(\varphi)\mu(\varphi'),
\end{align*}
we have by the triangle inequality
\begin{equation*}
\abs{\frac{1}{rc}\sum_{i}\frac{r^{k-1}}{rc}\sum_{j\in [cr^{2-k}]}\varphi(\xi^i)\varphi'(\xi^{J(j)}) - \mu(\varphi)\mu(\varphi')} = \abs{A_c + B_c} \leq \abs{A_c} + \abs{B_c}.
\end{equation*}
By the hypotheses of \theref{thm:butterfly_clt_mixed_radix}, by setting $\ka=1$ and $q=1$ in \eqnref{eq:required conv in prob mixed} yields for all $c\geq 1$
\begin{equation}\label{eq:lp convregence mixed}
\E\left[\abs{\frac{1}{rc}\sum_i \varphi(\xi^i) - \mu(\varphi)}^{2}\right]^{\frac{1}{2}} \leq b(\varphi)\sqrt{\frac{2}{rc}},
\end{equation}
from which we deduce that $\abs{B_c}$ converges to zero in probability as $c\to\infty$. It remains to show the same for $\abs{A_c}$. By Jensen's inequality, Cauchy-Schwartz inequality and \eqnref{eq:required conv in prob mixed} we have
\begin{eqnarray*}
\E\left[\abs{A_c}\right] 
&\leq&
\sqrt{\E\left[\bigg(\frac{1}{rc}\sum_i \varphi(\xi^i) \frac{r^{k-1}}{rc}\sum_{j\in [cr^{2-k}]}(\varphi'(\xi^{J(j)}) - \mu(\varphi'))\bigg)^2\right]}\\
&\leq&
\norm{\varphi}_\infty\sqrt{\frac{1}{rc}\sum_i\E\left[\Bigg(\frac{r^{k-1}}{rc}\sum_{j\in [cr^{2-k}]}(\varphi'(\xi^{J(j)}) - \mu(\varphi'))\Bigg)^2\right]}\\
&\leq&
\norm{\varphi}_\infty b(\varphi)\sqrt{\frac{2-k}{rc} + \frac{r^{k-1}}{rc}}\inprob{c\to\infty} 0,
\end{eqnarray*}
completing the proof.
\end{proof}

\begin{proof}[Proof of \propref{prop:var_conv_mixed_radix}]
Because $r\geq 2$ is assumed fixed, let us write $\bba(c)=\mradixMatrices{r}{c}$. By \propref{prop:mixed radix satisfies assumptions} we can apply \lemmaref{lem:bijectivity g} and on the other hand we can also use \lemmaref{lem:nonzero elements mixed} yielding for all $i,u_0\in[rc]$
\begin{equation}\label{eq:essential cardinalities mixed}
\big|\primeParentSet{k}{i}{\bba(c)}{}\big| = cr^{k-1},\quad \big|\parentSet{k}{i}{\bba(c)}{}\big| = c^{2-k}r^{k-1},\quad \big|\lowerPartFrom{k}{u_0}{\bba(c)}\big| = r^{2-k},
\end{equation}
and therefore by substitution in \propref{prop:tensor product formulation} 
\begin{align}\label{eq:mixed radix conditional second moment}
&\E\left[\left(\displaystyle\sum_{\varrho=1}^{2rc} X^{(rc,2)}_\varrho\right)^2 \mids \calF^{(rc,2)}_0\right]
= \dfrac{rc}{2}\dfrac{1}{(rc)^2}\displaystyle\sum_{i} g^{2}(\bxi{i}{0}{}{})\cvarphi^2_{rc}(\bxi{i_0}{0}{}{}) \\
&\quad+~ \dfrac{rc}{2}\dfrac{1}{(rc)^3}\displaystyle\sum_{k=1}^m\sum_{i}\sum_{j\neq i}g(\bxi{i}{0}{}{})\cvarphi^2_{rc}(\bxi{i}{0}{}{})g(\bxi{j}{0}{}{})\ind{}\Big(j\in \collisionStartSetSh{k}{i}{\bba(c)}\Big)r^{2-k}\nonumber\\
&\quad+~ \dfrac{rc}{2}\dfrac{1}{(rc)^3}\displaystyle\sum_{k=1}^m\sum_{i}\sum_{j\neq i}g(\bxi{i}{0}{}{})\cvarphi_{rc}(\bxi{i}{0}{}{})g(\bxi{j}{0}{}{})\cvarphi_{rc}(\bxi{j}{0}{}{})\ind{}\Big(j\in \collisionStartSetSh{k}{i}{\bba(c)}\Big)(rc-r^{2-k}).\nonumber
\end{align}
%
To obtain the limit of \eqnref{eq:mixed radix conditional second moment} we first observe that by \eqnref{eq:lp convregence mixed}, \eqnref{eq:varphi_bar_defn}, and the continuous mapping theorem
\begin{equation}\label{eq:mixed part 1}
\dfrac{rc}{2}\frac{1}{(rc)^2}\sum_i g^2(\xi^i_0)\cvarphi_{rc}^2(\xi^i_0) \inprob{c\to\infty} \frac{1}{2}\mu(g^2\cvarphi^2).
\end{equation}
%
For the second sum in \eqnref{eq:mixed radix conditional second moment} we define
\begin{align*}
A_c &\defeq \frac{1}{rc}\sum_{i}g(\bxi{i}{0}{}{})\cvarphi^2_{rc}(\bxi{i}{0}{}{})\frac{r}{rc}\sum_{j}g(\bxi{j}{0}{}{})\ind{}\Big(j\in\collisionStartSetSh{1}{i}{\bba(c)}\Big),\\
B_c &\defeq
\frac{1}{rc}\sum_{i}g(\bxi{i}{0}{}{})\cvarphi^2_{rc}(\bxi{i}{0}{}{})\frac{1}{rc}\sum_{j}g(\bxi{j}{0}{}{})\ind{}\Big(j\in\collisionStartSetSh{2}{i}{\bba(c)}\Big),
\end{align*}
in which case the second sum is equal to $(A_{c}+B_{c})/2$. By \lemmaref{lem:nonzero elements mixed} 
 we have 
\begin{equation*}
\collisionStartSetSh{1}{i}{\bba(c)} = \{j + (q_{i}-1)c: j \in [c]\}\setminus \{i\}
\end{equation*}
where $q_{i} \defeq \floor{(i-1)/c}+1$. Since $q_{i} \in [r]$ we can use \lemmaref{lem:subsample convergence for mixed radix} with $k=2$, the continuous mapping theorem, and the fact that
\begin{equation*}
\abs{\frac{1}{rc}\sum_{i}g^2(\bxi{i}{0}{}{})\cvarphi^2_{rc}(\bxi{i}{0}{}{})} \leq \norm{g}_\infty^2 \osc{\varphi}^2,
\end{equation*} 
and we have
\begin{align}
A_c 
&= \frac{1}{rc}\sum_{i}g(\bxi{i}{0}{}{})\cvarphi^2_{rc}(\bxi{i}{0}{}{})\left(\frac{r}{rc}\sum_{j\in [c]}g(\bxi{J(j)}{0}{}{}) - \frac{r}{rc}g(\bxi{i}{0}{}{})\right)\nonumber\\
&=
\frac{1}{rc}\sum_{i}g(\bxi{i}{0}{}{})\cvarphi^2_{rc}(\bxi{i}{0}{}{})\frac{r}{rc}\sum_{j\in[c]}g(\bxi{J(j)}{0}{}{}) - \frac{r}{rc}\frac{1}{rc}\sum_{i}g^2(\bxi{i}{0}{}{})\cvarphi^2_{rc}(\bxi{i}{0}{}{}) \nonumber\\
&\inprob{c\to\infty} \mu(g\cvarphi^2)\mu(g).\label{eq:part 2a}
\end{align}
where $J(j) = j + (q_{i}-1)c$ for all $j \in [c]$. Using \lemmaref{lem:nonzero elements mixed}, it can be checked that
\begin{equation*}
\collisionStartSetSh{2}{i}{\bba(c)} 
= \{j + (q-1)c \in [rc]:q \in [r]\setminus \{q_i\},~j \in [c]\}
\end{equation*}
Then, by \lemmaref{lem:subsample convergence for mixed radix} with $k=2$, \eqnref{eq:lp convregence mixed}, and the continuous mapping theorem
\begin{align}
B_c 
&= \frac{1}{rc}\sum_{i}g(\bxi{i}{0}{}{})\cvarphi^2_{rc}(\bxi{i}{0}{}{})\frac{1}{rc}\sum_{\substack{q \in [r]\\q\neq q_i}}\sum_{j \in [c]}g(\bxi{J_q(j)}{0}{}{})\nonumber\\
&= \Bigg(\frac{1}{rc}\sum_{i}g(\bxi{i}{0}{}{})\cvarphi^2_{rc}(\bxi{i}{0}{}{})\Bigg)\Bigg(\frac{1}{rc}\sum_{j}g(\bxi{j}{0}{}{})\Bigg) \nonumber\\
&\quad-~ \frac{1}{r}\frac{1}{rc}\sum_{i}g(\bxi{i}{0}{}{})\cvarphi^2_{rc}(\bxi{i}{0}{}{})\frac{1}{c}\sum_{j \in [c]}g(\bxi{J_{q_{i}}(j)}{0}{}{})\nonumber\\
& \inprob{c\to\infty}\left(1-\frac{1}{r}\right)\mu(g\cvarphi^2)\mu(g),\label{eq:part 2b}
\end{align}
where $J_q(j) = j + (q-1)c$ for all $j \in  [c]$ and $q \in [r]$. By combining \eqnref{eq:part 2a} and \eqnref{eq:part 2b} we have 
\begin{align}\label{eq:mixed radix part 2}
& \dfrac{rc}{2}\dfrac{1}{(rc)^3}\displaystyle\sum_{k=1}^m\sum_{i}\sum_{j\neq i}g(\bxi{i}{0}{}{})\cvarphi^2_{rc}(\bxi{i}{0}{}{})g(\bxi{j}{0}{}{})\ind{}\Big(j\in \collisionStartSetSh{k}{i}{\bba(c)}\Big)r^{2-k} \nonumber \\
& =\frac{1}{2}\left(A_c + B_c\right) \inprob{c\to\infty} \left(1-\frac{1}{2r}\right)\mu(g\cvarphi^2)\mu(g).
\end{align}

%
To conclude the proof we define
\begin{align*}
A'_c &\defeq \frac{1}{rc}\sum_{k=1}^m\sum_{i}\sum_{j}g(\bxi{i}{0}{}{})\cvarphi_{rc}(\bxi{i}{0}{}{})g(\bxi{j}{0}{}{})\cvarphi_{rc}(\bxi{j}{0}{}{})\ind{}\Big(j\in \collisionStartSetSh{k}{i}{\bba(c)}\Big),\\
B'_c &\defeq
\frac{1}{(rc)^2}\sum_{k=1}^m\sum_{i}\sum_{j}g(\bxi{i}{0}{}{})\cvarphi_{rc}(\bxi{i}{0}{}{})g(\bxi{j}{0}{}{})\cvarphi_{rc}(\bxi{j}{0}{}{})\ind{}\Big(j\in \collisionStartSetSh{k}{i}{\bba(c)}\Big)r^{2-k},
\end{align*}
in which case the third sum in \eqnref{eq:mixed radix conditional second moment} equals $(A'_{c}+B'_{c})/2$. 
By \lemmaref{lem:general properties of path sets}\ref{it:complement decomposition}, and the fact that ${(rc)^{-1}}\sum_{i}g(\bxi{i}{0}{}{})\cvarphi_{rc}(\bxi{i}{0}{}{}) = 0$,
for $A'_c$ we have
\begin{align}
A'_c 
&= \frac{1}{rc}\sum_{i}\sum_{j\neq i}g(\bxi{i}{0}{}{})\cvarphi_{rc}(\bxi{i}{0}{}{})g(\bxi{j}{0}{}{})\cvarphi_{rc}(\bxi{j}{0}{}{})\nonumber\\
&= rc\left(\frac{1}{rc}\sum_{i}g(\bxi{i}{0}{}{})\cvarphi_{rc}(\bxi{i}{0}{}{})\right)^2 - \frac{1}{rc}\sum_{i}g^2(\bxi{i}{0}{}{})\cvarphi_{rc}^2(\bxi{i}{0}{}{}) \nonumber \\
&\inprob{c\to\infty} -\mu(g^2\cvarphi^2),\label{eq:interm lim 2}
\end{align}
where the second equality follows similarly as in \eqnref{eq:tmp final form} and the convergence follows from \eqnref{eq:lp convregence mixed} together with the continuous mapping theorem. By arguments identical to those used in proving \eqnref{eq:mixed radix part 2} we see that $B'_n$ converges in probability to $\left(2 - {r^{-1}}\right)\mu(g\cvarphi)^2 = 0$ and combining this with \eqnref{eq:interm lim 2} gives
\begin{align}
&\dfrac{rc}{2}\dfrac{1}{(rc)^3}\displaystyle\sum_{k=1}^m\sum_{i}\sum_{j\neq i}g(\bxi{i}{0}{}{})\cvarphi_{rc}(\bxi{i}{0}{}{})g(\bxi{j}{0}{}{})\cvarphi_{rc}(\bxi{j}{0}{}{})\ind{}\Big(j\in \collisionStartSetSh{k}{i}{\bba(c)}\Big)(rc-r^{2-k}) \nonumber \\
& \qquad = \frac{1}{2}\left(A'_c - B'_c\right) \inprob{c\to \infty} -\frac{1}{2}\mu(g^2\cvarphi^2).\nonumber
\end{align}
The proof is completed by combining this limit and the limits in \eqnref{eq:mixed part 1} and \eqnref{eq:mixed radix part 2} with \eqnref{eq:mixed radix conditional second moment}.
\end{proof}

\subsection{Approximation of the conditional variance and independence analysis}
\label{sec:independence analysis mixed}

The main result of this section is the following proposition, which is the last remaining part in completing the proof of \theref{thm:butterfly_clt_mixed_radix}.
\begin{proposition}\label{prop:var_vanish_mixed_radix}
Under the hypotheses of \theref{thm:butterfly_clt_mixed_radix},
\begin{equation*}
\sum_{\varrho\in[2rc]}\left(\E \left[ \left.  \left(\mrbfX{\varrho}{rc}{2}{\varphi} \right)^2\right| \mrbfF{\varrho-1}{rc}{2} \right]- \E \left[ \left.\left(\mrbfX{\varrho}{rc}{2}{\varphi}\right)^2\right|\mrbfF{0}{rc}{2}\right]\right)\inprob{c\to\infty} 0.
\end{equation*}
\end{proposition}
\begin{proof}[Proof of \propref{prop:var_vanish_mixed_radix}]
Recall the definition of $Z^{(rc,2)}_\varrho$ in \eqnref{eq:def Z}. By Markov's inequality we have the same decomposition \eqnref{eq:fixed_radix_P_of_Z} as in the case of the radix-$r$ algorithm. By \eqnref{eq:boundedness of X}, $\big|\mrbfX{\varrho}{rc}{2}{\varphi}\big|\leq (2rc)^{-1/2}\infnorm{g}\osc{\varphi}$, hence
for any $\varrho,\varrho'\in[2rc]$,
\begin{equation}
\left|Z^{(rc,2)}_\varrho Z^{(rc,2)}_{\varrho'}\right| \leq \frac{1}{(rc)^2}\infnorm{g}^4 \osc{\varphi}^4,\label{eq:Z_times_Z_bounded mixed}
\end{equation}
and the first term on the r.h.s.~of \eqref{eq:fixed_radix_P_of_Z} converges to zero as $c\to\infty$. It remains to establish the convergence of the second term in a manner similar to that in the proof of \propref{prop:var_vanish_fixed_radix}.

There are altogether $2rc(2rc - 1)$ pairs $\big(Z^{(rc,2)}_{\varrho},Z^{(rc,2)}_{\varrho'}\big)$ with $\varrho\neq \varrho'$, and thus by \lemmaref{lem:indep pair count for mixed radix}, there are at most
\begin{equation*}
a_c = 2rc(2rc - 1) - rc(rc-1) - 3rc(rc-r) = 3r^2c - rc
\end{equation*}
pairs which are not conditionally independent given $\mrbfF{0}{rc}{2}$. Therefore it is enough to apply \eqnref{eq:Z_times_Z_bounded mixed}, and check that
$\lim_{c\rightarrow 0} a_c/(rc)^2= 0$, which is trivial.
\end{proof}

\begin{lemma}\label{lem:indep pair count for mixed radix}
Fix $r\geq 2$, $c\geq 1$ and $\bba = \mradixMatrices{r}{c}$. Then 
\begin{equation*}
\card{\ipair{\bba}} \geq rc(rc-1) + 3rc(rc-r),
\end{equation*}
where $\ipair{\bba} \defeq \left\{(\xi^{i_1}_{k_1},\xi^{i_2}_{k_2})  : k_1,k_2 \in \{1,2\},~i_1,i_2\in[rc],~\xi^{i_1}_{k_1}\independent \xi^{i_2}_{k_2}\,\big|\, \mrbfF{0}{rc}{2}\right\}$.
\end{lemma}
\begin{proof}
By the one step conditional independence 
\begin{equation*}
A \defeq \left\{(\xi^{i_1}_{k_1},\xi^{i_2}_{k_2})  : k_1=k_2=1,~i_1,i_2\in[rc],~i_1 \neq i_2\right\} \subset \ipair{\bba},
\end{equation*}
and readily $\card{A} = rc(rc-1)$. For the set
\begin{equation*}
B \defeq \left\{(\xi^{i}_{1},\xi^{j}_{2})  : i,j \in [rc],~i \notin \parentSet{2}{j}{\bba}{}\right\},
\end{equation*}
we also have $B \subset \ipair{\bba}$, since by the one step conditional independence, for all $i,j\in[rc]$ such that $i\notin \parentSet{2}{j}{\bba}{}$ and for all $S_{1}, S_{2} \in \sss$ we have by \eqnref{eq:rigorus formulation of xi} 
\begin{align*}
\P\big(\xi^i_1 \in S_1,~ \xi^j_2 \in S_2\big| \xi_0\big) 
&= \sum_{\ell \in \parentSet{2}{j}{\bba}{4}} \P\left(\xi^i_1 \in S_1\mids \xi_0\right)\P\big(I^{j}_2 = \ell,~ \xi^\ell_1 \in S_2\,\big|\, \xi_0\big)\\
&= \P\left(\xi^i_1 \in S_1\mids \xi_0\right)\P\big(\xi^j_2 \in S_2\big| \xi_0\big).
\end{align*}
Also, because by \eqnref{eq:essential cardinalities mixed}, $\big|\parentSet{2}{j}{\bba}{}\big| = r$, one has $\card{B} = rc(rc-r)$. Similarly we have
\begin{equation*}
C \defeq \left\{(\xi^{i}_{2},\xi^{j}_{1}) : i,j \in [rc],~j \notin 
\parentSet{2}{i}{\bba}{}\right\} \subset \ipair{\bba},
\end{equation*}
and $\card{C} = rc(rc-r)$. Moreover, by \lemmaref{lem:cond indep propagation} and the one step conditional independence, we also have
\begin{equation*}
D \defeq \left\{(\xi^{i}_{2},\xi^{j}_{2}) : \parentSet{2}{i}{\bba}{} \cap \parentSet{2}{j}{\bba}{} = \emptyset \right\} \subset \ipair{\bba},
\end{equation*}
and by \propref{prop:mixed radix satisfies assumptions} one can check similarly as in the proof of \lemmaref{lem:general properties of path sets}\ref{it:equal prime parent sets} that $\parentSet{2}{i}{\bba}{} \cap \parentSet{2}{j}{\bba}{} = \emptyset$ if and only if $j \notin \parentSet{2}{i}{\bba}{}$, and hence $\card{D} = rc(rc-r)$. Finally the claim follows by observing that $A \cap B \cap C \cap D = \emptyset$, and hence $\card{\ipair{\bba}} = \card{A}+\card{B}+\card{C}+\card{D}$.
\end{proof}


\section{Proofs for \secref{sec:convergence of particle filters}}
\label{sec:subsample convergence supp}


Our next objective is to prove \propref{prop:subsample absolute moment convergence}. The first step is a generalization of \eqnref{eq:aug_res_L_p_bound}.
\begin{lemma}\label{lem:absolute moment bound}
For all $\varphi\in\boundMeas{\ss}$ and $p>1$ there exists $b_p \in \real$, depending only on $p$, such that if $\big(\bba^{(N,m)},\calI,\ka)$ satisfies \assref{ass:partition and cond indep} for some $N,m\geq 1$ and $\ka\in[m]$, then for all $J \in \partitionMap{\calI}$
\begin{align*}
&\E\Bigg[\Bigg|\Bigg(\frac{1}{N}\sum_{i}g(\xi_{\mathrm{in}}^i)\Bigg)\Bigg(\frac{1}{\card{\calI}}\sum_{i=1}^{\card{\calI}}\varphi(\xi_{\mathrm{out}}^{J(i)})\Bigg)-\frac{1}{N}\sum_{i}g(\xi_{\mathrm{in}}^i)\varphi(\xi_{\mathrm{in}}^i)\Bigg|^p\Bigg]^{\frac{1}{p}} \nonumber \\
&\quad\leq \scale{m}{\ka}b_p\norm{g}_\infty\osc{\varphi}.
\end{align*}
\end{lemma}
\begin{proof}
Follows from \propref{prop:generalized martingale decomposition} similarly as in the proof of \propref{prop:intro_to_aug_resampling}.
\end{proof}
To prove \propref{prop:subsample absolute moment convergence}, we first establish a bound for the mean of order $p$ for the initialization of the filter. We then proceed to establish similar bounds inductively for the subsequent resampling and mutation steps. This strategy is embodied in the following three lemmata.
\begin{lemma}[Initialization]\label{lem:initialization moment bound}
Fix $N\geq 1$. For all $\varphi \in \boundMeas{\ss}$ and $p>1$, there exists $b_0(p) \in \real$, depending only on $p$, such that
\begin{equation*}
\E\left[\abs{\frac{1}{N}\sum_{i} \varphi(\zeta^{i}_0) - \pi_0(\varphi)}^p\right]^{\frac{1}{p}} \leq
b_0(p)\sqrt{\frac{1}{N}}\osc{\varphi}.
\end{equation*}
\end{lemma}
\begin{proof}
Because $\left\{\zeta_0^{i}\right\}_{i\in[N]}\iidsim\pi_0$ the claim follows straightforwardly by Burkholder's inequality.
\end{proof}
\begin{lemma}[Resampling]\label{lem:resampling moment bound}
Let $n\geq 0$ and $p>1$ be fixed. If the triple $(\bba^{(N,m)},\calI,\ka)$, satisfies \assref{ass:partition and cond indep} for some $N,m\geq 1$ and $\ka\in[m]$  and for all
$\varphi \in \boundMeas{\ss}$ there exists $b_n(\varphi,p) \in \real$ such that
\begin{equation}\label{eq:absolute moment induction assumption}
\E\left[\abs{\frac{1}{N}\sum_{i} \varphi(\zeta_n^i) - \pi_n(\varphi)}^p\right]^{\frac{1}{p}} \leq b_n(\varphi,p) \sqrt{\frac{m}{N}},
\end{equation}
then for all $\varphi \in \boundMeas{\ss}$ there exists $\hat{b}_n(\varphi,p)\in\real$
such that for all $J \in \partitionMap{\calI}$
\begin{equation*}
\E\left[\abs{\frac{1}{\card{\calI}}\sum_{i=1}^{\card{\calI}}\varphi(\hat{\zeta}^{J(i)}_n) - \hat{\pi}_n(\varphi)}^p\right]^{\frac{1}{p}} \leq \hat{b}_n(\varphi,p)\scale{m}{\ka}.
\end{equation*}
\end{lemma}
\begin{proof}
For brevity of notations, let us write $g_n^i\defeq g_n(\zeta^i_n)$ and $\varphi^i_n\defeq \varphi(\zeta^i_n)$. Define
\begin{equation*}
\overline\varphi_{N}(x):=\varphi(x)-\frac{\sum_{i}g_n^i\varphi^i_n}{\sum_{i} g_n^i},
\end{equation*}
and
\begin{align*}
A &\defeq \frac{1}{\pi_n(g_n)}\Bigg(\Bigg(\frac{1}{N}\sum_{i}g_n^i\Bigg)\Bigg(\frac{1}{\card{\calI}}\sum_{i=1}^{\card{\calI}}\varphi(\hat{\zeta}^{J(i)}_n)\Bigg) - \frac{1}{N}\sum_{i}g_n^i\varphi^i_n \Bigg),\\
B &\defeq \frac{\sum_{i}g_n^i\varphi^i_n}{\sum_{i}g_n^i} - \frac{\pi_n(g_n\varphi)}{\pi_n(g_n)},\\
C &\defeq \frac{1}{\pi_n(g_n)}\Bigg(\frac{1}{\card{\calI}}\sum_{i=1}^{\card{\calI}}\cvarphi_N(\hat{\zeta}^{J(i)}_n)\Bigg)\Bigg(\pi_n(g_n) - \frac{1}{N}\sum_{i}g_n^i\Bigg),
\end{align*}
for which we have the decomposition
\begin{eqnarray}\label{eq:abc decompo}
\frac{1}{\card{\calI}}\sum_{i=1}^{\card{\calI}}\varphi(\hat{\zeta}^{J(i)}_n) - \frac{\pi_n(g_n\varphi)}{\pi_n(g_n)} &=& A + B + C.
\end{eqnarray}
By \lemmaref{lem:absolute moment bound}
\begin{equation}\label{eq:part Am}
\E\left[\abs{A}^p\right]^{\frac{1}{p}} \leq \scaleb{m}{\ka}\frac{1}{\pi_n(g_n)}b_p\norm{g_n}_\infty\osc{\varphi}.
\end{equation}
For $B$ we then have, similarly as e.g.~in \cite[proof of Lemma 4]{crisan_et_doucet02}, by Minkowski's inequality and \eqnref{eq:absolute moment induction assumption}
\begin{align}
\E\left[\abs{B}^p\right]^{\frac{1}{p}}
&\leq \frac{\norm{\varphi}_\infty}{\pi_n(g_n)}\E\left[\abs{\pi_n(g_n) - \frac{1}{N}\sum_{i}g^i_n}^p\right]^{\frac{1}{p}} \nonumber \\
&\quad +~\frac{1}{\pi_n(g_n)}\E\left[\abs{ \frac{1}{N}\sum_{i}g^i_n\varphi^i_n - \pi_n(g_n\varphi)}^p\right]^{\frac{1}{p}}\nonumber \\
&\leq \frac{1}{\pi_n(g_n)}\left(\norm{\varphi}_\infty b_n(g_n,p) + b_n(g_n\varphi,p)\right)\sqrt{\frac{m-\ka}{N} + \frac{1}{\card{\calI}}},\label{eq:part Bm}
\end{align}
where we have also used the fact that by \assref{ass:partition and cond indep}\ref{it:is partition} $N/\card{\calI}\geq \ka$ and hence ${m}/{N} \leq (m-\ka)/{N} + {1}/{\card{\calI}}$.
For $C$ we have
\begin{eqnarray}
\E\left[\abs{C}^p\right]^{\frac{1}{p}} &\leq& \frac{\osc{\varphi}}{\pi_n(g_n)}\E\left[\abs{\pi_n(g_n) - \frac{1}{N}\sum_{i}g^i_n}^p\right]^{\frac{1}{p}} \nonumber \\
&\leq&
\frac{\osc{\varphi}}{\pi_n(g_n)}b_n(g_n,p)\sqrt{\frac{m-\ka}{N} + \frac{1}{\card{\calI}}}.\label{eq:part Cm}
\end{eqnarray}
Thus by combining \eqnref{eq:abc decompo}, \eqnref{eq:part Am}, \eqnref{eq:part Bm} and \eqnref{eq:part Cm} the claim follows by Minkowski's inequality.
\end{proof}
\begin{lemma}[Mutation]\label{lem:mutation moment bound}
Fix $N,m\geq 1$, $n\geq 1$, $p>1$, $\ka\in[m]$ and $J \in \partitionMapNoPar{\big(\calI\big)}$, where $\calI$ is a partition of $[N]$.
If for all $\varphi \in \boundMeas{\ss}$ there exists $\hat{b}_n(\varphi,p)\in\real$, such that
\begin{equation}
\E\Bigg[\Bigg|\frac{1}{\card{\calI}}\sum_{i=1}^{\card{\calI}}\varphi(\hat{\zeta}^{J(i)}_{n-1}) - \hat{\pi}_{n-1}(\varphi)\Bigg|^p\Bigg]^{\frac{1}{p}} \leq \hat{b}_n(\varphi,p)\sqrt{\frac{m-\ka}{N}+\frac{1}{\card{\calI}}},\label{eq:moment convergence after resample}
\end{equation}
then for all $\varphi \in \boundMeas{\ss}$ there exists $b_n(\varphi,p) \in \real$ such that
\begin{equation*}
\E\Bigg[\Bigg|\frac{1}{\card{\calI}}\sum_{i=1}^{\card{\calI}} \varphi(\zeta^{J(i)}_n) - \pi_n(\varphi)\Bigg|^p\Bigg]^{\frac{1}{p}} \leq b_n(\varphi,p)\sqrt{\frac{m-\ka}{N} + \frac{1}{\card{\calI}}}.
\end{equation*}
\end{lemma}
\begin{proof}
By defining
\begin{align*}
A &\defeq \frac{1}{\card{\calI}}\sum_{i=1}^{\card{\calI}}\varphi(\zeta^{J(i)}_n)-f(\varphi)(\hat{\zeta}_{n-1}^{J(i)}),\\
B &\defeq \frac{1}{\card{\calI}}\sum_{i=1}^{\card{\calI}}f(\varphi)(\hat{\zeta}_{n-1}^{J(i)}) - \hat{\pi}_{n-1}(f(\varphi)),
\end{align*}
we have the decomposition
\begin{equation}
\frac{1}{\card{\calI}}\sum_{i=1}^{\card{\calI}} \varphi(\zeta^{J(i)}_n) - \pi_n(\varphi) = A + B. \label{eq:ab decompo}
\end{equation}
With the sequence $X_j:=1/\card{\calI}\sum_{i=1}^j\varphi(\zeta^{J(i)}_n)-f(\varphi)(\hat{\zeta}_{n-1}^{J(i)})$ and $\sigma$-algebras $\mathcal{A}_j:=\sigma(\hat{\zeta}_{n-1},\zeta_n^{J(1)},\ldots,\zeta_n^{J(j)})$, the sequence $(X_j,\mathcal{A}_j)_{j\in[\card{\calI}]}$ is a martingale
%
and by Burkholder's inequality
\begin{equation*}
\E\left[\abs{A}^p\right]^{\frac{1}{p}}
\leq
b_p\osc{\varphi} \sqrt{\frac{1}{\card{\calI}}} \leq b_p\osc{\varphi} \sqrt{\frac{m-\ka}{N}+\frac{1}{\card{\calI}}}.\nonumber
\end{equation*}
For $B$ we have by \eqnref{eq:moment convergence after resample}
\begin{eqnarray*}
\E\left[\abs{B}^p\right]^{\frac{1}{p}}
\leq \hat{b}_{n}(f(\varphi),p)\sqrt{\frac{m-\ka}{N}+\frac{1}{\card{\calI}}},
\end{eqnarray*}
and the claim follows from \eqnref{eq:ab decompo} by Minkowski's inequality.
\end{proof}
%



The proofs of Theorems \ref{thm:radix-r} and \ref{thm:mixed radix-r} are composed of a number of lemmata. We start with the initialization of the particle filter, which is common to both Theorems. Results specific to each of the two butterfly resampling schemes then follow in Sections \ref{sec:proofs for fixed particle filters} and \ref{sec:proofs for mixed particle filters}

\begin{lemma}\label{lem:pf_init}
For all $\varphi \in \boundMeas{\ss}$,
\begin{eqnarray}
\frac{1}{N}\sum_{i}\varphi(\zeta_{0}^{i})-\pi_{0}(\varphi) &\almostsurely{N\to\infty}{\P}&0,\label{lem:lemma 1 slln}\\
\sqrt{N}\left(\frac{1}{N}\sum_{i}\varphi(\zeta_{0}^{i})-\pi_{0}(\varphi)\right)&\indist{N\to\infty}& \normal{0}{\sigma^2_0(\varphi)}.\label{eq:initial clt}
\end{eqnarray}
\end{lemma}
\begin{proof}
Because $\left\{ \zeta_{0}^{i}\right\}_{i\in[N]} \iidsim \pi_0$, the claim follows straightforwardly from the strong law of large numbers and central limit theorem for \iid~random variables.
\end{proof}

\subsection{Particle filter deploying the radix-\texorpdfstring{$r$}{r} algorithm}\label{sec:proofs for fixed particle filters}

For the following three Lemmata, we will assume $r\geq 2$ fixed and that for all $m\geq 1$, $(\zeta^{i}_{n},\hat{\zeta}^{i}_{n})_{n\geq 0, i\in[r^m]}$ are the random variables associated with the augmented resampling particle filter deploying matrices $\radixMatrices{r}{m}$.
\begin{lemma}[Resampling at time $n=0$]
\label{lem:pf_radix-r_res_0}
If for all $\varphi \in \boundMeas{\ss}$,
\begin{align}
\frac{1}{r^m}\sum_{i}\varphi(\zeta_{0}^{i})-\pi_{0}(\varphi) &\almostsurely{m\to\infty}{\P}0,\label{eq:fixed_radix_initial lln_hyp}\\
\sqrt{r^m}\left(\frac{1}{r^m}\sum_{i}\varphi(\zeta_{0}^{i})-\pi_{0}(\varphi)\right)&\indist{m\to\infty} \normal{0}{\frbfVar{0}{\varphi}{r}},\label{eq:fixed_radix_initial clt_hyp}
\end{align}
then for all $\varphi \in \boundMeas{\ss}$,
\begin{align}
\label{eq:fixed_radix_initial lln}
\frac{1}{r^{m}}\sum_{i}\varphi(\hat{\zeta}_{0}^{i})-\hat{\pi}_{0}(\varphi)&\almostsurely{m\to\infty}{\P} 0,\\
\label{eq:fixed_radix_initial_clt}
\sqrt{\frac{r^{m}}{m}}\left(\frac{1}{r^{m}}\sum_{i}\varphi(\hat{\zeta}_{0}^{i})-\hat{\pi}_{0}(\varphi)\right)&\indist{m\to\infty}\normal{0}{\frbfVarHat{0}{\varphi}{r}}.
\end{align}
where $\frbfVarHat{0}{\varphi}{r}$ is as defined in \eqnref{eq:asymptotic variance fixed radix filter 1}.
\end{lemma}
\begin{proof}
With
\begin{equation}\label{eq:def centered phi filter}
\cvarphi_{r^m}(x):=\varphi(x)-\frac{\sum_{i}g_0(\zeta_0^i)\varphi(\zeta_0^i)}{\sum_{i} g_0(\zeta_0^i)},
\end{equation}
and the shorthand notations:
\begin{equation}\label{eq:abc decompo fixed radix}
\arraycolsep+1.4pt
\begin{array}{rcl}
A_m&:=&\dfrac{1}{\pi_0(g_0)}\Bigg(\Bigg(\dfrac{1}{r^m}\displaystyle\sum_{i}g_0(\zeta_0^i)\Bigg)\Bigg(\dfrac{1}{r^m}\displaystyle\sum_{i}\varphi(\hat{\zeta}_0^i)\Bigg)-\dfrac{1}{r^m}\displaystyle\sum_{i}g_0(\zeta_0^i)\varphi(\zeta_0^i)\Bigg),\\[.4cm]
B_m&:=&\dfrac{\sum_{i}g_0(\zeta_0^i)\varphi(\zeta_0^i)}{\sum_{i} g_0(\zeta_0^i)}-\displaystyle\frac{\pi_0(g_0\varphi)}{\pi_0(g_0)},\\[.4cm]
C_m&:=&\dfrac{1}{\pi_0(g_0)}\Bigg(\dfrac{1}{r^m}\displaystyle\sum_{i}\overline\varphi_{r^m}(\hat{\zeta}_0^i)\Bigg)\Bigg(\pi_0(g_0)-\dfrac{1}{r^m}\displaystyle\sum_{i}g_0(\zeta_0^i)\Bigg),
\end{array}
\end{equation}
we have
\begin{equation*}
\frac{1}{r^{m}}\sum_{i}\varphi(\hat{\zeta}_{0}^{i})-\hat{\pi}_{0}(\varphi) = A_m+B_m+C_m,
\end{equation*}
because of the fact that $\hat{\pi}_0(\varphi)=\pi_0(g_0\varphi)/\pi_0(g_0)$.
For the law of large numbers, \eqref{eq:fixed_radix_initial lln}, we shall check that the terms $A_m,B_m,C_m$, each converge to zero as $m\to\infty$, $\P$-almost surely. For $A_m$,  note that the random variables $(\zeta_0^i)_{i\in[r^m]}$ are input to the resampling scheme, and $(\hat{\zeta}_0^i)_{i\in[r^m]}$ are the corresponding output, so the desired convergence follows from the identity \eqnref{it:another martingale representation} in \propref{prop:generalized martingale decomposition} and \theref{thm:butterfly_lln_fixed_radix}. For $B_m$ the desired convergence follows from \eqref{eq:fixed_radix_initial lln_hyp}. For $C_m$, it follows from \theref{thm:butterfly_lln_fixed_radix} and \eqref{eq:fixed_radix_initial lln_hyp} that
\begin{equation}
\frac{1}{r^m}\sum_{i}\overline\varphi_{r^m}(\hat{\zeta}_0^i)\almostsurely{m\rightarrow\infty}{\P} 0,\label{eq:varbar_radix_r_to_zero_initial}
\end{equation}
and the desired convergence then holds since $$|C_m|\leq \pi_0(g_0)^{-1} |r^{-m}\sum_{i}\overline\varphi_{r^m}(\hat{\zeta}_0^i)|2\infnorm{g_0}.$$

For the CLT, \eqref{eq:fixed_radix_initial_clt}, first apply  \eqref{eq:fixed_radix_initial clt_hyp} to establish
\begin{equation*}
\sqrt{r^m}\left(\frac{1}{r^m}\sum_{i}g_0(\zeta_{0}^{i})-\pi_{0}(g_0)\right)\indist{m\to\infty} \normal{0}{\frbfVar{0}{g_0}{r}},
\end{equation*}
and combining this fact with \eqref{eq:varbar_radix_r_to_zero_initial} and Slutsky's theorem, we find that $(r^m/m)^{1/2}C_m$ converges to zero in probability.

Noting that
\begin{equation}
B_m=\frac{\sum_{i}g_0(\zeta_0^i)(\varphi(\zeta_0^i)-\hat{\pi}_0(\varphi)) }{\sum_{i}g_0(\zeta_0^i)},
\label{eq:expression for b}
\end{equation}
we have by \eqref{eq:fixed_radix_initial lln_hyp}, \eqref{eq:fixed_radix_initial clt_hyp} and Slutsky's theorem that $(r^m)^{1/2}B_m$ converges in distribution as $m\to\infty$ to a Gaussian random variable, 
so $(r^m/m)^{1/2}B_m$ converges in probability to zero.

So, by another application of Slutsky's theorem, in order to complete the proof, it suffices to show
\begin{equation}
\sqrt{\frac{r^m}{m}}A_m\indist{m\to\infty}\normal{0}{\frbfVarHat{0}{\varphi}{r}}.\label{eq:fixed_radix_clt_proof_initial}
\end{equation}
By Propositions \ref{prop:fixed radix satisfies assumptions} and \ref{prop:subsample absolute moment convergence}, we can apply \theref{thm:butterfly_clt_fixed_radix} to the test function $\varphi(\,\cdot\,)/\pi_0(g_0)$, yielding
\begin{equation*}
\E\left[\left.\exp(iu(r^m/m)^{1/2}A_m)\right|\zeta_0\right]\inprob{m\rightarrow\infty}\exp(-(u^2/2)\frbfVarHat{0}{\varphi}{r}),
\end{equation*}
and since the modulus of the complex exponential is no greater than $1$, this convergence in fact holds in the $L_1$ sense, and hence, by Levy's continuity theorem, \eqref{eq:fixed_radix_clt_proof_initial} holds.
\end{proof}

\begin{lemma}[Mutation at time $n\geq1$]
\label{lem:pf_radix-r_mut}
Fix $n\geq 1$. If for all $\varphi \in \boundMeas{\ss}$,
\begin{align}
\frac{1}{r^{m}}\sum_{i}\varphi(\hat{\zeta}_{n-1}^{i})-\hat{\pi}_{n-1}(\varphi)&\almostsurely{m\to\infty}{\P} 0,\label{eq:fixed_radix_mut_lln_hyp}\\
\label{eq:fixed_radix_mut_clt_hyp}
\sqrt{\frac{r^{m}}{m}}\left(\frac{1}{r^{m}}\sum_{i}\varphi(\hat{\zeta}_{n-1}^{i})-\hat{\pi}_{n-1}(\varphi)\right)&\indist{m\to\infty}\normal{0}{\frbfVarHat{n-1}{\varphi}{r}},
\end{align}
then for all $\varphi \in \boundMeas{\ss}$,
\begin{align}
\frac{1}{r^{m}}\sum_{i}\varphi(\zeta_{n}^{i})-\pi_{n}(\varphi)&\almostsurely{m\to\infty}{\P} 0,\label{eq:fixed_radix_mut_lln}\\
\label{eq:fixed_radix_mut_clt}
\sqrt{\frac{r^{m}}{m}}\left(\frac{1}{r^{m}}\sum_{i}\varphi(\zeta_{n}^{i})-\pi_{n}(\varphi)\right)&\indist{m\to\infty}\normal{0}{\frbfVar{n}{\varphi}{r}}.
\end{align}
where $\frbfVarHat{n}{\varphi}{r}$ and $\frbfVarHat{n-1}{\varphi}{r}$ are as defined in \eqnref{eq:asymptotic variance fixed radix filter 1}.
\end{lemma}
\begin{proof}
With
\begin{equation*}
A_m:=\frac{1}{r^m}\sum_{i}\varphi(\zeta_n^i)-f(\varphi)(\hat{\zeta}_{n-1}^i),~ B_m:=\frac{1}{r^m}\sum_{i}f(\varphi)(\hat{\zeta}_{n-1}^i)-\hat{\pi}_{n-1}(f(\varphi)),
\end{equation*}
we have
\begin{equation*}
\frac{1}{r^{m}}\sum_{i}\varphi(\zeta_{n}^{i})-\pi_{n}(\varphi)=A_m+B_m.
\end{equation*}
With $X_j:=(r^m)^{-1/2}\sum_{i=1}^j \varphi(\zeta_n^i)-f(\varphi)(\hat{\zeta}_{n-1}^i)$ and $\mathcal{A}_j:=\sigma(\hat{\zeta}_{n-1},\zeta_n^1,\ldots,\zeta_n^j)$, $(X_j,\mathcal{A}_j)_{j\in[r^m]}$ is a martingale, and by application of Burkholder's inequality, Markov's inequality, the fact that $\varphi\in\boundMeas{\ss}$, and Borel-Cantelli, we find that $A_m$ converges to zero as $m\to\infty$, $\P$-almost surely, and $(r^m/m)^{1/2}A_m$ does too. $B_m$ converges to zero almost surely by \eqref{eq:fixed_radix_mut_lln_hyp}, and $(r^m/m)^{1/2}B_m$ converges to a $\normal{0}{\frbfVarHat{n-1}{f(\varphi)}{r}}$ by \eqref{eq:fixed_radix_mut_clt_hyp}.
\end{proof}

\begin{lemma}[Resampling at time $n\geq1$]
\label{lem:pf_radix-r_res}
Fix $n\geq 1$. If for all $\varphi \in \boundMeas{\ss}$,
\begin{align}
\label{eq:fixed_radix_res_lln_hyp}
\frac{1}{r^m}\sum_{i}\varphi(\zeta_{n}^{i})-\pi_{n}(\varphi) &\almostsurely{m\to\infty}{\P} 0,\\
\label{eq:fixed_radix_res_clt_hyp}
\sqrt{\frac{r^m}{m}}\left(\frac{1}{r^m}\sum_{i}\varphi(\zeta_{n}^{i})-\pi_{n}(\varphi)\right)&\indist{m\to\infty} \normal{0}{\frbfVar{n}{\varphi}{r}},
\end{align}
then for all $\varphi \in \boundMeas{\ss}$,
\begin{align}
\frac{1}{r^{m}}\sum_{i}\varphi(\hat{\zeta}_{n}^{i})-\hat{\pi}_{n}(\varphi)&\almostsurely{m\to\infty}{\P} 0,\label{eq:fixed_radix_res_lln}\\
\label{eq:fixed_radix_res_clt}
\sqrt{\frac{r^{m}}{m}}\left(\frac{1}{r^{m}}\sum_{i}\varphi(\hat{\zeta}_{n}^{i})-\hat{\pi}_{n}(\varphi)\right)&\indist{m\to\infty} \normal{0}{\frbfVarHat{n}{\varphi}{r}}.
\end{align}
where $\frbfVar{n}{\varphi}{r}$ and $\frbfVarHat{n}{\varphi}{r}$ are as defined in \eqnref{eq:asymptotic variance fixed radix filter 1}.
\end{lemma}
\begin{proof}
By defining $\cvarphi_{r^m}$, $A_m$, $B_m$ and $C_m$ as in \eqnref{eq:def centered phi filter} and \eqnref{eq:abc decompo fixed radix} but by replacing $0$ with $n$
we have
\begin{equation*}
\frac{1}{r^{m}}\sum_{i}\varphi(\hat{\zeta}_{n}^{i})-\hat{\pi}_{n}(\varphi) = A_m+B_m+C_m.
\end{equation*}
For the law of large numbers, \eqref{eq:fixed_radix_res_lln}, very similar arguments to those in the proof of \lemmaref{lem:pf_radix-r_res_0} establish that $A_m,B_m,C_m$, each converge to zero as $m\to\infty$, $\P$-almost surely.

The proof of the CLT \eqref{eq:fixed_radix_res_clt}, also uses arguments similar to those in the proof of \lemmaref{lem:pf_radix-r_res_0}, the main difference being that due to the statistically different nature of the input $(\zeta^i_n)_{i\in[r^m]}$, the term $B_m$ does not vanish. From \eqref{eq:fixed_radix_res_lln_hyp}, \eqref{eq:fixed_radix_res_clt_hyp} and \eqref{eq:fixed_radix_res_lln} it follows that $(r^m/m)^{1/2}C_m$ converges to zero in probability. So in order to complete the proof, it suffices to show
\begin{equation}
\sqrt{\frac{r^m}{m}}A_m+\sqrt{\frac{r^m}{m}}B_m\indist{m\to\infty}\normal{0}{\frbfVarHat{n}{\varphi}{r}}.\label{eq:fixed_radix_clt_proof}
\end{equation}
By Propositions \ref{prop:fixed radix satisfies assumptions} and \ref{prop:subsample absolute moment convergence}, we can apply \theref{thm:butterfly_clt_fixed_radix} to the test function $\varphi(\,\cdot\,)/\pi_n(g_n)$,
\begin{equation*}
\E\left[\left.\exp(iu(r^m/m)^{1/2}A_m)\right|\zeta_n\right]\inprob{m\rightarrow\infty}\exp(-(u^2/2)\sigma^2),
\end{equation*}
where $\sigma^2=(1-r^{-1}) \hat{\pi}_n((\varphi-\hat{\pi}_n(\varphi))^2$.

For $B_m$ we have an expression analogous to \eqnref{eq:expression for b} from which we see by
by \eqref{eq:fixed_radix_res_lln_hyp}, \eqref{eq:fixed_radix_res_clt_hyp} and Slutsky's theorem that $(r^m/m)^{1/2}B_m$ converges in distribution
as $m\to\infty$ to a Gaussian random variable, call it $Z$, with mean zero and variance $\pi_n(g_n)^{-2}\frbfVar{n}{g_n(\varphi - \hat{\pi}_n(\varphi))}{r}$. Then by the continuous mapping theorem, $\exp(iu(r^m/m)^{1/2}B_m)$ converges
in distribution to $\exp(iuZ)$, and by yet another application of Slutsky's theorem,
\begin{align*}
&\E\left[\left.\exp(iu(r^m/m)^{1/2}A_m)\right|\zeta_n\right]\exp(iu(r^m/m)^{1/2}B_m)\\
&\quad\inprob{m\rightarrow\infty}\exp(-(u^2/2)\sigma^2)\exp(iuZ),
\end{align*}
from which \eqref{eq:fixed_radix_clt_proof} follows.
\end{proof}

From the Lemmata \ref{lem:pf_radix-r_res_0}, \ref{lem:pf_radix-r_mut} and \ref{lem:pf_radix-r_res}, together with \lemmaref{lem:pf_init}, \theref{thm:radix-r} follows.

\subsection{Particle filter deploying the mixed radix-\texorpdfstring{$r$}{r} algorithm}
\label{sec:proofs for mixed particle filters}

For the following two Lemmata, we will assume $r\geq 2$ fixed and that for all $c\geq 1$, $(\zeta^{i}_{n},\hat{\zeta}^{i}_{n})_{n\geq 0, i\in[rc]}$ are the random variables associated with the augmented resampling particle filter deploying matrices $\mradixMatrices{r}{c}$.

\begin{lemma}[Resampling $n\geq 0$]
\label{lem:pf_mixed-radix-r_res_0}
Fix $n\geq 0$. If for all $\varphi \in \boundMeas{\ss}$,
\begin{align}
\label{eq:mixed_radix_initial lln_hyp}
\frac{1}{rc}\sum_{i}\varphi(\zeta_{n}^{i})-\pi_{n}(\varphi) &\almostsurely{c\to\infty}{\P}0,\\
\label{eq:mixed_radix_initial clt_hyp}
\sqrt{rc}\left(\frac{1}{rc}\sum_{i}\varphi(\zeta_{n}^{i})-\pi_{n}(\varphi)\right)&\indist{c\to\infty}\normal{0}{\mrbfVar{n}{\varphi}{r}},
\end{align}
then for all $\varphi \in \boundMeas{\ss}$,
\begin{align}
\label{eq:mixed_radix_initial lln}
\frac{1}{rc}\sum_{i}\varphi(\hat{\zeta}_{n}^{i})-\hat{\pi}_{n}(\varphi)&\almostsurely{c\to\infty}{\P} 0,\\
\label{eq:mixed_radix_initial_clt} \sqrt{rc}\left(\frac{1}{rc}\sum_{i}\varphi(\hat{\zeta}_{n}^{i})-\hat{\pi}_{n}(\varphi)\right)&\indist{c\to\infty}\normal{0}{\mrbfVarHat{n}{\varphi}{r}}.
\end{align}
where $\mrbfVar{n}{\varphi}{r}$ and $\mrbfVarHat{n}{\varphi}{r}$ are as defined in \eqnref{eq:asymptotic variance mixed radix filter}.
\end{lemma}
\begin{proof}[Proof of \lemmaref{lem:pf_mixed-radix-r_res_0}]
By defining $\cvarphi_{rc}$, $A_c$, $B_c$ and $C_c$ as in \eqnref{eq:def centered phi filter} and \eqnref{eq:abc decompo fixed radix} but by replacing $0$ with $n$ and $r^m$ with $rc$
we have
\begin{equation*}
\frac{1}{rc}\sum_{i}\varphi(\hat{\zeta}_{n}^{i})-\hat{\pi}_{n}(\varphi) = A_c+B_c+C_c.
\end{equation*}
The law of large numbers follows from \eqnref{it:another martingale representation} of \propref{prop:generalized martingale decomposition}, \theref{thm:butterfly_lln_mixed_radix} and \eqref{eq:mixed_radix_initial lln_hyp} analogously to the proof of \lemmaref{lem:pf_radix-r_res_0} so the details are omitted.

To prove \eqnref{eq:mixed_radix_initial_clt} it suffices to show that
\begin{equation}
\sqrt{rc}A_c+\sqrt{rc}B_c + \sqrt{rc}C_c\indist{c\to\infty}\normal{0}{\mrbfVarHat{n}{\varphi}{r}}.\label{eq:mixed_radix_clt_proof}
\end{equation}
For $\sqrt{rc}B_c$ and $\sqrt{rc}C_c$ we proceed similar to the proofs of \lemmaref{lem:pf_radix-r_res_0} and \lemmaref{lem:pf_radix-r_res}. For $\sqrt{rc}A_c$ we observe that by \propref{prop:mixed radix satisfies assumptions} and \propref{prop:subsample absolute moment convergence} we can apply \theref{thm:butterfly_clt_mixed_radix} to the test function $\sqrt{2}\varphi(\,\cdot\,)/\pi_n(g_n)$, yielding by \eqnref{it:another martingale representation} of \propref{prop:generalized martingale decomposition}
\begin{equation*}
\E\left[\left.\exp(iu\sqrt{rc}A_c)\right|\zeta_n\right]\inprob{c\to\infty}\exp(-(u^2/2)\sigma^2),
\end{equation*}
where $\sigma^2=\left(2-{r^{-1}}\right)\hat{\pi}_n\left((\varphi - \hat{\pi}_n(\varphi))^2\right)$. We then proceed analogously to the proof of \lemmaref{lem:pf_radix-r_res} to establish \eqnref{eq:mixed_radix_clt_proof} completing the proof.
\end{proof}

\begin{lemma}[Mutation at time $n\geq1$]
\label{lem:pf_mixed_radix-r_mut}
Fix $n\geq 1$. If for all $\varphi \in \boundMeas{\ss}$,
\begin{align}
\frac{1}{rc}\sum_{i}\varphi(\hat{\zeta}_{n-1}^{i})-\hat{\pi}_{n-1}(\varphi)&\almostsurely{c\to\infty}{\P} 0,\label{eq:mixed_radix_mut_lln_hyp}\\
\label{eq:mixed_radix_mut_clt_hyp} \sqrt{rc}\left(\frac{1}{rc}\sum_{i}\varphi(\hat{\zeta}_{n-1}^{i})-\hat{\pi}_{n-1}(\varphi)\right)&\indist{c\to\infty}\normal{0}{\mrbfVarHat{n-1}{\varphi}{r}},
\end{align}
then for all $\varphi \in \boundMeas{\ss}$,
\begin{align}
\frac{1}{rc}\sum_{i}\varphi(\zeta_{n}^{i})-\pi_{n}(\varphi)&\almostsurely{c\to\infty}{\P} 0,\label{eq:mixed_radix_mut_lln}\\
\label{eq:mixed_radix_mut_clt} \sqrt{rc}\left(\frac{1}{rc}\sum_{i}\varphi(\zeta_{n}^{i})-\pi_{n}(\varphi)\right)&\indist{c\to\infty}\normal{0}{\mrbfVar{n}{\varphi}{r}}.
\end{align}
where $\mrbfVarHat{n-1}{\varphi}{r}$ and $\mrbfVar{n}{\varphi}{r}$ are as defined in \eqnref{eq:asymptotic variance mixed radix filter}.
\end{lemma}
\begin{proof}[Proof of \lemmaref{lem:pf_mixed_radix-r_mut}]
The proof of \eqnref{eq:mixed_radix_mut_lln} is analogous to that in the proof of \lemmaref{lem:pf_radix-r_mut}, and \eqnref{eq:mixed_radix_mut_clt} follows from same arguments as \cite[Lemma A.1]{smc:the:C04}.
\end{proof}

From Lemmata \ref{lem:pf_mixed-radix-r_res_0} and \ref{lem:pf_mixed_radix-r_mut}, together with \lemmaref{lem:pf_init}, \theref{thm:mixed radix-r} follows.

\bibliography{my}
\bibliographystyle{imsart-nameyear}

\end{document}